\documentclass[a4paper,superscriptaddress,12pt]{article}
\usepackage{graphicx}
\usepackage{amsfonts}
\usepackage{geometry}
\usepackage{comment}
\usepackage{theorem}
\usepackage{amsmath}
\usepackage{authblk}

\usepackage{amssymb}%
\newtheorem{theorem}{Theorem}[section]
\newtheorem{lemma}[theorem]{Lemma}

\newenvironment{proof}[1][Proof]{\begin{trivlist}
\item[\hskip \labelsep {\bfseries #1}]}{\end{trivlist}}

\newcommand{\qed}{\nobreak \ifvmode \relax \else
      \ifdim\lastskip<1.5em \hskip-\lastskip
      \hskip1.5em plus0em minus0.5em \fi \nobreak
      \vrule height0.75em width0.5em depth0.25em\fi}
\providecommand{\U}[1]{\protect\rule{.2in}{.2in}}
\numberwithin{equation}{section}
\begin{document}
\title{Hamiltonian constraint in Euclidean LQG revisited: First hints of off-shell Closure}
\author[a]{Alok Laddha}
\affil[a]{Chennai Mathematical Institute\\Siruseri}

\maketitle

\begin{abstract}
We initiate the hunt for a definition of Hamiltonian constraint in Euclidean Loop Quantum Gravity (LQG) which faithfully represents quantum Dirac algebra. Borrowing key ideas from \cite{qsd1}, \cite{tv},\cite{brugmann} and \cite{hlt1}, we present some evidence that there exists such a continuum Hamiltonian constraint operator which is well defined on a suitable generalization of the Lewandowski-Marolf Habitat and is anomaly free off-shell.
\end{abstract}
\newpage

\tableofcontents{}

\section{Introduction}
A quantum Hamiltonian constraint in Loop Quantum Gravity (LQG) was proposed by Thiemann in a series of remarkable papers [\cite{qsd1}, \cite{qsd2}, \cite{qsd3}].  Owing to the scalar (as opposed to scalar density) nature of the lapse, and diffeomorphism covariant regularization scheme, this operator admitted a well defined continuum limit on the kinematical Hilbert space of LQG. The continuum limit in turn had a remarkable property of diffeomorphism covariance and satisfied a consistency  condition known in the LQG literature as on-shell closure; which briefly implies that the Hamiltonian constraint $\hat{H}[N]$ satisfied 
\begin{equation}
\begin{array}{lll}
\hat{V}(g)\hat{H}[N]\hat{V}(g)^{-1}\ =\ \hat{H}[N]\\
\hat{U}(\phi)\hat{H}[N]\hat{U}(\phi)^{-1}\ =\ \hat{H}[\phi^{\star}\cdot N]\\
\Psi[\hat{H}[N], \hat{H}[M]]\vert{\bf \psi}\rangle\ =\ 0
\end{array}
\end{equation}
where $\hat{V}(g),\hat{U}(\phi)$ are unitaries associated to $SU(2)$ gauge transformations and spatial diffeomorphisms respectively. $\Psi$ is any diffeomorphism invariant state of the theory, $\vert{\bf \psi}\rangle\ \in\ {\cal H}_{kin}$ and hence the last equation implies that on the diffeomorphism invariant states commutator of two continuum constraints vanishes. This is consistent with a plausible quantization of the right hand side of $\{H[N], H[M]\}$ which is proportional to diffeomorphism constraint and could be expected to trivialize on diffeomorphism invariant states.\\ 
The most severe criticism to Thiemann's proposal came soon after in [\cite{lm}, \cite{lm2}] where it was essentially argued that although consistent \emph{on shell}, the operator hides an anomaly owing to its specify density weight. 
A related issue which has plagued this Hamiltonian constraint is the infinite multitude of choices which feed into its definition, each of which can potentially give rise to a different physical theory. It would be nice to have a physical criterian which can reduce this enormous ambiguity in the definition.\\
These arguments have been revisited recently. First in certain two dimensional generally covariant theories [\cite{alppftham}, \cite{ttppftham}], then in Hussain Kuchar model \cite{aldiff}, and finally in certain three and four dimensional theories whose gauge algebra is the Dirac algebra of canonical gravity.[\cite{hlt1}, \cite{tv}]\\
Through the analysis of these toy models, a clear picture emerges where few key facts tie themselves together. \noindent {(1)} The higher density weight constraints whose quantization leads to rather singular quantum operators (see below) such that the commutator does not trivialize, \noindent {(2)} The choice of quantum operator based on the understanding of its Hamiltonian vector fields in classical theory, and \noindent {(3)} Anomaly free representation of Dirac algebra in Quantum theory. Moreover for the models where the spectrum of the theory is known via other more transparent methods, constraints so obtained lead to kernels which match with the spactrum of the theory.\footnote{The reader is also encouraged to read introduction sections in (\cite{hlt1},\cite{tv}) to understand the underlying  philosophy behind the current pursuit of Hamiltonian constraint in LQG.}\\
In this paper, we apply these lessons to four dimensional Euclidean Hamiltonian constraint. We motivate the choices made in the quantization of the constraint through it`s action in classical Phase space and show that \noindent{\bf (a)} A continuum limit $\hat{H}[N]$ exists, \noindent{\bf (b)} The continuum limit of two finite triangulation constraints satisfy the elusive off-shell closure criteria. We believe that this provides some persuasive first hints that four-dimensional Euclidean Loop Quantum Gravity is anomaly free \emph{even off-shell}  \cite{nicolaietal}.\\
We heavily rely upon three key works available in the literature.\\
\noindent{\bf (i)} Thiemann's seminal work on Hamiltonian constraint in \cite{qsd1} which due to its remarkable properties like diffeomorphism covariance, still remains state of the art as far as $3+1$-dimensional canonical LQG is concerned.\\
\noindent{\bf (ii)} Tomlin, Varadarajan's quantization of  Hamiltonian constraint in $U(1)^{3}$ theory in $3+1$ dimensions in (\cite{tv}) and,\\
\noindent{\bf (iii)} Brugmann's derivation of Quantum Constraint algebra in the Loop representation \cite{brugmann}.\\
Indeed our work here is an attempt at taking key lessons out of these works and applying them to Euclidean gravity.\\
We will proceed as follows. After explaining a basic assumption underlying the current work, we will briefly outline the content of each section. As some of the later sections might appear a bit involved in the first reading, we explain the basic ideas behind our constructions in \ref{outline}. We then proceed to the main body of the paper.
\subsection{Quantizing the constraint of density weight two.} 
We are interested in a continuum Hamiltonian constraint which {\bf (a)} generates an anomaly-free quantum Dirac algebra and {\bf (b)} is well defined on habitat spaces of the type constructed in \cite{lm} or appropriate generalization thereof. 
Such Habitats are obtained by looking at specific (infinite) linear combination of spin-networks with coefficients in the linear combination being smooth functions on $\Sigma$. As argued by the authors in \cite{lm},\cite{lm2}; a finite triangulation Hamiltonian constraint will admit continuum limit which satisfies the above two conditions if it carries an over-all explicit factor of the ``triangulation scale" $\delta$ in the denominator.  This can be achieved by working with a Hamiltonian constraint of density weight $\frac{4}{3}$.\\
That is the constraint that we would ideally like to work with is, $C[N]\ =\ \frac{\epsilon^{ink}F_{ab}^{i}E^{a}_{j}E^{b}_{k}}{q^{\frac{1}{6}}}$. This constraint is such that a finite triangulation constraint operator $\hat{C}_{\delta}[N]\ =\ \frac{1}{\delta}\sum_{\triangle\ \in\ T}C_{\triangle}$. Quantization of this constraint such that the quantization of $F\cdot E\cdot E$ is motivated by certain classical considerations involves a \emph{separate} quantization of  ${q^{-\frac{1}{6}}}$. Classically, a slight variant of the Thiemann trick can be used to rewrite $q^{-\frac{1}{6}}$ as \cite{qsd5},\cite{tv}
\begin{equation}
\begin{array}{lll}
q^{-\frac{1}{6}}(x)\ =\ {\cal B}\tilde{\eta}^{abc}\epsilon_{ijk}\{A_{a}^{i}(x), V^{\frac{1}{3}}\}\cdot\{A_{b}^{j}(x),V^{-\frac{1}{3}}\}\cdot\{A_{c}^{k}(x), V^{-\frac{1}{3}}\}
\end{array}
\end{equation}
where all the numerical factors have been absorbed in ${\cal B}$.\\
At finite regularization, this identity ensures that $q^{-\frac{1}{6}}(x)$ can be quantized as a densely defined operator on ${\cal H}_{kin}$. However despite the stunning progress in the understanding of volume operator (\cite{brunemann},\cite{brunemann-thiemann},\cite{brunemann-rideout}) its action on a generic spin-network state can not be written down analytically. Hence asking detailed questions regarding the off-shell closure of constraint algebra generated by such density $\frac{4}{3}$ constraint is a formidable task.\\
We do not aim to solve this problem in the current paper, but seek a more modest  goal. Previous Experience regarding the off-shell closure of Dirac algebra in toy models suggest that, quantization of Hamiltonian constraint which yields such a closure involves a quantization of $F\cdot E\cdot E$ separately from quantization of $\frac{1}{q^{\alpha}}$. That is, schematically we quantize the constraint as,
\begin{equation}
\begin{array}{lll}
\hat{H}_{\delta}(v)\ \approx\ \widehat{F\cdot E\cdot E}\widehat{\frac{1}{q^{\alpha}}}
\end{array}
\end{equation}
Based on the work in such toy models, the central issue in the quantization of Hamiltonian constraint is regarding the quantization choices that one makes in defining $\widehat{F\cdot E\cdot E}$. The choices involved in the definition of $\frac{1}{q^{\frac{1}{\alpha}}}$ is then done so as to ensure off-shell closure,\cite{tv},\cite{hlt1}. As such exploring the quantization ambiguities in defining $\widehat{F\cdot E\cdot E}$ takes a precedent (conceptually as well as technically) over those involved in $\frac{1}{q^{\frac{1}{\alpha}}}$. In any case the upshot is that one quantizes $\widehat{F\cdot E\cdot E}$ and $\widehat{\frac{1}{q^{\alpha}}}$ separately. It is the quantization of the former that will be the focus of current work.\\
That is we will work with the Hamiltonian constraint which is polynomial in phase space variables and is a scalar density of weight two. At finite triangulation with triangulation scale $\delta$, the associated function will carry an over-all power of $\frac{1}{\delta^{3}}$. That is, schematically the finite triangulation operator would be of the form
\begin{equation}
\begin{array}{lll}
\hat{H}_{T(\delta)}[N]\ \approx\ \frac{1}{\delta^{3}}\sum_{\triangle\ \in\ T}\hat{H}_{\triangle}(v)
\end{array}
\end{equation}
This constraint has ``too many factors" of $\delta$ in the denominator and 
on the type of Habitat considered so far in the literature , this operator would not yield a continuum limit. Due to this reason, in this paper we work with the re-scaled constraint $\delta^{2}\hat{H}_{T(\delta)}[N]$. \emph{Our main goal in this paper is to show that there exists a quantum operator $\hat{H}_{T(\delta)}[N]$ and a linear space of distributions (which we will always refer to as a Habitat) such that $\lim_{\delta\rightarrow 0}\Psi\left(\delta^{2}\hat{H}_{T(\delta)}[N]\right)\vert{\bf s}\rangle$ is well defined $\forall\ \Psi\ \in\ \textrm{Habitat}, \forall\ {\bf s}$ and},
\begin{equation}\label{finaleqn}
\begin{array}{lll}
\lim_{\delta\rightarrow 0}\lim_{\delta^{\prime}\rightarrow 0}\Psi\left(\left[(\delta^{\prime})^{2}\hat{H}_{T(\delta^{\prime})}[M]\delta^{2}\hat{H}_{T(\delta)}[N]\ -\ N\leftrightarrow M\right]\right)\vert{\bf s}\rangle\ =\\
\vspace*{0.1in}
\hspace*{1.2in}\lim_{\delta\rightarrow 0}\lim_{\delta^{\prime}\rightarrow 0}\Psi\left(\widehat{q^{ab}\left(N\nabla_{a}M\ -\ M\leftrightarrow N\right)H_{b}}(\delta,\delta^{\prime})\right)\vert{\bf s}\rangle
\end{array}
\end{equation}
for a specific quantization of $q^{ab}\left(N\nabla_{a}M\ -\ M\leftrightarrow N\right)H_{b}$. Although we work with (re-scaled) constraints of density-weight two, we will  always remember our ultimate goal, which is to quantize the density $\frac{4}{3}$ constraint. Whence along the way we will make several choices which we believe will be robust under the re-scaling by  fractional power of inverse determinant. 

\subsection{Outline of the work}\label{outline}
We now provide a brief qualitative outline for the later sections. This we believe will help the reader to get a feel for some of the underlying ideas.\\
As the constraints of canonical gravity are composite fields built out of $SU(2)$ connection and the conjugate (sensitized) triad, like any composite operator in a local quantum field theory, there is an infinitude of choices which one can make to quantize such objects. In ordinary (flat-space) QFT one uses physical criteria like Causality and Locality to define composite operators in the theory. We would similarly  like to have certain set of  physical criteria which guide the choices we make to define Composite operators (specifically the constraints as far as we are concerned) in LQG. One strategy which has been particularly successful in recent years has been to quantize the constraints such that the choices we make, to define such operators at finite regularization (triangulation) be clearly associated with finite geometric transformations generated by their classical counter-parts. This strategy can be precisely employed for Gauss constraint and Diffeomorphism constraint. that is, e.g. in the case of diffeomorphism constraint, where we need to choose certain quantization prescription for $F_{ab}^{i}E^{b}_{i}$ in terms of holonomies and fluxes, we demand that the resulting operator (which is regularized
by the choice of loops and surfaces) generate spatial diffeomorphisms on the states. This strategy leads to a quantization of diffeomorphism constraint in LQG which is anomaly free (in the sense that it represents the Lie algebra of spatial vector fields faithfully) and whose kernel is the space of diffeomorphism invariant states labelled by (generalized) knot classes.\cite{aldiff}.\\
However in order to employ this strategy for the Hamiltonian constraint, we need to have a precise geometric interpretation for the symplectomorphisms generated by $H[N]$ on the phase space of canonical gravity. As was shown in \cite{hlt1},\cite{tv} just such an interpretation exists as far as action of $H[N]$ on $A_{a}^{i}$ is concerned. For the Hamiltonian constraint of density weight two, $H[N]\ =\ \int N(x)\textrm{Tr}(F\cdot E\cdot E)$, 
\begin{equation}
\{H[N],A_{a}^{i}(x)\}\ =\ \epsilon^{ijk}\left[{\cal L}_{NE_{j}}A_{a}^{k}(x)\ +\ {\cal D}_{a}(E^{b}_{j}A_{b}^{k})\right]
\end{equation}
That is, action of $H[N]$ on $A(x)$ can be understood in terms of certain dynamical spatial diffeomorphisms and dynamical gauge transformations. Indeed this was the starting point for quantization of Hamiltonian constraint in $U(1)^{3}$ theory in [\cite{hlt1},\cite{tv}]. This picture was completed by Ashtekar in \cite{abhaycalc}. He has shown that  even on  conjugate triad fields the Hamiltonian constraint action can be understood in terms of the above, very specific kind of dynamical  diffeomorphism and gauge transformation. This interpretation of the Hamiltonian constraint is a starting point for quantization performed in this paper.\\ 
Fixing a basis in $su(2)$, the dynamical spatial diffeomorphisms mentioned above are generated by the triplet of vector fields $NE_{i}$. In section (\ref{qshift}), we quantize this vector field and interpret the action of the resulting operator on any spin-network state in terms of tangent vectors at the vertices of the state, and certain insertion of Pauli-matrices $\tau_{i}$ along the edges of the state. In accordance with the nomenclature adopted in \cite{hlt1}, we refer to $\hat{NE_{i}}$ for each $i$ as a quantum shift.\\
In section (\ref{finitetrianglecon}), we quantize $H[N]$ to obtain $\hat{H}_{T(\delta)}[N]$  at (one parameter family of) finite triangulations $T(\delta)$\footnote{$\delta$ denotes fine-ness of the triangulation, and we will always refer to it as triangulation scale. In $\delta\rightarrow 0$ limit, the triangulation $T\rightarrow\ \Sigma$.} 
which satisfies two criteria.\\
(1) \emph{The quantum constraint action on a spin-network state reflects, in a precise way the action of classical Hamiltonian constraint on (classical) cylindrical functions.}\\
(2)  \emph{The new vertices created  by $\hat{H}_{T(\delta)}[N]$ are non-degenrate.}\footnote{Throughout this paper, we will refer to vertices which are in the kernel of volume operator as degenerate vertices.}

In section (\ref{deformedstates}) we analyze the states that lie in the image of $\hat{H}_{T(\delta)}[N]$. That is, we analyze how a given spin-net state $\vert{\bf s}\rangle$  gets deformed under the action of the quantum constraint. The resulting states which are some very specific linear combination of spin-net states are referred to as cylindrical-networks. We analyze structure of these states in connection representation and establish several orthogonality relations for these states, which make them suitable for the analysis of the continuum limit, without further decomposing them into spin-network states.\\
In section \ref{con-limit} we  construct a subspace of the space of distributions ${\cal D}^{\star}$, ${\cal V}_{hab}$ on which continuum limit of finite triangulation constraint exists as a linear map ${\cal V}_{hab}\rightarrow\ {\cal D}^{\star}$ and off-shell closure of constraint algebra can be probed. We will refer to this space simply as  Habitat as it is a natural(albeit a weaker) generalization of the Habitat defined by Lewandwoski and Marolf to explore the continuum limit of regularized diffeomorphism constraint \cite{lm}. In the case of Lewandowski and Marolf's Habitat, each state was labelled by a class of spin-networks $[{\bf s}_{1}]\ =\ ({\bf s}_{1},{\bf s}_{2},\dots)$ and a smooth function $f$ on the set of vertices associated to $[{\bf s}_{1}]$.As detailed in sections (\ref{habitat},\ref{naivcont}) in our case, each habitat state will be labelled by an appropriate set of cylindrical networks $[{\bf s}]$, and 
 not only by a smooth function $f$ but a function ${\cal P}$ which will probe the intertwiner structure at the vertices of the cylindrical states.\\
In section (\ref{conlimitofLHS})  we show that on such a Habitat, the continuum limit of $\hat{H}_{T(\delta)}[N]$ as well as finite triangulation commutator is well defined.\\
In sections (\ref{conlimitofLHS}) and (\ref{RHS-section}) we quantize the Right hand side of the story by using the remarkable identity from \cite{tv} expressing $q^{ab}H_{b}(N\nabla_{a} M\ -\ M\leftrightarrow N)$ as a Poisson bracket of two (triad dependent) diffeomorphisms. Once again, we show that the continuum limit of the resulting finite triangulation operator is well defined on the Habitat and it matches with the continuum limit LHS.\\
In the final section we  discuss the shortcomings of our work and conclude with the resulting open problems.\\
We set $\hbar\ =\ G\ =\ c\ =\ 1$.\\

\section{Geometric interpretation of Hamiltonian constraint action}\label{geometricinterofcon}
In the spirit of previous quantization of constraints (e.g. Diffeomorphism constraints, Parametrized field theories, $U(1)^{3}$ gauge theories), in LQG or loop quantized generally covariant field theories, we recognize the importance of understanding the geometric interpretation of Hamiltonian constraint action on phase space data. We propose just such an interpretation for the constraint of $SU(2)$ theory in this section. We propose a number of results which are inter-related to each other and try to paint a coherent picture of the underlying geometric interpretation. Based on this interpretation we will make quantization choices.\\
We first recall the results in $U(1)^{3}$ theory as obtained in [\cite{hlt1}, \cite{tv}].
\begin{equation}
\begin{array}{lll}
H[N]\ =\ \frac{1}{2}\int_{\Sigma}\ d^{3}x \epsilon^{ijk}F_{ab}^{i}E^{a}_{j}E^{b}_{k}\\
\vspace*{0.1in}
\{H[N], A_{a}^{i}(x)\}\ =\ \epsilon^{ijk}{\cal L}_{NE_{j}}A_{a}^{k}(x)
\end{array}
\end{equation}
We now seek a similar interpretation for the Euclidean gravity. Briefly, as we show below, the action of Hamiltonian constraint on the phase space variables can, in this case be cast in the following form
\begin{equation}\label{key-eq}
\begin{array}{lll}
\{H[N], A_{a}(x)\}\ = [\tau_{k}, {\cal L}_{N\tilde{E}_{k}}A_{a}(x)]\ -\ [\tau_{k}, {\cal R}_{\Lambda_{(k)}}A_{a}(x)]\\
\vspace*{0.1in}
\{H[N], \tilde{E}^{a}(x)\}\ =\ [\tau_{k}, {\cal L}_{N\tilde{E}_{k}}\tilde{E}^{a}(x)] \ +\ [\tau_{k}, {\cal R}_{\Lambda^{'}_{(k)}}\tilde{E}^{a}(x)]
\end{array}
\end{equation}
In the above equations, we understand $NE_{k}$ as being a triplet of vector fields labelled by $k$, and $\Lambda_{k}, \Lambda^{\prime}_{k}$ as being triple of Gauge generators, which will be defined below.\\
The second equation in \ref{key-eq} was derived by Ashtekar in \cite{abhaycalc}.\footnote{We are indebted to A.Ashtekar for explaining us his result prior to publication.} In fact he interpreted the result in terms of so called Gauge-covariant Lie derivative. We reproduce his derivation here, but with a minor spin-off on the final result. Instead of expressing the action of Hamitonian constraint in terms of gauge-covariant Lie derivative\cite{jackiw}, we will express it in terms of ordinary Lie derivative. This of course has a direct utility in quantum theory as ordinary Lie derivatives can be approximated by one parameter family of spatial diffeomorphisms.\footnote{Anologously, the gauge covariant Lie derivatives can be approximated by so called Matrix-valued diffeomorphisms\cite{jandeboer}. However to the best of author's understanding, such structures are not sufficiently understood, hence we refrain from exploring related ideas here.}
\\ 
We now turn to the detailed derivation of the eqs.(\ref{key-eq})\\
\subsection{Action of $H[N]$ on $A(x)$}
Let $A_{b}(x)\ =\ A_{b}^{i}(x)\tau_{i}$, where $\tau_{i}\ =\ -i\sigma_{i}$, with $\sigma_{i}$ being a Pauli-matrix.
\begin{equation}\label{eqjuly7(1)}
\begin{array}{lll}
\{H[N], A_{b}(x)\}\ =\ -2N(x)\ \epsilon^{ijk}\ F_{ab}(x)\tilde{E}^{a}_{j}(x)\tau_{k}\\
\vspace*{0.1in}
=\ 2\ [\tau_{j}, \tau_{i}\ F_{ab}^{i}(x)\ N(x)\tilde{E}^{a}_{j}(x)]
\end{array}
\end{equation}
Our motivation for keeping $\tau_{j}$ separate from $E_{j}$ comes from our experience with the $U(1)^{3}$ theory where the three vector fields, $E_{j}$ 
defined a triple of so-called quantum shifts which turned out to be the generator of (singular) diffeomorphisms in the quantum theory.\\
We now note that
\begin{equation}\label{eqjuly7(2)}
\begin{array}{lll}
{\cal L}_{N\tilde{E}_{j}}A_{b}^{i}(x)\ =\ N\tilde{E}^{a}_{j}(x)F_{ab}^{i}(x) + {\cal D}_{b}\left(\tilde{E}^{a}_{j}A_{a}^{i}(x)\right)\\
\vspace*{0.1in}
\implies\ N\tilde{E}^{a}_{j}(x)F_{ab}^{i}(x)\ =\ {\cal L}_{N\tilde{E}_{j}}A_{b}^{i}(x) - {\cal D}_{b}\Lambda_{(j)}^{i}(x)
\end{array}
\end{equation}
Substituting \ref{eqjuly7(2)} in \ref{eqjuly7(1)} we have
\begin{equation}\label{HonA-final}
\begin{array}{lll}
\{H[N], A_{b}(x)\}
 =\ 2[\tau_{j}, {\cal L}_{N\tilde{E}_{j}}A_{b}(x)]\ -\ 2[\tau_{j}, {\cal D}_{b}\Lambda_{(j)}(x)]
 \end{array}
 \end{equation}
 where the triplet of Gauge parameters are defined as, $\Lambda_{(j)}\ :=\ \tilde{E}^{a}_{j}A_{a}^{i}\tau_{i}$.\\ 
 We thus see that Action of Hamiltonian constraint on a connection $A$ is a concomitant of diffeomorphism and gauge transformation generated by the triples $\tilde{E}_{j}$ and $\Lambda_{(j)}$ respectively, supplemented by an adjoint action of $\vec{\tau}$.
\subsection{Action of $H[N]$ on $\tilde{E}(x)$} 
For the benefit of the reader, we now reproduce the computation done in \cite{abhaycalc} which shows that  {\emph precisely} the same interpretation can be given to the action of Hamiltonian constraint on $\tilde{E}^{a}(x)\ =\ \tilde{E}^{a}_{i}(x)\tau_{i}$.\\
\begin{equation}\label{july7(6)}
\begin{array}{lll}
\{H[N], \tilde{E}^{b}_{k}(x)\}\ =\ 2{\cal D}_{a}(N\tilde{E}^{a}_{i}\tilde{E}^{b}_{j}\epsilon^{ijk}) \\
=\ 2\partial_{a}(N\tilde{E}^{a}_{i} \tilde{E}^{b}_{j}\epsilon^{ijk})\ +\ 2\epsilon^{kmn}\ A_{a}^{m}\ \omega^{n}
\end{array}
\end{equation}
where $\omega^{n} = \epsilon^{ijn}N\tilde{E}^{a}_{i}\tilde{E}^{b}_{j}$.
As shown in appendix \ref{HactiononE}, the right hand side of eq.(\ref{july7(6)}) can be re-written as,
\begin{equation}\label{july7(7)}
\begin{array}{lll}
\{H[N], \tilde{E}^{b}_{k}(x)\} =\ 2\partial_{a}(N\tilde{E}^{a}_{i} \tilde{E}^{b}_{j}\epsilon^{ijk})\ +\ 2\epsilon^{kmn}\ A_{a}^{m}\ \omega^{n}\\
\vspace*{0.1in}
\hspace*{0.5in}=\ {\cal L}_{NE_{jk}}\tilde{E}^{b}_{j}\ +\ N_{0}e^{b}_{j}\partial_{a}\tilde{E}^{a}_{jk} + 2 \epsilon^{kmn}\ A_{a}^{m}\ \epsilon^{ijn}N\tilde{E}^{a}_{i}\tilde{E}^{b}_{j}\\
\vspace*{0.1in}
\hspace*{0.5in}=\ {\cal L}_{NE_{jk}}\tilde{E}^{b}_{j}\ +\ N_{0}e^{b}_{j}{\cal D}_{a}\tilde{E}^{a}_{jk} +  \epsilon^{kmn}\ A_{a}^{m}\ \epsilon^{ijn}N\tilde{E}^{a}_{i}\tilde{E}^{b}_{j}
\end{array}
\end{equation}

As ${\cal D}_{a}\tilde{E}^{a}_{jk}\ =\ \epsilon^{ink}{\cal D}_{a}\tilde{E}^{a}_{i}$ vanishes on the constraint surface, we have the following remarkable result
\begin{equation}\label{HonE-final}
\begin{array}{lll}
\{H[N], \tilde{E}^{b}_{k}(x)\tau_{k}\}\ =\\
\vspace*{0.1in}
\hspace*{1.5in}[\tau_{i}, {\cal L}_{N\tilde{E}_{i}}\tilde{E}^{b}_{j}\tau^{j}]\ +\ [\tau^{m}, \epsilon^{ijn}A_{a}^{m}\left(NE^{a}_{i}\right)E^{b}_{j}\tau_{n}]\\
\vspace*{0.1in}
\hspace*{1.5in}= [\tau_{m}, {\cal L}_{N\tilde{E}_{m}}\tilde{E}^{b}]\ +\ [\tau^{m}, {\cal R}_{\Lambda^{\prime}_{(m)}}E^{b}]
\end{array}
\end{equation}
where $\Lambda^{\prime}_{(m)}$ is a triple of gauge parameters  with $(\Lambda^{\prime}_{(m)})^{i}$ being $A_{a}^{m}\ N\tilde{E}^{a}_{i}$.
Eqs. (\ref{HonA-final}) and (\ref{HonE-final}) con action of Euclidean Hamiltonian constraint on the phase space can be understood in terms of spatial diffeomorphism , local Gauge transformation and an adjoint action by a fixed element in $su(2)$. 

The action of Hamiltonian constraint on connection can also be understood by considering the following vector field on phase space.
\begin{equation}\label{eq:hamvecfield}
\begin{array}{lll}
X_{H[N]}\ =\ \int_{\Sigma}d^{3}x\textrm{Tr}\left(\left(\left[\tau^{k}, {\cal L}_{NE^{k}}A_{a}\right]\ +\ \left[\tau^{k}, {\cal R}_{\Lambda^{(k)}}A_{a}\right]\right)(x)\frac{\delta}{\delta A_{a}(x)}\right)
\end{array}
\end{equation}
This formal operator is the $SU(2)$ counterpart of a similar expression in the $U(1)^{3}$ theory \cite{tv}, whose heuristic operator action on cylindrical functions motivated a finite triangulation operator in quantum theory.\\
In fact
 it is instructive to compute the commutator $[X_{H[N]}, X_{H[M]}]$ as done in appendix (\ref{commutator}). It equals the vector field associated to Diffeomorphism constraint with triad dependent Shift field. This shows why we expect to quantize Right hand side of the constraint algebra in terms of phase space dependent diffeomorphisms.\\ 
We will now analyze the geometric action of Hamiltonian constraint on a holonomy functional and show that there exists a regularization of the resulting expression in terms of finite diffeomorphisms and holonomies along ``small" segments. A priori this seems incorrect as Hamiltonian is quadratic in the triads and hence action of it`s Hamiltonian vector field will map a holonomy into a function of triad and connection. However, and this is one of the key points of our construction that this function can be thought of as a holonomy defined with respect to an edge which is diffeomorphically related to original edge by a diffeomorphism which generated by the triad field.\\
Whence if $X_{H[N]}$ is the vector field associated to the Hamiltonian constraint as defined in (\ref{eq:hamvecfield}), we have 
\begin{equation}
\begin{array}{lll}
X_{H[N]}\ h_{e}(A)\ =\ \int dt h_{e[0,t]}\circ X_{H[N]}A_{a}(e(t))\dot{e}^{a}(t)\circ h_{e[t,1]}\\
\vspace*{0.1in}
\hspace*{0.3in}=\ \int dt h_{e[0,t]}\circ \left([\ \tau_{k}, {\cal L}_{NE_{k}}A_{a}(e(t))\dot{e}^{a}(t) ] - [\ \tau_{k}, {\cal D}_{a}N\tilde{E}^{b}_{k}A_{b}\ ]\circ\ h_{e[t,1]}\right) 
\end{array}
\end{equation}
We now make the following assumption. Let us assume that support of $N\tilde{E}_{i}$ is only in an infinitesimal neighborhood of $b(e)$. This assumption is motivated by the the following scenario in quantum theory. In LQG one usually works with quantum operator corresponding to $N\frac{\tilde{E}^{i}}{q^{\alpha}}$ where $\alpha\ \neq\ 0$. As the inverse volume operator $\widehat{\frac{1}{q^{\alpha}}}$ has support only at a certain non-degenerate vertices of the spin-networks.\\ 
If we consider the support of $N\tilde{E}_{k}$ to be in a ball of co-ordinate radius $\delta$ then up to terms of order $\delta^{2}$ one can approximate the above equation by, 
\begin{equation}\label{eq:regularizedconstraintaction}
\begin{array}{lll}
X_{H[N]}(h_{e}(A))\ \approx\ \delta X_{H[N]}A_{a}(b(e))\dot{e}^{a}(0)\circ h_{e}(A)\\
\vspace*{0.1in}
\hspace*{0.5in}=\ \delta[\tau_{k}, {\cal L}_{NE_{k}}A_{a}(b(e))\dot{e}^{a}(0)]\circ h_{e}(A)\ -\ \delta[\tau_{k}, {\cal D}_{a}\left(N\tilde{E}^{b}_{k}A_{b}\right)(b(e))\dot{e}^{a}(0)]\circ h_{e}(A)\\
\end{array}
\end{equation}
We now analyze the two terms in the above equation separately and choose one \emph{specific} regularized approximant which is expressed in terms of holonomies, finite diffeomorphisms and finite gauge transformations.\\
Let $\delta{\cal L}_{NE_{k}}A_{a}(b(e))\dot{e}^{a}(0)\ \approx\ h_{s_{e}^{\delta}}({\cal L}_{NE_{k}}A_{a})\ -\ {\bf 1}$. $s_{e}^{\delta}$ is a segment of parameter length $\delta$ which begins at $b(e)$ and is contained in $e$.
\begin{equation}\label{dec27-1}
\begin{array}{lll}
\delta[\tau_{k}, {\cal L}_{NE_{k}}A_{a}(b(e))\dot{e}^{a}(0)]\circ h_{e}(A)\ =\\
\vspace*{0.1in}
\hspace*{0.5in} [\tau_{k}, h_{s_{e}^{\delta}}\left({\cal L}_{NE_{k}}(A)\right)]\circ h_{e}(A)\approx\\
\vspace*{0.1in}
\hspace*{0.5in} \frac{1}{\delta}[\tau_{k}, h_{\phi^{\delta}_{N\tilde{E}_{k}}\circ s_{e}^{\delta}}(A)\ -\ h_{s_{e}^{\delta}}(A)]\circ h_{e}(A)=\\
\vspace*{0.1in}
\hspace*{0.5in} \frac{1}{\delta}[\tau_{k}, h_{\phi^{\delta}_{N\tilde{E}_{k}}\circ s_{e}^{\epsilon}}(A)h_{s_{e}^{\delta}}(A)^{-1}]\circ h_{s_{e}^{\delta}}(A)\circ h_{e}(A)\approx\\
\vspace*{0.1in}
\hspace*{0.5in} \frac{1}{\delta}[\tau_{k}, h_{\phi^{\delta}_{N\tilde{E}_{k}}\circ s_{e}^{\epsilon}}(A)h_{s_{e}^{\delta}}(A)^{-1}]\circ h_{e}(A)=\\
\vspace*{0.1in}
\hspace*{0.5in} \frac{1}{\delta}\left(\tau_{k}\circ h_{\phi^{\delta}_{N\tilde{E}_{k}}\circ s_{e}^{\epsilon}}\circ h_{e-s_{e}^{\delta}}\ -\ h_{\phi^{\delta}_{N\tilde{E}_{k}}\circ s_{e}^{\epsilon}}h_{s_{e}^{\delta}}^{-1}\circ \tau_{k}\circ h_{e}\right)
\end{array}
\end{equation}
In the second line, we have approximated a Lie dragged holonomy by difference of two holonomies which are related by a one parameter family of (parametrized by $\delta$) finite diffeomorphisms. In the forth line, we have replaced $h_{s_{e}^{\delta}}$ outside the commutator by 1, as this is the only term that is relevant at leading order in $\delta$. we now perform the same analysis for the second term in (\ref{eq:regularizedconstraintaction}) as 
follows.
\begin{equation}\label{eq:Gausstermregularized}
\begin{array}{lll}
\int dt h_{e[0,t]}\circ [\tau_{k}, {\cal D}_{a}(\Lambda^{(k)})]\circ h_{e[t,1]}=\\
\vspace*{0.1in}
\int dt \left(h_{e[0,t]}\circ\left({\cal D}_{a}(\Lambda^{(k)}\circ \tau_{k})\right)(e(t))\dot{e}^{a}(t)\circ h_{e[t,1]}\right)\ -\ \int dt \left(h_{e[0,t]}\circ\left({\cal D}_{a}(\tau_{k}\circ\Lambda^{(k)})\right)(e(t))\dot{e}^{a}(t)\circ h_{e[t,1]}\right)\\
\vspace*{0.1in}
\hspace*{0.5in}=\left(\Lambda^{(k)}\circ\tau_{k}\right)(b(e))\circ h_{e}\ -\ \left(\tau^{k}\circ\Lambda^{(k)}\right)(b(e))\circ h_{e}\\
\vspace*{0.1in}
\hspace*{0.5in}=\left(\tau_{i}\circ\ N\tilde{E}^{a}_{i}A_{a}^{k}\tau_{k}\right)(b(e))\circ h_{e}\ -\  \left(N\tilde{E}^{a}_{i}A_{a}^{k}\tau_{k}\circ\ \tau_{i}\right)(b(e))\circ h_{e}\\
\vspace*{0.1in}\approx\frac{1}{\delta}\left([\tau^{i}, h_{s^{\delta}_{N\tilde{E}_{i}(b(e))}}]\circ h_{e}\right)
\end{array}
\end{equation}
In going from the second to the third line in the above equation, we have used the usual action of infinitesimal gauge transformation on a holonomy. However as, we have assumed that support of $N\tilde{E}_{i}$ is only in the neighborhood of $b(e)$, the terms $h_{e}\circ \Lambda(f(e))$ vanish. In the final equation we have approximated $N\tilde{E}^{a}_{i}A_{a}^{k}\tau_{k}(b(e))$ by a holonomy along the vector $N\tilde{E}_{i}(b(e))$ of parameter length $\delta$.\\ 
Whence using eqs. (\ref{dec27-1}) and (\ref{eq:Gausstermregularized}), we have the following.
\begin{equation}\label{eq:regconstraintfinal}
\begin{array}{lll}
X_{H[N]}\ h_{e}\ =\\
\frac{1}{\delta}\left(\tau_{k}\circ h_{\phi^{\delta}_{N\tilde{E}_{k}}\circ s_{e}^{\delta}}\circ h_{e-s_{e}^{\delta}}\ -\ h_{\phi^{\delta}_{N\tilde{E}_{k}}\circ s_{e}^{\delta}}h_{s_{e}^{\delta}}^{-1}\circ \tau_{k}\circ h_{e}\right)\ -\ \frac{1}{\delta}\left([\tau^{k}, h_{s^{\delta}_{N\tilde{E}_{k}(b(e))}}]\circ h_{e}\right)
\end{array}
\end{equation}
We would like to make choices to define regularized (density two) Hamiltonian constraint such that its action on a spin-network ``mimics" the action of the classical constraint as in (\ref{eq:regconstraintfinal}). However as the Hamiltonian constraint operator involves curvature which in-turn can be defined in terms of closed loops. On the other hand the (regularized) action of $X_{H[N]}$ on $h_{e}$ involves open segments $h_{\phi_{N\tilde{E}_{k}}\circ s_{e}^{\delta}}h_{s_{e}^{\delta}}^{-1}$, $h_{s^{\delta^{\prime}}_{N\tilde{E}_{k}(b(e))}}$. 
However as can be easily verified, to leading order in $\delta$ we can replace the right hand side of (\ref{eq:Gausstermregularized}) by yet another expression which involves closed loops based at $b(e)$.\\
\begin{equation}\label{eq:finalclassicalaction}
\begin{array}{lll}
X_{H[N]}\ h_{e}\ =\\ \sum_{k}\left[\frac{1}{\delta}\left(\tau_{k}\circ h_{s^{\delta}_{N\tilde{E}_{k}(b(e))}}\circ h_{\phi^{\delta}_{N\tilde{E}_{k}}\circ s_{e}^{\delta}}\circ h_{s_{e}^{\delta}}^{-1}\ -\ h_{s^{\delta}_{N\tilde{E}_{k}(b(e))}}\circ h_{\phi^{\delta}_{N\tilde{E}_{k}}\cdot s_{e}^{\delta}}h_{s_{e}^{\delta}}^{-1}\circ \tau_{k}\right)\right]\circ h_{e}
\end{array}
\end{equation}
$s^{\delta}_{N\tilde{E}_{k}(b(e))}$ is any segment with tangent at the initial point being given by $N\tilde{E}_{k}(b(e))$.\\
Intuitively, $h_{s^{\delta}_{N\tilde{E}_{k}(b(e))}}\circ h_{\phi_{N\tilde{E}_{k}}\circ s_{e}^{\delta}}\circ h_{s_{e}^{\delta}}^{-1}$ can form a loop based at $b(e)$ as it first moves along the integral vector field of $N\tilde{E}_{k}$ by say affine parameter distance $\delta$, then it joins the segment $\phi^{\delta}_{N\tilde{E}_{k}}\circ s_{e}^{\delta}$ and finally rejoins $b(e)$ along $(s_{e}^{\delta})^{-1}$.\\
Using the Leibiniz rule and a basic identity $\delta\sum_{i}x_{i}\ =\ \prod_{i}(1+\delta x_{i})\ -\ 1\ +\ \textrm{O}(\delta^{2})$, one can also show that, 
\begin{equation}
\begin{array}{lll}
X_{H[N]}\ \pi_{j_{e}}(h_{e})\ =\\ \sum_{k}\left[\frac{1}{\delta}\left(\pi_{j_{e}}(\tau_{k})\circ \pi_{j_{e}}(h_{s^{\delta}_{N\tilde{E}_{k}(b(e))}}\circ h_{\phi^{\delta}_{N\tilde{E}_{k}}\circ s_{e}^{\delta}}\circ h_{s_{e}^{\delta}}^{-1})\ -\ \pi_{j_{e}}(h_{s^{\delta}_{N\tilde{E}_{k}(b(e))}}\circ h_{\phi^{\delta}_{N\tilde{E}_{k}}\cdot s_{e}^{\delta}}h_{s_{e}^{\delta}}^{-1})\circ \pi_{j_{e}}(\tau_{k})\right)\right]\\
\vspace*{0.1in}
\hspace*{5.7in}\circ \pi_{j_{e}}(h_{e})
\end{array}
\end{equation}
Whence given a spin-network functional $f(A)\ =\ \otimes_{i=1}^{N_{v}}\pi_{j_{e_{i}}}(h_{e_{i}})$ based on a collection of edges $e_{1},\dots, e_{N_{v}}$ which meet at a vertex $v$, the action of $X_{H[N]}$ on $f(A)$ can be approximated as,
\begin{equation}\label{actionofHonclassicalspinnet}
\begin{array}{lll}
X_{H[N]}f(A)\ =\\
\frac{1}{\delta}\sum_{i=1}^{N_{v}}\pi_{j_{e_{1}}}(h_{e_{1}})\otimes\dots\otimes\left( \sum_{k}\left[\pi_{j_{e_{i}}}(\tau_{k}), \pi_{j_{e_{i}}}(h_{s^{\delta}_{N\tilde{E}_{k}(v)}}\circ h_{\phi^{\delta}_{N\tilde{E}_{k}}\circ s_{e_{i}}^{\delta}}\circ h_{s_{e_{i}}^{\delta}}^{-1})\right]\circ \pi_{j_{e_{i}}}(h_{e_{i}})\right)\otimes\dots\\
\vspace*{0.1in}
\hspace*{5.7in}\otimes\pi_{j_{e_{N_{v}}}}(h_{e_{N_{v}}})
\end{array}
\end{equation}

We will motivate a quantization of Hamiltonian constraint using holonomies and fluxes whose action of spin-network \emph{precisely} encapsulates (\ref{actionofHonclassicalspinnet}).\\
Finally before proceeding further we note that in order to quantize the constraint, we need to construct an analog of $N\tilde{E}^{a}_{k}(v)$ in quantum theory. In the context of $U(1)^{3}$ theory, such operators were called quantum-shifts. In section \ref{qshift}, we construct Quantum shifts for the $SU(2)$ theory, taking a cue from [\cite{tv},\cite{hlt1}].
\section{The Quantum shift}\label{qshift}
Let $\vert{\bf s}\rangle$ be a spin-network with a vertex $v$. Consider a chart $\{x\}$ at $v$. Let ${\cal B}_{\epsilon}(v)$ be a ball of radius $\epsilon$ contained in the co-ordinate chart $\{x\}$.\\
Consider
\begin{equation}\label{qs1}
\begin{array}{lll}
\widehat{N^{a}_{i}}(v)\vert_{\epsilon}\vert{\bf s}\rangle :=\ N(x(v))\frac{1}{\frac{4\pi\epsilon^{3}}{3}}\int_{{\cal B}_{\epsilon}(v)}d^{3}x\hat{E}^{a}_{i}(x)\vert{\bf s}\rangle
\end{array}
\end{equation}
for $v^{\prime}\in{\cal B}_{\epsilon}(v)$, 
\begin{equation}\label{qs2}
\hat{E}^{a}_{i}(x)\vert{\bf s}\rangle\ =\ \sum_{e\in E({\bf s})\vert b(e)=v}\int_{0}^{1}dt\dot{e}^{a}(t)\delta^{3}(e(t), x(v^{\prime}))\hat{\tau}^{j}\vert_{e(t)}\ \vert{\bf s}\rangle
\end{equation}
where $\hat{\tau}^{j}\vert_{e(t)}$ is an insertion of the Pauli matrix in the spin-network ${\bf s}$ at $e(t)$.\\
Substituting \ref{qs2} in \ref{qs1} we get,
\begin{equation}
\begin{array}{lll}
\widehat{N^{a}_{i}}\vert_{\epsilon}(v)\vert{\bf s}\rangle=\\
\vspace*{0.1in}
N(x(v))\frac{1}{\frac{4\pi\epsilon^{3}}{3}}
\sum_{e\in E({\bf s})\vert b(e)=v}\int_{{\cal B}_{\epsilon}(v)\cap e}\ de^{a}\ \hat{\tau}^{j}\vert_{e}\vert{\bf s}\rangle\\
\end{array}
\end{equation}
In the above equation $\hat{\tau}^{j}\vert_{e}$ is an insertion at a point along edge $e$ which lies inside the ball ${\cal B}_{\epsilon}(v)$.Whence to leading order in $\epsilon$,
\begin{equation} 
\begin{array}{lll}
\int_{{\cal B}_{\epsilon}(v)\cap e}\ de^{a}\ \hat{\tau}^{j}\vert_{e}\vert{\bf s}\rangle\ = \epsilon\hat{e}^{a}(0)\hat{\tau}^{j}\vert_{v}\vert{\bf s}\rangle
\end{array}
\end{equation}
where $\hat{e}(0)$ is the vector tangent to $e$ at it`s beginning point $v$ and is normalized w.r.t the co-ordinate chart $\{x\}_{v}$.\\
Thus finally, the (regularized) quantum shift is defined as
\begin{equation}\label{qshift-definition}
\begin{array}{lll}
\widehat{N^{a}_{i}}(v)\vert_{\epsilon}\vert{\bf s}\rangle :=\ N(x(v))\frac{1}{\frac{4\pi\epsilon^{2}}{3}}\sum_{e\in E(\gamma)\vert b(e)=v}\hat{e}^{a}(0)\hat{\tau}_{i}\vert_{v}\vert{\bf s}\rangle
\end{array}
\end{equation}
Reader can easily verify for herself that action of the quantum shift on state $\vert{\bf s}\rangle$ as defined in eq. (\ref{spinnetdef}) results in a linear combination of spin-network states as
\begin{equation}
\begin{array}{lll}
\widehat{N^{a}_{i=1}}(v)\vert_{\epsilon}\vert{\bf s}\rangle_{m_{e_{1}}\dots m_{e_{N_{v}}}}\ =\\ 
\frac{1}{2i}N(x(v))\frac{1}{\frac{4\pi\epsilon^{2}}{3}}\sum_{e\in E(\gamma)\vert b(e)=v}\hat{e}^{a}(0)\left[\alpha(+,j_{e},m_{e})\vert{\bf s}\rangle_{m_{e_{1}}\dots m_{e}+1\dots m_{e_{N_{v}}}}\ +\ \alpha(-,j_{e},m_{e})\vert{\bf s}\rangle_{m_{e_{1}}\dots m_{e}-1\dots m_{e_{N_{v}}}}\right]\\
\widehat{N^{a}_{i=2}}(v)\vert_{\epsilon}\vert{\bf s}\rangle_{m_{e_{1}}\dots m_{e_{N_{v}}}}\ =\\
\frac{1}{2i}N(x(v))\frac{1}{\frac{4\pi\epsilon^{2}}{3}}\sum_{e\in E(\gamma)\vert b(e)=v}\hat{e}^{a}(0)\left[\alpha(+,j_{e},m_{e})\vert{\bf s}\rangle_{m_{e_{1}}\dots m_{e}+1\dots m_{e_{N_{v}}}}\ -\ \alpha(-,j_{e},m_{e})\vert{\bf s}\rangle_{m_{e_{1}}\dots m_{e}-1\dots m_{e_{N_{v}}}}\right]\\
\widehat{N^{a}_{i=3}}(v)\vert_{\epsilon}\vert{\bf s}\rangle_{m_{e_{1}}\dots m_{e_{N_{v}}}}\ =\ \frac{1}{2i}N(x(v))\frac{1}{\frac{4\pi\epsilon^{2}}{3}}\sum_{e\in E(\gamma)\vert b(e)=v}m_{e}\hat{e}^{a}(0)\vert{\bf s}\rangle_{m_{e_{1}}\dots m_{e_{N_{v}}}}
\end{array}
\end{equation}
where
\begin{equation}
\begin{array}{lll}
\alpha(+,j_{e},m_{e})\ =\ \sqrt{j_{e}(j_{e}+1)\ -\ m_{e}(m_{e}-1)}\\
\vspace*{0.1in}
\alpha(-, j_{e},m_{e})\ =\ \sqrt{j_{e}(j_{e}+1)\ -\ m_{e}(m_{e}+1)}
\end{array}
\end{equation}

\section{Defining $\hat{H}_{T(\delta)}[N]$.}\label{finitetrianglecon}
Following the seminal work of Thiemann in \cite{qsd1,qsd2,qsd3} there is a clear skeleton of how to define quantum Hamiltonian constraint in Loop Quantum Gravity. However as is well known, the construction comes with an infinite-dimensional parameter space worth of ambiguities. Although a Large subset of the ambiguities are irrelevant once the continuum limit is taken, an infinite dimensional worth of them, typically associated to choice of discrete approximant for the curvature survive giving rise to an infinite number of candidate Hamiltonian constraints. We will now argue that among such choices, there do exist a set, such that the resulting quantum constraint action of spin-network states correspond to the discretized classical action on holonomies as datailed in \ref{eq:regconstraintfinal}. We assume that the reader is familiar with the constructions in \cite{qsd1}, \cite{aldiff}. We refer the reader to \cite{alreview},\cite{nicolaietal},\cite{blreview} for slightly less technical overview of the discretization schemes in LQG.\\
We sketch a quantization of $H[N]$ on ${\cal H}_{kin}$. A rigorous 
derivation of all the composite operators involved in the definition is not presented here in the interest of brevity. We will outline the main steps involved in the derivation and highlight the key assumptions which are based on defining such operators in LQG. For detailed applications of such ideas, we refer the reader to \cite{aldiff},\cite{hlt1}.\\ 
Once again we will like to remind the reader that we will be quantizing Hamiltonian constraint of density weight two. However as our main goal (beyond this paper) is to quantize density $\frac{4}{3}$ constraint, we will make certain quantization choices with an eye towards quantization of such a constraint.\\

Given a graph $\gamma({\bf s})$ underlying some spin-network ${\bf s}$, we choose a one parameter family of ``triangulations" $T(\delta)$ where $\delta$ measures an appropriate fine-ness of the triangulation.\footnote{It is not strictly necessary that one works with triangulations. Cubulations suffice as well and in fact, in \cite{aldiff}, \cite{hlt1} some weaker versions of cubulations were used, where certain cuboids overlapped each other in co-dimension zero regions, which however did not affect the approximation of an integral by a Riemann sum in the sense that, the error induced by such regions was sub-dominant. Such details on triangulation of $\Sigma$ drop out in the continuum limit.}  \cite{qsd1} 
\begin{displaymath}
\lim_{\delta\rightarrow 0}T(\delta)\ =\ \Sigma
\end{displaymath}
Our choice of $T({\delta})$ is as follows.\\
{\noindent (1)} Away from the support of $\gamma({\bf s})$ we choose an arbitrary triangulation of $\Sigma$.\\
{\noindent (2)} Given a $v\in V(\gamma({\bf s})$, which is $N_{v}$-valent, for all $\delta$  the triangulation in the neighborhood of $v$ is given by a single three-dimensional tetrahedron which intersects all the edges incident at $v$.\\
{\noindent (3)} With each $\triangle\ \in\ T$, we associate a vertex $v_{\triangle}$ such that for $\triangle$ which bounds $v$, we have $v_{\triangle}\ =\ v$.\\
As a result of triangulating $\Sigma$, $H[N]$ can be approximated by
\begin{equation}
\begin{array}{lll}
H_{T}[N]\ =\\
\vspace*{0.1in}
\frac{1}{2}\sum_{\triangle\in T_{\delta}}\textrm{vol}(\triangle)\ N(v_{\triangle})\epsilon^{ijk}F_{ab}^{i}(v_{\triangle})\ E^{a}_{j}(v_{\triangle})\ E^{b}_{k}(v_{\triangle})
\end{array}
\end{equation}
The basic idea now is to use $N(v_{\triangle})\hat{E}^{a}_{j}(v_{\triangle})$ as a dynamical Shift field. If we were quantizing a density $\frac{4}{3}$ constraint, then the  Shift vector would be $N(v_{\triangle})\hat{E}^{a}_{j}(v_{\triangle})\frac{1}{q^{\frac{1}{6}}}(v_{\triangle})$, whose quantization will have non-trivial action on a state $\vert{\bf s}\rangle$ iff $v_{\triangle}$ is a non-degenerate (in the sense of non-degeneracy of volume operator) vertex of ${\bf s}$.  In order to replicate the same scenario in the case at hand, we assume  a definition of $N(v_{\triangle})\hat{E}^{a}_{j}(v_{\triangle})$ such that given a state $\vert{\bf s}\rangle$ with a $N_{v}$ valent vertex $v$, 
\begin{equation}
N(v_{\triangle})\hat{E}^{a}_{i}(v_{\triangle})\vert{\bf s}\rangle\ =\ 0
\end{equation}
unless $v_{\triangle}\ =\ v$. In which case the operator acts via eq.(\ref{qshift-definition}).
Whence we now quantize the Hamiltonian constraint on $\vert{\bf s}\rangle$ as,
\begin{equation}\label{dec28-2}
\begin{array}{lll}
\hat{H}_{T}[N]\ =\ \sum_{m=1}^{N_{v}}\frac{{\cal A}^{\prime}}{\delta^{2}}\delta^{3}\frac{1}{2}N(v)\ \left[\epsilon^{ijk}\widehat{\hat{e}^{a}_{m}(0)\ F_{ab}^{i}\ E^{b}_{k}}(v)\otimes\pi_{j_{e}}(\hat{\tau}_{j})\right]\vert{\bf s}\rangle
\end{array}
\end{equation}
${\cal A}^{\prime}\ =\ \frac{3}{4\pi}$.\\
The above equation needs some clarification. The quantum shift acts additively by inserting a Pauli-matrix along each edge $e$.This is the reason for the hat on $\tau_{i}$; implying that it is an insertion operator.
 We will now assume that, there exists a quantization of $\epsilon^{ijk}\widehat{\hat{e}^{a}_{m}(0)\ F_{ab}^{i}\ E^{b}_{k}}(v)$ on $\vert{\bf s}\rangle$ which is given by
\begin{equation}\label{dec28-1}
\begin{array}{lll}
\epsilon^{ijk}\ \widehat{\hat{e}^{a}_{m}(0)\ F_{ab}^{i}\ E^{b}_{k}}(v)\ \vert{\bf s}\rangle\ =\ \frac{1}{\delta^{4}}\sum_{n=1\vert n\neq m}^{N_{v}}[\ \pi_{j_{e_{n}}}(\hat{\tau}_{j}), \hat{h}_{\alpha(s_{e_{m}^{\delta},s_{e_{n}}^{\delta}})}]\vert{\bf s}\rangle
\end{array}
\end{equation}
$s_{e}^{\delta}$ is a segment in the edge $e$ with $b(s_{e}^{\delta})\ =\ v$. its parameter range is $\delta$. $\alpha(s_{e_{m}}^{\delta},s_{e_{n}}^{\delta})$ is a loop obtained by first traversing along the segment $s_{e_{m}}^{\delta}$, then along an arc that we will refer to as $\phi^{\delta}_{\hat{e_{m}}(0)}(s_{e_{n}}^{\delta})$ and finally along $(s_{e_{n}}^{\delta})^{-1}$.  The meaning of $\phi^{\delta}_{\hat{e_{m}}(0)}$ is the following : It is formally defined to be a one parameter family of diffeomorphism along the (regularized) quantum shift as defined in eq.(\ref{qshift-definition}). One possibility to define such a diffeomorphism precisely is to extend the tangent $\hat{e}_{m}(0)$ smoothly to some compact set around $v$ as done in \cite{tv}. In this paper, we will simply assume that the arc $\phi^{\delta}_{\hat{e}_{m}(0)}(s_{e_{n}}^{\delta})$ is an arc which is such that the loop $\alpha(s_{e_{m}}^{\delta},s_{e_{n}}^{\delta})$ is planar which is consistent with the idea that $\phi^{\delta}_{\hat{e}_{m}(0)}(s_{e_{n}}^{\delta})$ is a deformation of $s_{e_{n}}^{\delta}$ along a ``diffeomorphism" generated by $\hat{e}_{m}(0)$. All the details regarding the precise shape of the loop $\alpha$ will drop out in the continuum limit, except the degree of its differentiabilityy at the beginning and end points of $s_{e_{n}}^{\delta}$ which we will specify below.\\
We also assume that the co-ordinate area of the loop $\alpha$ is $\delta^{2}$. the hat on  $\hat{h}_{\alpha(s_{e_{m}^{\delta},s_{e_{n}}^{\delta}})}$ is meant to indicate that this operator acts along the edge $e_{n}$ by multiplying $\pi_{j_{e_{n}}}(h_{e_{n}})$ from the left
\begin{equation}
\hat{h}_{\alpha(s_{e_{m}^{\delta},s_{e_{n}}^{\delta}})}\vert{\bf s}\rangle\ \equiv
\dots\otimes\pi_{j_{e_{n}}}(\hat{h}_{\alpha(s_{e_{m}^{\delta},s_{e_{n}}^{\delta}})}\cdot h_{e_{n}})\otimes\dots
\end{equation}
Based on similar constructions in \cite{aldiff},\cite{hlt1} we believe there is sufficient evidence that such an assumption can be justified by detailed rigorous analysis. It is easy to see that such a quantization exists in $j=\frac{1}{2}$ case. That is, there is a choice of curvature operator such that action of $\epsilon^{ijk}\ \widehat{\hat{e}^{a}_{m}(0)F_{ab}^{i}E^{b}_{k}}(v)$ on a spin-$\frac{1}{2}$, single edge state is given by
\begin{equation}
\begin{array}{lll}
\epsilon^{ijk}\ \widehat{\hat{e}^{a}_{m}(0)F_{ab}^{i}E^{b}_{k}}(v)\ (h_{e^{\prime}})_{AB}
\approx\\
\epsilon^{ijk}\ \widehat{\hat{e}^{a}_{m}(0)F_{ab}^{i}\hat{e}^{\prime b}(0)}\hat{\tau}_{k}\vert_{b(e^{\prime})}\ (h_{e^{\prime}})_{AB}\\
\vspace*{0.1in}
= \epsilon^{ijk}\ \textrm{Tr}(\hat{h}_{\alpha(s_{e},s_{e^{\prime}})}\cdot\tau_{i})\vert_{b(e^{\prime})}\hat{\tau}_{k}\vert_{b(e^{\prime})} \ (h_{e^{\prime}})_{AB}
\end{array}
\end{equation}
where in the second line, we have quantized $E^{b}_{k}$ as a operator valued distribution which acts via an insertion operator at $b(e^{\prime})$.\footnote{The distributional factors like $\delta(v, b(e^{\prime}))$ are suppressed as they are not reelvant for the discussion at hand.} In the third line, we have chosen a particular (in fact standard) quantization of $\hat{F}_{ab}^{i}\hat{e}^{a}(0)\hat{e^{\prime}}^{b}(0)$, which due to operator ordering acts after the insertion of $\tau_{k}$. On using 
\begin{equation}
\epsilon^{ijk}\hat{\tau}_{i}\ =\ [\hat{\tau}_{j},\hat{\tau}_{k}]
\end{equation}
and the completeness relation for Pauli Matrices,
\begin{equation}
\sum_{i}(\tau_{i})_{AB}(\tau_{i})_{CD}\ =\ 2\delta_{AD}\delta_{BC}\ -\ \delta_{AB}\delta_{CD}
\end{equation}
we see that our assumption is valid when the operator acts on a spin$\frac{1}{2}$ holonomy.
\begin{equation}
\epsilon^{ijk}\ \hat{e}^{a}_{m}(0)F_{ab}^{i}E^{b}_{k}(v)\ (h_{e_{n}}) =\ 
[\ (\hat{\tau}_{j}), \hat{h}_{\alpha(s_{e_{m},s_{e_{n}}})}]\ (h_{e_{n}})
\end{equation}
On substituting eq.(\ref{dec28-1}) in eq.(\ref{dec28-2}) we see that, the finite triangulation Hamiltonian constraint operator acting on $\vert{\bf s}\rangle$ is given by,
\begin{equation}\label{dec28-5}
\begin{array}{lll}
\hat{H}_{T(\delta)}[N]\ =\ \frac{{\cal A}^{\prime}}{\delta^{3}}\frac{1}{2}N(v)\sum_{m=1}^{N_{v}}\sum_{n=1\vert n\neq m}^{N_{v}}\left([\ \pi_{j_{e_{n}}}(\hat{\tau}_{j}), \hat{h}_{\alpha(s_{e_{m}^{\delta},s_{e_{n}}^{\delta}})}]\otimes\pi_{j_{e_{m}}}(\hat{\tau}_{j})\right)\vert{\bf s}\rangle
\end{array}
\end{equation}
Notice that as the two operators $[\ \pi_{j_{e_{n}}}(\hat{\tau}_{j}), \hat{h}_{\alpha(s_{e_{m}^{\delta},s_{e_{n}}^{\delta}})}]$, $\pi_{j_{e_{m}}}(\hat{\tau}_{j})$ are insertion operators at $v$ along distinct edges $e_{n}, e_{m}$, there is no operator ordering ambiguity in the definition above.\\
This is one of our key results. We immediately recognize the conceptual similearity between the finite triangulation constraint defined in eq.(\ref{dec28-5}) and the discrete approximate to the $X_{H[N]}$ action as derived in eq. (\ref{actionofHonclassicalspinnet}). Indeed, the vector field $NE^{k}$ which turns into the quantum shift, is a linear combination of (Lie algebra valued) edge tangents. Whence the loop $h_{s^{\delta}_{N\tilde{E}_{k}(v)}}\circ h_{\phi^{\delta}_{N\tilde{E}_{k}}\circ s_{e_{n}}^{\delta}}\circ h_{s_{e_{n}}^{\delta}}^{-1}$ in eq. (\ref{actionofHonclassicalspinnet}) has a precise quantum analog in 
$\sum_{m} \pi_{j_{e_{m}}}(\hat{\tau}_{k})\otimes \hat{h}_{\alpha(s_{e_{m}}^{\delta},s_{e_{n}}^{\delta})}$ which occurs in eq.(\ref{dec28-5}).\\
Rather remarkably, close variants of the operator defined above have occurred (at least) twice in the literature before. We briefly comment on those two works below.\\
\noindent{\bf (1)}\ There is a spiritual similarity between the constraint written above and the Hamiltonian constraint operator constructed by Brugmann in a beautiful piece of work in the Loop representation, \cite{brugmann}. Brugmann's constraint was an operator on the so-called loop states and Modulo various smearings which are not relevant for the discussion here, his operator had the following structure. Consider a loop $\gamma$ which has a kink at $v$. (That is, the loop is not differentiable at $v$.) Let $\Psi(\gamma)$ be trace of Wilson loop along $\gamma$. Then, 
\begin{equation}
\hat{H}_{\delta}[N]\Psi(\gamma) \approx\ N(v)\hat{\tau}_{j}(s)\otimes\hat{e}^{a}(t)[\frac{\delta}{\delta \gamma^{a}(t)},\hat{\tau}_{j}(t)]\kappa_{\delta}(v,\gamma(s))\kappa_{\delta}(v,\gamma(t))
\end{equation}
where $\kappa_{\delta}$ is a characteristic function which has support in a ball of radius ($\delta$) around $v$. $\frac{\delta}{\delta \gamma^{a}(t)}$ is a so-called Shift operator which deforms the loop $\gamma$ along the direction in which it is contracted, (in our case $\hat{e}^{a}(t)$). This implies that the loop deforms at $v$ along one of the edges which is outgoing at the kink. Conceptually the action of this constraint is remarkably similar to the one we have obtained in eq.(\ref{dec28-5}). 
What is surprising is that Brugmann arrived at the quantum constraint, (specifically the Shift operator) through quite different considerations then the way that led us to our definition. As his quantum constraints satisfied the  \emph{formal} off-shell closure of constraint algebra in the old loop representation, it gives us an indication that there is chance that the constraint defined here may generate anomaly free Dirac algebra.\\

\noindent{\bf (2)} A slight variant of choices made here would lead to (in two space dimensions) the finite triangulation Hamiltonian constraint which is closely related to the constraint used by Bonzom and Freidel  in \cite{valentine1}  in the case of $2+1$ dimensional gravity.\\ 
Indeed it is easy to see that, we could have chosen a different quantization of $\epsilon^{ijk}\ \hat{e}^{a}_{m}(0)F_{ab}^{i}E^{b}_{k}(v)$ such that,
\begin{equation}\label{bonzom1}
\begin{array}{lll}
\hat{H}_{T}[N]\ =\ \frac{{\cal A}^{\prime}}{\delta^{3}}\frac{1}{2}N(v)\sum_{m=1}^{N_{v}}\sum_{n=1\vert n\neq m}^{N_{v}}\left(\left(\pi_{j_{e_{n}}}(\hat{\tau}_{j})
-\pi_{j_{e_{n}}}(\hat{h}_{\alpha(s_{e_{m}^{\delta},s_{e_{n}}^{\delta}})}^{-1}\tau_{j}\cdot\hat{h}_{\alpha(s_{e_{m}^{\delta},s_{e_{n}}^{\delta}})})\right)\otimes\pi_{j_{e_{m}}}(\hat{\tau}_{j})\right)\vert{\bf s}\rangle
\end{array}
\end{equation}
This is because the two composite operators $\pi_{j_{e_{n}}}(\hat{\tau}_{j})
-\pi_{j_{e_{n}}}(\hat{h}_{\alpha(s_{e_{m}^{\delta},s_{e_{n}}^{\delta}})}^{-1}\tau_{j}\cdot\hat{h}_{\alpha(s_{e_{m}^{\delta},s_{e_{n}}^{\delta}})})$ and,\\ $[\ \pi_{j_{e_{n}}}(\hat{\tau}_{j}), \hat{h}_{\alpha(s_{e_{m}^{\delta},s_{e_{n}}^{\delta}})}]$ differ only by operators of $O(\delta^{3})$. However the constraint defined in eq.(\ref{bonzom1}) is closely related to the quantum constraint used in \cite{valentine1} in the sense that, it uses precisely the same definition of curvature operator and same structure of insertions at the vertex. They key difference between the two definition lies in the choice of triangulation. As we are interested in the continuum limit, we work with one parameter family of triangulations, and hence the loops used in the construction of curvature are finer then the graph underlying ${\bf s}$. On the other hand, in \cite{valentine1}, the authors work with a fixed triangulation which is provided by the graph itself. However we still find the similarities surprising as the underlying motivations for making quantization choices are completely different in the two cases. The constraints defined in \cite{valentine1} has several remarkable properties e.g. its precise relationship with dynamical equations of Ponzano-Regge model. Hence it is an interesting open question to see if such an operator, analyzed in the spirit advocated in this paper would lead to an anomaly-free off-shell closure in three dimensional gravity.\\

As far as Hamiltonian constraint of density weight two is concerned, starting with $\hat{H}_{T(\delta)}[N]$ in eq.(\ref{dec28-5}) we could proceed to analyze the continuum limit of the finite triangulation commutators of the rescaled constraint $\delta^{2}\hat{H}_{T(\delta)}[N]$. However as our ultimate goal is to seek an anomaly-free quantization of constraint whose density weight is $\frac{4}{3}$, we modify the definition further which would make the structures chosen in this paper plausibly robust to the density $\frac{4}{3}$ case.\\
The new vertices created by the quantum constraint in eq.(\ref{dec28-5}) when it acts on $\vert{\bf s}\rangle$ are at most trivalent. This implies that if we were working with any constraint which was rescaled by $q^{-\alpha},\ \alpha\ >\ 0$, 
the action of quantum constraint on newly created vertices would vanish as they are in the kernel of volume operator.\footnote{There may very well exist alternative quantization of inverse volume operator which are such that the action of such an operator on so called degenerate (zero-volume) vertices is non-trivial, however we will ignore such alternatives in this paper.} In order to ensure that the newly created vertices can be non-degenerate, we modify $\hat{H}_{T(\delta)}[N]$ as follows.
Consider a finite triangulation operator, whose action on $\vert{\bf s}\rangle$ is defined as
\begin{equation}\label{dec30-1}
\begin{array}{lll}
\hat{H}^{\prime}_{T(\delta)}[N]\vert{\bf s}\rangle\ =\\
\frac{{\cal A}^{\prime}}{\delta^{3}}\frac{1}{2}N(v)\sum_{m=1}^{N_{v}}\sum_{n=1\vert n\neq m}^{N_{v}}\left(\left\{[\ \pi_{j_{e_{n}}}(\hat{\tau}_{j}), \hat{h}_{\alpha(s_{e_{m}^{\delta},s_{e_{n}}^{\delta}})}]\otimes_{i\neq n,m}\hat{h}_{\alpha(s_{e_{m}}^{\delta},s_{e_{i}}^{\delta})}\right\}\otimes\pi_{j_{e_{m}}}(\hat{\tau}_{j})\right)\vert{\bf s}\rangle
\end{array}
\end{equation}
where $\hat{h}_{\alpha(s_{e_{m}}^{\delta},s_{e_{i}}^{\delta})}$ acts on the state $\vert{\bf s}\rangle$ by a left multiplication on $\pi_{j_{e_{i}}}(h_{e_{i}})$.
To leading order in $\delta$ we have,
\begin{equation}
\hat{H}_{T(\delta)}[N]\ =\ \hat{H}^{\prime}_{T(\delta)}[N]
\end{equation}
 This is clear as, to $O(\delta^{2})$, only one of the holonomy operators along loops will contribute, and it has to be $\alpha(s_{e_{m}^{\delta},s_{e_{n}}^{\delta}})$, as to order $\delta^{2}$ if we consider the contribution due to any other loop, the commutator inside the operator vanishes. As long as we are interested in the continuum limit, we could choose to work with either of the two constraints given in eqs.(\ref{dec28-5}) or (\ref{dec30-1}).\\
For $i\neq n,m$ we will choose the loops as
\begin{equation}
\alpha(s_{e_{m}}^{\delta},s_{e_{i}}^{\delta})\ =\ (s_{e_{i}}^{\delta})^{-1}\circ a_{im}\circ s_{e_{m}}^{\delta}
\end{equation}
where $a_{im}$ is an arc chosen so that the loops $\alpha(s_{e_{m}}^{\delta},s_{e_{i}}^{\delta})$ are planar. That is, for  a given pair $(n,m)$ if $i \neq\ n,m$ then we do not relate these loops to the singular diffeomorphisms generated by the quantum shift.\\
A moment of reflection reveals that the new vertices created by $\hat{H}_{T(\delta)}[N]$ are at least $N_{v}$ valent, (we will see this feature explicitly in section (\ref{deformedstates})) whence \emph{even for} lower density weighted constraints which involve the inverse volume operator, these vertices will yield a non-trivial commutator in the continuum limit. Because of this reason $\hat{H}_{T(\delta)}^{\prime}[N]$ is an interesting candidate for the finite triangulation quantum constraint. Instead of however working with either $\hat{H}_{T(\delta)}[N]$ or $\hat{H}_{T(\delta)}^{\prime}[N]$ we choose to work with a quantum constraint defined as
\begin{equation}\label{averagingonH}
\begin{array}{lll}
\hat{H}_{T(\delta)}^{\prime\prime}[N]\ =\ \frac{1}{2}\hat{H}_{T(\delta)}[N]\ +\ \frac{1}{2}\hat{H}_{T(\delta)}^{\prime}[N]
\end{array}
\end{equation}
As we will see in later sections, $\hat{H}_{T(\delta)}^{\prime\prime}[N]$ not only creates new non-degenerate vertices; it also admits continuum limit on the Habitat considered in this paper. The same however is not true for $\hat{H}_{T(\delta)}^{\prime}[N]$. That is, on the Habitat space defined below, $\hat{H}_{T(\delta)}^{\prime}[N]$ does not admit a well defined continuum limit. It is plausible (and in fact likely in our opinion) that as the definition of Habitat is refined in the future, $\hat{H}_{T(\delta)}^{\prime}[N]$ will have a continuum limit. In this paper, we solve the problem rather unsatisfactorily by instead choosing the quantum constraint as defined in eq.(\ref{averagingonH}).\footnote{As far as density weight two Hamiltonian is concerned, the key mechanisms which underlie off-shell closure in this paper will in fact go through (modulo some technical modifications in the definition of the habitat), had we chosen to work purely with the original constraint which created only degenerate vertices. This we believe is an interesting check on the ideas presented in this paper.}\\
It will be useful to express the action of $\hat{H}_{T(\delta)}^{\prime}[N]$ as 
\begin{equation}\label{keydefinitionofH}
\begin{array}{lll}
\hat{H}^{\prime}_{T(\delta)}[N]\vert{\bf s}\rangle\ =\  N(v)\frac{{\cal A}^{\prime}}{2\delta^{3}}\sum_{m=1}^{N_{v}}\sum_{n=1\vert n\neq m}^{N_{v}}\left[\vert{\bf s}_{\delta}(1;v,e_{m},e_{n})\rangle\ -\ \vert{\bf s}_{\delta}(2;v,e_{m},e_{n})\rangle\right]
\end{array}
\end{equation}
where 
\begin{equation}
\begin{array}{lll}
\vert{\bf s}_{\delta}(1;v,e_{m},e_{n})\rangle\ =\ \left(\{\pi_{j_{e_{n}}}(\tau_{j})\cdot \hat{h}_{\alpha(s_{e_{m}}^{\delta},s_{e_{n}}^{\delta})}\}\otimes_{p\neq n,m}\hat{h}_{\alpha(s_{e_{m}}^{\delta},s_{e_{p}}^{\delta})}\right)\otimes\pi_{j_{e_{m}}}(\hat{\tau}_{j})\ \vert{\bf s}\rangle\\
\vspace*{0.1in}
\vert{\bf s}_{\delta}(2;v,e_{m},e_{n})\rangle\ =\ \left(\{\hat{h}_{\alpha(s_{e_{m}}^{\delta},s_{e_{n}}^{\delta})}\cdot\pi_{j_{e_{n}}}(\hat{\tau}_{j})\}\otimes_{p\neq n,m}\hat{h}_{\alpha(s_{e_{m}}^{\delta},s_{e_{p}}^{\delta})}\right)\otimes\pi_{j_{e_{m}}}(\hat{\tau}_{j})\ \vert{\bf s}\rangle
\end{array}
\end{equation}
Similarly we define states obtained by the action of $\hat{H}_{T(\delta)}[N]$ as,
\begin{equation}
\begin{array}{lll}
\hat{H}_{T(\delta)}[N]\vert{\bf s}\rangle\ =\ N(v)\frac{{\cal A}^{\prime}}{2\delta^{3}}\sum_{m=1}^{N_{v}}\sum_{n=1\vert n\neq m}^{N_{v}}\left[\vert{\bf s}_{\delta}(3;v,e_{m},e_{n})\rangle\ -\ \vert{\bf s}_{\delta}(4;v,e_{m},e_{n})\rangle\right]
\end{array}
\end{equation}
where
\begin{equation}
\begin{array}{lll}
\vert{\bf s}_{\delta}(3;v,e_{m},e_{n})\rangle\ =\ \left(\pi_{j_{e_{n}}}(\tau_{j})\cdot \hat{h}_{\alpha(s_{e_{m}}^{\delta},s_{e_{n}}^{\delta})}\right)\otimes\pi_{j_{e_{m}}}(\hat{\tau}_{j})\ \vert{\bf s}\rangle\\
\vspace*{0.1in}
\vert{\bf s}_{\delta}(4;v,e_{m},e_{n})\rangle\ =\ \left(\hat{h}_{\alpha(s_{e_{m}}^{\delta},s_{e_{n}}^{\delta})}\cdot\pi_{j_{e_{n}}}(\hat{\tau}_{j})\right)\otimes\pi_{j_{e_{m}}}(\hat{\tau}_{j})\ \vert{\bf s}\rangle
\end{array}
\end{equation}
So finally we have
\begin{equation}\label{modifiedHonstate}
\hat{H}_{T(\delta)}^{\prime\prime}[N]\vert{\bf s}\rangle\ =\ 
N(v)\frac{{\cal A}^{\prime}}{4\delta^{3}}\sum_{m=1}^{N_{v}}\sum_{n=1\vert n\neq m}^{N_{v}}\sum_{I=1}^{4}(-1)^{I+1}\vert{\bf s}_{\delta}(I;v,e_{m},e_{n})\rangle
\end{equation}
As we will not have any occasion to use either $\hat{H}_{T(\delta)}[N]$ or $\hat{H}_{T(\delta)}^{\prime}[N]$ individually in this paper, from now on we will denote $\hat{H}_{T(\delta)}^{\prime\prime}[N]$ simply as $\hat{H}_{T(\delta)}[N]$.\\
It is important to note that, there is yet another possibility to obtain non-trivial continuum commutator even for the case of density $\frac{4}{3}$ weighted constraint without modifying the definition given in eq.(\ref{dec28-5}). This is because, Even though the vertices created by such a $\hat{H}_{T(\delta)}[N]$ are degenerate, there location is dynamical in the sense that it is governed by the quantum shift. As the triad itself transforms non-trivially under the action of $H[N]$, it is possible to modify the definition of quantum constraint such that on the vertices which are obtained by motion along the quantum shift, constraint  acts in such a way so as to take into account a quantization of eq.(\ref{HonE-final}). 
It was precisely this scenario that was realized in \cite{hlt1} in the case of a $U(1)^{3}$ theory. Preliminary investigations indicate that this approach will also lead to an anomaly free definition of the Hamiltonian constraint, \cite{Iinprep}.\\

\section{Defining the constraint on ${\cal L}^{2}({\cal A})$}
At this stage it is not clear to us if $\hat{H}_{T(\delta)}[N]$ is a gauge invariant operator.\footnote{Although the explicit choice of curvature operator in the case of $j=\frac{1}{2}$ holonomy indicates that it maybe so.}
 Whence instead of defining the operator on gauge-invariant spin-nets, we will define the regularized operator on ${\cal L}^{2}(A)$. This even serves a pedagogical purpose as it will be notationally more convenient for us to explicitly write the states which lie in the image of $\hat{H}_{T(\delta)}[N]$ in the functional representation. We will touch upon the issue of gauge invariance of the \emph{continuum} Hamiltonian constraint in section (\ref{conclusions}), however a detailed analysis of this question is left for future work.\\
Given a graph $\gamma$ with a vertex $v$, such that edges $e_{1},\dots,e_{N_{v}}$ are incident at $v$ we will consider gauge-variant spin-networks defined as,
\begin{equation}\label{spinnetdef}
\begin{array}{lll}
\vert{\bf s}\rangle_{m_{e_{1}}\dots m_{e_{N_{v}}}}\ :=\ \left(\pi_{j_{e_{1}}}(h_{e_{1}})\otimes\dots\otimes\pi_{j_{e_{N_{v}}}}(h_{e_{N_{v}}})\right)_{m_{e_{1}}\dots m_{e_{N_{v}}},n_{e_{1}}\dots n_{e_{N_{v}}}}\cdot {\cal M}_{\gamma-\{v\}}^{n_{e_{1}}\dots n_{e_{N_{v}}}}
\end{array}
\end{equation}
where ${\cal M}_{\gamma-\{v\}}$ conceals the data contained in the other vertices of ${\bf s}$. We will always assume that the spin incident at $v$, are such that $
0\ \in\ \vec{j_{e_{1}}}\ +\ \dots\ +\ \vec{j}_{e_{N_{v}}}$.\\

Throughout this paper, we will consider action of the quantum constraint on spin-network vertices which are not gauge-invariant as described above. We will always denote any spin-network state as $\vert{\bf s}\rangle$, however it should be clear from the context that such a state might carry uncontracted magnetic indices.

\section{Structure of deformed states}\label{deformedstates}
In the previous section we quantized the Hamiltonian constraint and obtained a densely defined operator on ${\cal H}_{kin}$. We now analyze the action of this (finite triangulation) operator in some detail by explicitly writing down the resulting states. This will tell us, among other things the intertwiner structure at the new as well as original vertices of the spin-networks. It will also help us understand the conceptual underpinnings of the mechanism through which, in the continuum limit the quantum constraints satisfy Dirac algebra.\\
We recall that the finite triangulation constraint action on a spin-network is a sum of local operators which act on the vertices of the state. Whence we define
\begin{equation}
\hat{H}_{T(\delta)}(v)\vert{\bf s}\rangle\ :=\ \sum_{i,j\vert j\neq i, b(e_{i})=b(e_{j})=v}(-1)^{I+1}\vert{\bf s}_{\delta}(I;v,e_{i}, e_{j})\rangle
\end{equation}

In the connection  representation $\vert{\bf s}\rangle$ is given by
\begin{equation}
T_{{\bf s}}(A)\ =\ \left(\pi_{j_{e_{1}}}(h_{e_{1}})\otimes\dots\otimes\pi_{j_{e_{N_{v}}}}(h_{e_{N_{v}}})\right)_{m_{e_{1}}\dots m_{e_{N_{v}}},n_{e_{1}}\dots n_{e_{N_{v}}}}{\cal M}_{\gamma({\bf s})-\{v\}}^{n_{e_{1}}\dots n_{e_{N_{v}}}}
\end{equation}
Hence the cylindrical functions corresponding to $\vert{\bf s}_{\delta}(1;v,e_{i},e_{j})\rangle\,\ \vert{\bf s}_{\delta}(2;v,e_{i},e_{j})\rangle$
are given by,
\begin{equation}\label{deformedstates-1}
\begin{array}{lll}
T_{{\bf s}_{\delta}(1;v,e_{i},e_{j})}(A)\ =\\
\left(\pi_{j_{e_{1}}}(h_{\alpha(s_{e_{i},\delta}, s_{e_{1},\delta})}\cdot h_{e_{1}})\otimes\dots\otimes\pi_{j_{e_{i}}}(\tau_{i}\cdot h_{e_{i}})\otimes\dots\otimes\pi_{j_{e_{j}}}(\tau_{i}\cdot h_{\alpha(s_{e_{i},\delta},s_{e_{j},\delta})}\cdot h_{e_{j}})\otimes\right.\\
\hspace*{3.3in}\left.\dots\otimes\pi_{j_{e_{N_{v}}}}(h_{\alpha(s_{e_{i},\delta},s_{e_{N_{v}},\delta})}\cdot h_{e_{N_{v}}})\right)_{m_{e_{1}}\dots m_{e_{N_{v}}},n_{e_{1}}\dots n_{e_{N_{v}}}}\\
\hspace*{5.5in}{\cal M}_{\gamma({\bf s})-\{v\}}^{n_{e_{1}}\dots n_{e_{N_{v}}}}\\
=\ \left(\pi_{j_{e_{i}}}(\tau_{i})\otimes\pi_{j_{e_{j}}}(\tau_{i})\right)_{m_{e_{i}}m_{e_{i}}^{\prime},\ m_{e_{j}}m_{e_{j}}^{\prime}}\\
\hspace*{1.0in}\left(\pi_{j_{e_{1}}}(h_{s_{e_{1},{\delta}}})\dots\otimes\pi_{j_{e_{N_{v}}}}(h_{s_{e_{1}},\delta})\right)_{m_{e_{1}}\dots m_{e_{i}}^{\prime}\dots m_{e_{j}}^{\prime}\dots m_{e_{N_{v}}},m_{e_{1}}^{\prime\prime}\dots m_{e_{i}}^{\prime\prime}\dots m_{e_{j}}^{\prime\prime}\dots m_{e_{N_{v}}}^{\prime\prime}}\\
\hspace*{0.7in}\left(\pi_{j_{e_{1}}}(h_{a_{i1}^{\delta}}\cdot h_{e_{1}-s_{e_{1},\delta}})\otimes\dots\otimes\pi_{j_{e_{i}}}(h_{e_{i}-s_{e_{i},\delta}})\otimes\dots\otimes\pi_{j_{e_{j}}}(h_{\phi^{\delta}_{\hat{e}_{i}(0)}\cdot s_{e_{j},\delta}}\cdot h_{e_{j}-s_{e_{j}\delta}})\otimes\right.\\
\hspace*{3.4in}\left.\dots\otimes \pi_{j_{e_{N_{v}}}}(h_{a_{iN_{v}}^{\delta}}\cdot h_{e_{N_{v}}-s_{e_{N_{v}}\delta}})\right)_{m_{e_{1}}^{\prime\prime}\dots m_{e_{N_{v}}}^{\prime\prime},n_{e_{1}}\dots n_{e_{N_{v}}}}\\
\vspace*{0.1in}
=\  \left(\pi_{j_{e_{1}}}(h_{s_{e_{1},{\delta}}})\dots\otimes\pi_{j_{e_{N_{v}}}}(h_{s_{e_{1}},\delta})\right)_{m_{e_{1}}\dots m_{e_{i}}\dots m_{e_{j}}\dots m_{e_{N_{v}}},m_{e_{1}}^{\prime\prime}\dots m_{e_{i}}^{\prime}\dots m_{e_{j}}^{\prime}\dots m_{e_{N_{v}}}^{\prime\prime}}\\
\hspace*{3.5in}\left(\pi_{j_{e_{i}}}(\tau_{i})\otimes\pi_{j_{e_{j}}}(\tau_{i})\right)_{m_{e_{i}}^{\prime}m_{e_{i}}^{\prime\prime},\ m_{e_{j}}^{\prime}m_{e_{j}}^{\prime\prime}}\\
\hspace*{0.7in}\left(\pi_{j_{e_{1}}}(h_{a_{i1}^{\delta}}\cdot h_{e_{1}-s_{e_{1},\delta}})\otimes\dots\otimes\pi_{j_{e_{i}}}(h_{e_{i}-s_{e_{i},\delta}})\otimes\dots\otimes\pi_{j_{e_{j}}}(h_{\phi^{\delta}_{\hat{e}_{i}(0)}\cdot s_{e_{j},\delta}}\cdot h_{e_{j}-s_{e_{j}\delta}})\otimes\right.\\
\hspace*{3.4in}\left.\dots\otimes \pi_{j_{e_{N_{v}}}}(h_{a_{iN_{v}}^{\delta}}\cdot h_{e_{N_{v}}-s_{e_{N_{v}}\delta}})\right)_{m_{e_{1}}^{\prime\prime}\dots m_{e_{N_{v}}}^{\prime\prime},n_{e_{1}}\dots n_{e_{N_{v}}}}\\
\hspace*{5.5in}{\cal M}_{\gamma({\bf s})-\{v\}}^{n_{e_{1}}\dots n_{e_{N_{v}}}}
\end{array}
\end{equation}
In going from first to second line in the above equation we have used,\\
\noindent{\bf (i)} $h_{\alpha(s_{e_{i},\delta},s_{e_{j},\delta})}\ =\ h_{s_{e_{i},\delta}}\circ h_{\phi^{\delta}_{\hat{e}_{i}(0)}\cdot s_{e_{j},\delta}}\circ h_{s_{e_{j},\delta}}^{-1}$, $h_{\alpha(s_{e_{i},\delta},s_{e_{k},\delta})}\ =\ h_{s_{e_{i}}^{\delta}}\circ h_{a_{ik}^{\delta}}\circ h_{s_{e_{k}}^{\delta}}^{-1}\ \forall\ k\ \neq\ j$.\\

\noindent{\bf (ii)} In the last line we have used the following property of \emph{gauge-invariant} intertwiner 
$\pi_{j_{e_{i}}}(\tau_{i})\otimes\pi_{j_{e_{j}}}(\tau_{i})$ in the last line.
\begin{equation}\label{jan1-1}
\begin{array}{lll}
\left(\pi_{j_{e_{i}}}(\tau_{i})\otimes\pi_{j_{e_{j}}}(\tau_{i})\right)_{m_{e_{i}}m_{e_{i}}^{\prime},\ m_{e_{j}}m_{e_{j}}^{\prime}}\\
\vspace*{0.1in}
\hspace*{1.0in}\left(\pi_{j_{e_{1}}}(h_{s_{e_{1},{\delta}}})\dots\otimes\pi_{j_{e_{N_{v}}}}(h_{s_{e_{1}},\delta})\right)_{m_{e_{1}}\dots m_{e_{i}}^{\prime}\dots m_{e_{j}}^{\prime}\dots m_{e_{N_{v}}},m_{e_{1}}^{\prime\prime}\dots m_{e_{i}}^{\prime\prime}\dots m_{e_{j}}^{\prime\prime}\dots m_{e_{N_{v}}}^{\prime\prime}}\ =\\
\left(\pi_{j_{e_{1}}}(h_{s_{e_{1},{\delta}}})\dots\otimes\pi_{j_{e_{N_{v}}}}(h_{s_{e_{1}},\delta})\right)_{m_{e_{1}}\dots m_{e_{i}}\dots m_{e_{j}}\dots m_{e_{N_{v}}},m_{e_{1}}^{\prime\prime}\dots m_{e_{i}}^{\prime}\dots m_{e_{j}}^{\prime}\dots m_{e_{N_{v}}}^{\prime\prime}}\\
\vspace*{0.1in}
\hspace*{2.7in}\left(\pi_{j_{e_{i}}}(\tau_{i})\otimes\pi_{j_{e_{j}}}(\tau_{i})\right)_{m_{e_{i}}^{\prime}m_{e_{i}}^{\prime\prime},\ m_{e_{j}}^{\prime}m_{e_{j}}^{\prime\prime}}
\end{array}
\end{equation}
We thus see that in the deformed state associated to ${\bf s}_{\delta}(1;v,e_{i},e_{j})$ has (with respect to the state $T_{{\bf s}}(A)$) the following properties.\\
\noindent {\bf (a)} As can be clearly seen from the above equation, the segments $s_{e_{m}}\ \forall\ m \neq i\in\{1,\dots,n_{v}\}$ are absent in the resulting state and whence the original vertex $v$ is no longer a  \emph{non-degenerate} vertex for the graph underlying the new state. Instead, the new vertex sits at $v_{e,\delta}$ with the invariant tensor $\left(\hat{\tau}_{i}\vert_{e}\otimes\hat{\tau}_{i}\vert_{e^{\prime}}\right)$. The reader can easily verify in-fact that had we started with a gauge-invariant spin-network state with a vertex $v$, in the resulting state ${\bf s}_{\delta}(1;v, e_{i},e_{j})$ would have no vertex at $v$ but infact a new vertex at $v_{e_{i},\delta}$ with a different intertwiner then the one at $v$. \footnote{The state $\vert {\bf s}_{1}^{\delta}(v,e,e^{\prime})\rangle$ is infact a linear combination of spin-network states obtained by expanding $\pi_{j_{e_{1}}}(h_{s_{e_{1},\delta}})\otimes\dots\otimes\pi_{j_{e_{N_{v}}}}(h_{s_{e_{N_{v}}},\delta})$ in terms of direct sum of irreducible representations.}\\
\noindent{\bf (b)} $f(s_{e_{j}}^{\delta})$ and $f(s_{e_{m}}^{\delta}), m \ \neq\ (i,j)$ are bi-valent vertices. As the diffeomorphisms generated by quantum shift are ``singular", we 
quantify this singular structure by placing certain constraints on differentiability of the edges at such vertices. Whence we choose the loops $\alpha(s_{e_{j}}^{\delta}, s_{e_{m}}^{\delta})$ to be such that, $f(e_{j}^{\delta})$ is a $C^{0}$ ``kink vertex,\footnote{This means that the edge $\phi^{\delta}_{\hat{e}_{i}(0)}(e_{j})$ is only continuous at $f(s_{e_{j}}^{\delta})$.} and for technical convenience we choose the loops $\alpha(s_{e_{j}}^{\delta}, s_{e_{m}}^{\delta})$ such that $f(a_{I'm}^{\delta})$ are $C^{1}$-kinks. That is, the tangent vector to the edge $a_{I'm}^{\delta}\circ(e_{m}-s_{e_{m}}^{\delta})$ exists, but its acceleration is ill-defined.\footnote{This condition will ensure that the key mechanism which ensures off-shell closure is same as the one which would have occurred for the constraint defined in eq.(\ref{dec28-5}).}\\
Note that as $f(e_{j}^{\delta})$ is a unique $C^{0}$ bivalent vertex in a $\delta_{0}$ co-ordinate nbd. of $v$ in ${\bf s}_{\delta}(1;v,e_{i}, e_{j})$, we have,
\begin{equation}
\langle{\bf s}_{\delta}(1;v,e_{i}, e_{j})\vert{\bf s}_{\delta}(1;v,e_{i}, e_{j}^{\prime})\rangle\ =\ 0
\end{equation}
because $\gamma({\bf s}_{\delta}(1;v,e_{i}, e_{j}))\ \neq\ \gamma({\bf s}_{\delta}(1;v,e_{i}, e_{j}^{\prime}))$.\\
In addition to the above two properties which are natural consequence of the definition of the quantum constraint, we put an additional condition on the (unit) tangent vector $\hat{e}_{j}^{\delta}(0)$ to the edge $e_{j}^{\delta}$ at the new non-degenerate vertex $v_{e_{i},\delta}$. Consider $\delta$ to be small enough such that $v_{e_{i},\delta}$ lies inside the co-ordinate patch $\{x\}_{v}$. We then define the deformation generated by $\hat{H}_{T(\delta)}[N]$ to be such that 
\begin{equation}\label{deformedtangent}
\hat{e}_{j}^{\delta}(0)\ =\ \hat{e}_{j}(0). 
\end{equation}
This condition is by the following heuristic interpretation of quantum Shift as an averaging of a vector field over a co-ordinate ball. The anomaly freedom of constraint algebra is not sensitive to such a choice however.
The ``deformed" state $\vert{\bf s}_{\delta}(1;v,e_{i}, e_{j})\rangle$ is shown in the figure below.\\

\begin{figure}
\begin{center}
\includegraphics[height=2in, width=5in]{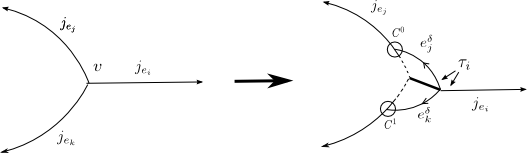}
\caption{${\bf s}_{\delta}(1;v,e_{i},e_{j})$ is shown. For simplicity we have displayed a three-valent vertex. The intertwiner at the new vertex is $\pi_{j_{e_{i}}}(\tau_{i})\otimes\pi_{j_{e_{j}}}(\tau_{i})$. The differentiability of edges at degenerate vertices is displayed explicitly. The edges $e_{j}^{\delta}$ and $e_{k}^{\delta}$ corresponds to $\phi^{\delta}_{\hat{e}_{i}(0)}\cdot e_{j}$ etc.}
\label{fig1}
\end{center}
\end{figure}

We similarly analyze the structure of $\vert{\bf s}_{\delta}(2;v,e_{i},e_{j})\rangle$ in some detail by writing it explicitly in the functional representation.
\begin{equation}\label{deformedstates-2}
\begin{array}{lll}
T_{{\bf s}_{\delta}(2;v,e_{i},e_{j})}(A)\ =\\
\left(\pi_{j_{e_{1}}}(h_{\alpha(s_{e_{i},\delta}, s_{e_{1},\delta})}\cdot h_{e_{1}})\otimes\dots\otimes\pi_{j_{e_{i}}}(\tau_{i}\cdot h_{e_{i}})\otimes\dots\otimes\pi_{j_{e_{j}}}(h_{\alpha(s_{e_{i},\delta},s_{e_{j},\delta})}\cdot\tau_{i}\cdot h_{e_{j}})\otimes\right.\\
\hspace*{3.4in}\left.\dots\otimes\pi_{j_{e_{N_{v}}}}(h_{\alpha(s_{e_{i},\delta},s_{e_{N_{v}},\delta})}\cdot h_{e_{N_{v}}})\right)_{m_{e_{1}}\dots m_{e_{N_{v}}},n_{e_{1}}\dots n_{e_{N_{v}}}}\\
\hspace*{4.8in}{\cal M}^{n_{e_{1}}\dots n_{e_{N_{v}}}}_{\gamma-\{v\}}\\
\vspace*{0.1in}
=\pi_{j_{e_{i}}}(\tau_{i})_{m_{e_{i}}m_{e_{i}}^{\prime\prime}}\cdot\left(\pi_{j_{e_{1}}}(h_{s_{e_{i}}^{\delta}})\otimes\dots\otimes\pi_{j_{e_{N_{v}}}}(h_{s_{e_{i}}^{\delta}})\right)_{m_{e_{1}}\dots m_{e_{i}}^{\prime\prime}\dots m_{e_{N_{v}}};m_{e_{1}}^{\prime}\dots m_{e_{i}}^{\prime}\dots m_{e_{N_{v}}}^{\prime}}\cdot\\
\left(\pi_{j_{e_{1}}}(h_{a_{i1}^{\delta}}\cdot h_{e_{1}-s_{e_{1},\delta}})\otimes\dots\otimes\pi_{j_{e_{i}}}(h_{e_{i}-s_{e_{i},\delta}})\otimes\dots\otimes\pi_{j_{e_{j}}}(h_{\phi^{\delta}_{\hat{e}_{i}(0)}\cdot s_{e_{j},\delta}}\cdot h_{s_{e_{j},\delta}}^{-1}\cdot \tau_{j}\cdot h_{e_{j}})\otimes\right.\\
\hspace*{2.7in}\left.\dots\otimes \pi_{j_{e_{N_{v}}}}(h_{a_{iN_{v}}^{\delta}}\cdot h_{e_{N_{v}}-s_{e_{N_{v}}\delta}})\right)\cdot{\cal M}\\
\vspace*{0.1in}
=\left\{\left(\pi_{j_{e_{1}}}(h_{s_{e_{i}}^{\delta}})\otimes\dots\otimes\pi_{j_{e_{N_{v}}}}(h_{s_{e_{i}}^{\delta}})\right)_{m_{e_{1}}\dots m_{e_{i}}\dots m_{e_{N_{v}}};m_{e_{1}}^{\prime}\dots m_{e_{i}}^{\prime}\dots m_{e_{N_{v}}}^{\prime}}\otimes\pi_{J=1}(h_{s_{e_{i}}}^{\delta})_{ij}\right\}
\pi_{j_{e_{i}}}(\tau_{j})_{m_{e_{i}}^{\prime}m_{e_{i}}^{\prime\prime}}\cdot\\
\hspace*{0.4in}\left(\pi_{j_{e_{1}}}(h_{\phi^{\delta}_{\hat{e}_{i}(0)}\cdot s_{e_{1},\delta}}\cdot h_{e_{1}-s_{e_{1},\delta}})\otimes\dots\otimes\pi_{j_{e_{i}}}(h_{e_{i}-s_{e_{i},\delta}})\otimes\dots\otimes\pi_{j_{e_{j}}}(h_{\phi^{\delta}_{\hat{e}_{i}(0)}\cdot s_{e_{j},\delta}}\cdot h_{s_{e_{j},\delta}}^{-1}\cdot \tau_{j}\cdot h_{e_{j}})\otimes\right.\\
\hspace*{3.0in}\left.\dots\otimes \pi_{j_{e_{N_{v}}}}(h_{\phi^{\delta}_{\hat{e}_{i}(0)}\cdot s_{e_{N_{v}},\delta}}\cdot h_{e_{N_{v}}-s_{e_{N_{v}}\delta}})\right)_{m_{e_{1}}\dots m_{e_{N_{v}}},n_{e_{1}}\dots n_{e_{N_{v}}}}\\
\hspace*{4.7in}{\cal M}^{n_{e_{1}}\dots n_{e_{N_{v}}}}_{\gamma-\{v\}}
\end{array}
\end{equation}
In going from first to second line in eq.(\ref{deformedstates-2}), we have used (i) and (ii) listed below eq.(\ref{deformedstates-1}) and the following property of a spin-1 intertwiner.
\begin{equation}
\begin{array}{lll}
\left({\cal I}\cdot\left(\pi_{j_{e_{i}}}(\tau_{i})\right)\right)^{m_{e_{1}}\dots m_{e_{N_{v}}}}_{v}\left(\pi_{j_{e_{1}}}(h_{s_{e_{i},\delta}})\otimes\dots\otimes\pi_{j_{e_{N_{v}}}}(h_{s_{e_{i},\delta}})\right)_{m_{e_{1}},\dots,m_{e_{N_{v}}}}^{n_{e_{1}}\dots n_{e_{N_{v}}}}\\
\vspace*{0.1in}
\hspace*{1.9in}=\pi_{j=1}(h_{s_{e_{i},\delta}})_{ij}\left({\cal I}\cdot\left(\pi_{j_{e_{i}}}(\tau_{i})\right)\right)^{m_{e_{1}}\dots m_{e_{N_{v}}}}_{v}
\end{array}
\end{equation}
From Eq. (\ref{deformedstates-2}) we see that,\\
\noindent{\bf (i)} $\vert{\bf s}_{\delta}(2;v,e_{i},e_{j})$ has a gauge-variant $N_{v}+1$-valent vertex at $v_{e_{i},\delta}$ with a spin-1 intertwiner\\ $\left(\pi_{j_{e_{i}}}(\tau_{i})\right)$.\\
\noindent{\bf (ii)} There is a bivalent ``spin-1" vertex at $v$ (which was the location of non-degenerate vertex of $s$) with the corresponding spin-1 intertwiner being $\pi_{j_{e_{i}}}(\tau_{i})$.\\
\noindent{\bf (iii)} $\vert{\bf s}_{\delta}(2;v,e_{i},e_{j})$ has $N_{v}-2$ gauge-invariant bivalent kink vertices which are located at $\{v_{e_{k},\delta}\vert k\neq\{i,j\}\}$, and a three-valent vertex located $v_{e_{j},\delta}$. We place the same demand on the differentiability of edges at such vertices as given in \noindent {\bf (b)} below eq.(\ref{jan1-1}). That is, all the bivalent kink vertices are $C^{1}$ and the trivalent kink vertex $f(s_{e_{j}}^{\delta})$ is such that, the edge $\phi^{\delta}_{\hat{e}_{i}(0)}(s_{e_{j}}^{\delta})\circ (e_{j}-s_{e_{j}}^{\delta})$ is continuous but not differentiable at the vertex.\\
Details of this state are easiest to analyze by meditating over the figure \ref{figure2}.\\
\begin{figure}\label{figure2}
\begin{center}
\includegraphics[height=2in, width=5in]{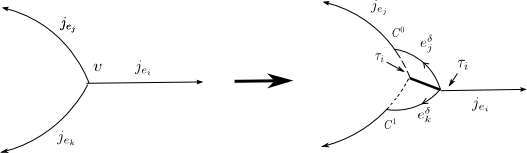}
\caption{${\bf s}_{\delta}(2;v,e_{i},e_{j})$ shown. The insertion operators are displayed at different vertices to contrast with the structure of ${\bf s}_{\delta}(1;v,e_{i},e_{j})$. Thick line represents edge being in a reducible rep. of $SU(2)$ and thick dotted line represents edge being in reducible rep. of $j_{e_{j}}\otimes j_{e_{j}}$}
\end{center}
\end{figure}

Similarly  we can write down the functional representation for remaining two states. The details are pretty similar to the ones for the states $\vert{\bf s}_{\delta}(I,\dots)\rangle_{I\in\{1,2\}}$ so here we simply give the final result.
\begin{equation}
\begin{array}{lll}
T_{{\bf s}_{\delta}(3;v,e_{i},e_{j})}\ =\\
\left(\pi_{j_{e_{i}}}(h_{s_{e_{i},{\delta}}})\otimes\pi_{j_{e_{j}}}(h_{s_{e_{j}}^{\delta}})\right)_{m_{e_{i}}m_{e_{j}},m_{e_{i}}^{\prime}m_{e_{j}}^{\prime}}\\
\hspace*{3.5in}\left(\pi_{j_{e_{i}}}(\tau_{i})\otimes\pi_{j_{e_{j}}}(\tau_{i})\right)_{m_{e_{i}}^{\prime}m_{e_{i}}^{\prime\prime},\ m_{e_{j}}^{\prime}m_{e_{j}}^{\prime\prime}}\\
\hspace*{0.7in}\left(\pi_{j_{e_{1}}}(h_{e_{1}})\otimes\dots\otimes\pi_{j_{e_{i}}}(h_{e_{i}-s_{e_{i},\delta}})\otimes\dots\otimes\pi_{j_{e_{j}}}(h_{\phi^{\delta}_{\hat{e}_{i}(0)}\cdot s_{e_{j},\delta}}\cdot h_{e_{j}-s_{e_{j}\delta}})\otimes\right.\\
\hspace*{3.4in}\left.\dots\otimes \pi_{j_{e_{N_{v}}}}(h_{e_{N_{v}}})\right)_{m_{e_{1}}\dots m_{e_{N_{v}}},n_{e_{1}}\dots n_{e_{N_{v}}}}\\
\hspace*{5.5in}{\cal M}_{\gamma({\bf s})-\{v\}}^{n_{e_{1}}\dots n_{e_{N_{v}}}}
\end{array}
\end{equation}
Figure \ref{figure3} is sufficient to understand the details associated to graph and intertwiner structure of $T_{{\bf s}_{\delta}(3;v,e_{i},e_{j})}$.\\
\begin{figure}\label{figure3}
\begin{center}
\includegraphics[height=2in, width=5in]{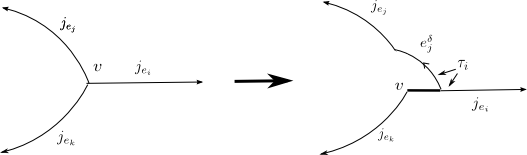}
\caption{${\bf s}_{\delta}(3;v,e_{i},e_{j})$ shown. The valence of $v$ goes down by one, whence it is non-degenerate. The new vertex is degenerate.}
\end{center}
\end{figure}

Similarly a functional representation for $\vert{\bf s}_{\delta}(4;v,e_{i},e_{j})\rangle$ ca be written down. However in light of its similarity with the state associated to ${\bf s}_{\delta}(2;v,e_{i},e_{j})$ we simply provide a figure (figure \ref{figure4}) which 
captures all the relevant details.\\ 
\begin{figure}\label{figure4}
\begin{center}
\includegraphics[height=2in, width=4in]{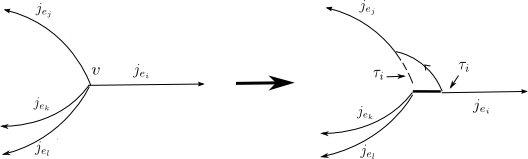}
\caption{${\bf s}_{\delta}(4;v,e_{i},e_{j})$ shown. The valence of $v$ goes down by one, whence it is non-degenerate. The new vertex is degenerate.The dotted thick line along $e_{j}$ is in reducible representation $j_{e_{j}}\otimes j_{e_{j}}$. The thick line along $e_{i}$ is in reducible rep. $j_{e_{i}}\otimes j_{e_{j}}$.}
\end{center}
\end{figure}
We once again stress that $\vert{\bf s}_{\delta}(I;v,\dots)$ are not spin-network states. They are linear combination of gauge-variant spin-nets which could be obtained by expansion of product of various intertwiners in terms of $6-j$ symbols. However our contention is to check if the Hamiltonian constraint proposed above is anomaly free, and as we will see, this question can be asked precisely by working with the states obtained above. \emph{We will refer to such states as cylindrical-networks}.\\

Each of the above states are linear combination of spin-network states based on a graph which has \emph{at most} a degenerate vertex at $v$ and a unique non-degenerate vertex at $v_{e,\delta}$. The grasping operators $\tau_{i}$ can be expressed as (appropriately normalized) 3-j symbols, and can be used to re-write the above states in terms of spin-network basis. 
We now show some rather surprising orthogonality relations of such cylindrical networks which make them well suited for construction of Habitat spaces.\\
\subsection{Orthogonality relations among deformed states.}
In this section we display two orthogonality properties of the cylindrical networks which make them suitable for the construction of Habitat states.\\
We will show that,
\begin{equation}\label{orthoproperty}
\begin{array}{lll}
\langle {\bf s}_{\delta}(I;v,e_{m}, e_{n})\vert{\bf s}_{\delta}(J;v,e_{m^{\prime}},e_{n^{\prime}})\rangle\ =\ 0\ \forall\ I,J\ \textrm{if}\ (m,n)\neq (m^{\prime}, n^{\prime})\\
\vspace*{0.1in}
\langle {\bf s}_{\delta}(I;v,e_{m}, e_{n})\vert{\bf s}_{\delta}(J;v,e_{m^{\prime}},e_{n^{\prime}})\rangle\ =\ 0 
\end{array}
\end{equation}
$\forall\ ((m,n);(m^{\prime},n^{\prime}))$ if $I\neq J$.\\
The first orthogonality follows from the fact that, as long as $(m,n)\ \neq\ (m^{\prime}, n^{\prime})$ (we recall that by this we mean, $m\ \neq\ m^{\prime},\ n\ \neq n^{\prime}$), the graphs underlying the cylindrical networks are necessarily distinct in such a way as to make the two states orthogonal.\\
More in detail, when $m\ \neq\ m^{\prime}$, the states are easily seen to be orthogonal as deformations are along different edges $e_{m}$ or $e_{m^{\prime}}$.
If $m\ =\ m^{\prime}$ but $n\ \neq\ n^{\prime}$ then when $\vert{\bf s}_{\delta}(I;v,e_{m}, e_{n})\rangle$ is written as a linear combination of spin-networks, 
each under-lying spin-network has a (at most trivalent) $C^{0}$-kink vertex at  $f(e_{n}^{\delta})$ and \emph{all the other} bivalent kink vertices are $C^{1}$.\\
Thus the graphs underlying each such spin-networks are distinct from the graph underlying spin-networks associated to ${\bf s}_{\delta}(I;v,e_{m}, e_{n^{\prime}})$. Hence the orthogonality property follows.\\
We now prove that the second orthogonality property holds.\\
In light of the first orthogonality relation, it suffices to prove that 
\begin{equation}
\langle {\bf s}_{\delta}(I;v,e_{m}, e_{n})\vert{\bf s}_{\delta}(J;v,e_{m},e_{n})\rangle\ =\ 0 
\end{equation}
It is of sufficient to analyze $I=1, J\ \in\ \{2,3,4\}$ case, as other cases are analogous. In fact if $J\ \in\ \{3,4\}$ then as the under-lying graphs are distinct, the states are orthogonal to each other.Thus we only need to compute 
$\langle {\bf s}_{\delta}(I=1;v,e_{m}, e_{n})\vert{\bf s}_{\delta}(J=2;v,e_{m},e_{n})\rangle$ to prove the second relation.\\
Note that there exists a canonical decomposition of cylindrical-networks into three cylindrical functions, one for each pair of insertion operators.
In other words, we can write
\begin{equation}
\vert{\bf s}_{\delta}(I;v,e_{m}, e_{n})\rangle\ =\ \sum_{i=1}^{3}\vert{\bf s}_{\delta}( I, i;v,e_{m}, e_{n})\rangle
\end{equation}
with
\begin{equation}
\begin{array}{lll}
\vert{\bf s}_{\delta}( I, i=1;v,e_{m}, e_{n})\rangle\ :=\ \vert{\bf s}_{\delta}( I;v,(e_{m}, +), (e_{n}, -))\rangle\\
\vspace*{0.1in}
\vert{\bf s}_{\delta}( I, i=2;v,e_{m}, e_{n})\rangle\ :=\ \vert{\bf s}_{\delta}( I;v,(e_{m}, -), (e_{n}, +))\rangle\\
\vert{\bf s}_{\delta}( I, i=3;v,e_{m}, e_{n})\rangle\ :=\ \vert{\bf s}_{\delta}( I;v,(e_{m}, 3), (e_{n}, 3))\rangle
\end{array}
\end{equation}
where the labels accompanying the edges correspond to one of the three insertion operators $(\tau_{+}, \tau_{-},\ \textrm{or}\ \tau_{3})$.\footnote{We will not have an occasion to use $\vert{\bf s}_{\delta}( I, i;v,e_{m}, e_{n})\rangle$ after establishing the orthogonality of cylindrical-networks. Whence this notation will never be used anywhere else in the paper.}
We will now show that,
\begin{equation}
\begin{array}{lll}
\langle{\bf s}_{\delta}(I=1; v, (e_{m}, i), (e_{n}, i))\vert  {\bf s}_{\delta}(I=2; v, (e_{m}, j), (e_{n}, j))\rangle\ =\ 0\\
\end{array}
\end{equation}
$\forall\ i,j$.\\
\underline{{\bf Proof}} : Let us look at the state $\vert{\bf s}_{\delta}(I=2; v, (e_{m}, j), (e_{n}, j))\rangle$ in holonomy representation.
\begin{equation}
\begin{array}{lll}
T_{{\bf s}_{\delta}(2;v,e_{m},e_{n}, j)}\ =\\
\vspace*{0.1in}
\left(\pi_{j_{e_{1}}}(h_{s_{e_{m}}^{\delta}}\cdot h_{\alpha(s_{e_{m}}^{\delta}, s_{e_{1}}^{\delta})}\cdot h_{e_{1}-s_{e_{1}}^{\delta}})_{m_{e_{1}}n_{e_{1}}}\otimes\dots\otimes\pi_{j_{e_{m}}}(\tau_{j}\cdot h_{s_{e_{m}}^{\delta}}\cdot h_{e_{m}-s_{e_{m}}^{\delta}})_{m_{e_{m}}n_{e_{m}}}\otimes\dots\otimes\right.\\
\hspace*{2.8in}\pi_{j_{e_{n}}}(h_{s_{e_{m}}^{\delta}}\cdot h_{\alpha(s_{e_{m}}^{\delta},s_{e_{n}}^{\delta})}\cdot h_{s_{e_{n}}^{\delta}}^{-1}\tau_{j}\cdot h_{s_{e_{n}}^{\delta}}\cdot h_{e_{n}-s_{e_{n}}^{\delta}})_{m_{e_{n}}n_{e_{n}}}\otimes\\
\hspace*{3.0in}\dots\otimes \pi_{j_{e_{N_{v}}}}(h_{s_{e_{m}}^{\delta}}\cdot h_{\alpha(s_{e_{m}}^{\delta}, s_{e_{N_{v}}}^{\delta})}\cdot h_{e_{N_{v}}-s_{e_{N_{v}}}^{\delta}})_{m_{e_{N_{v}}}n_{e_{N_{v}}}}
\end{array}
\end{equation}
When the above state is written as a linear combination of spin-networks, the only states which are such that their underlying graph is same as the graph underlying 
one of the spin-networks contained in $\vert{\bf s}_{\delta}(1;v,e_{m},e_{n}, i=1)\rangle$ will contribute non-trivially to $\langle{\bf s}_{\delta}(I=1; v, (e_{m}, i), (e_{n}, i))\vert  {\bf s}_{\delta}(I=2; v, (e_{m}, j), (e_{n}, j))\rangle$. This means that when $T_{{\bf s}_{\delta}(2;v,e_{m},e_{n}, j)}$ is expressed as linear combination of cylindrical networks over fixed graph, we are only interested in those states whose underlying graphs have no $s_{e_{n}}^{\delta}$ segment.\\
Thus consider,
\begin{equation}\label{J=0ins(2)state}
\begin{array}{lll}
\pi_{j_{e_{n}}}(h_{s_{e_{m}}^{\delta}}\cdot h_{\alpha(s_{e_{m}}^{\delta},s_{e_{n}}^{\delta})}\cdot h_{s_{e_{n}}^{\delta}}^{-1}\tau_{j}\cdot h_{s_{e_{n}}^{\delta}}\cdot h_{e_{n}-s_{e_{n}}^{\delta}})_{m_{e_{n}}n_{e_{n}}}\ =\\
(h_{s_{e_{m}}^{\delta}})_{m_{e_{n}}m_{e_{n}}^{\prime}}\cdot (h_{\alpha(s_{e_{m}}^{\delta},s_{e_{n}}^{\delta})})_{m_{e_{n}}^{\prime}m_{e_{n}}^{\prime\prime}}\cdot 
\left((h_{s_{e_{n}}^{\delta}}^{-1})_{m_{e_{n}}^{\prime\prime}k}\otimes\ (h_{s_{e_{n}}})_{l n_{e_{n}}^{\prime}}\otimes(\tau_{j})_{kl}\right)\cdot(h_{e_{n}-s_{e_{n}}^{\delta}})_{n_{e_{n}}^{\prime}n_{e_{n}}} \approx\\
(h_{s_{e_{m}}^{\delta}})_{m_{e_{n}}m_{e_{n}}^{\prime}}\cdot (h_{\alpha(s_{e_{m}}^{\delta},s_{e_{n}}^{\delta})})_{m_{e_{n}}^{\prime}m_{e_{n}}^{\prime\prime}}\cdot\left((-1)^{2j_{e_{n}}} \epsilon_{m_{e_{n}}^{\prime\prime}l}\ \epsilon_{k n_{e_{n}}}\cdot (\tau_{j})_{kl}\right)\cdot (h_{e_{n}-s_{e_{n}}^{\delta}})_{n_{e_{n}}^{\prime}n_{e_{n}}}=\\ 
\pi_{j_{e_{n}}}(h_{s_{e_{m}}^{\delta}})_{m_{e_{n}}m_{e_{n}}^{\prime}}\cdot (h_{\alpha(s_{e_{m}}^{\delta},s_{e_{n}}^{\delta})})_{m_{e_{n}}^{\prime}m_{e_{n}}^{\prime\prime}}\left(\epsilon\cdot \epsilon\cdot\tau_{j}\right)_{m_{e_{n}}^{\prime\prime}n_{e_{n}}^{\prime}}(h_{e_{n}-s_{e_{n}}^{\delta}})_{n_{e_{n}}^{\prime}n_{e_{n}}}
\end{array}
\end{equation}
where $\epsilon\ =\ -\tau_{2}$.\\
In the third line, we expressed the tensor product $(h_{s_{e_{n}}^{\delta}}\otimes (h_{s_{e_{n}}}^{\delta})^{-1})$ in terms of direct sum of irreducible representations and considered only the $J=0$ component as this is the only state which could be non-orthogonal to $\vert{\bf s}_{\delta}(1,\dots)\rangle$.
That is, we have used 
\begin{equation}
\begin{array}{lll}
(h_{s_{e_{n}}^{\delta}}^{-1})_{m_{e_{n}}^{\prime\prime}k}\otimes\ (h_{s_{e_{n}}})_{l n_{e_{n}}^{\prime}}=\\
(-1)^{2_{j_{e_{2}}}}\sum_{J=0}^{2j_{e_{2}}}d_{J}\sum_{K,L}C_{j_{e_{2}}j_{e_{2}}J}^{m_{e_{n}}^{\prime\prime}lK}C_{j_{e_{2}}j_{e_{2}}J}^{kn_{e_{n}}L}\pi_{J}(h_{s_{e_{n}}})_{KL}(-1)^{K-L}\\
\vspace*{0.1in}
=\ (-1)^{2_{j_{e_{2}}}}\epsilon_{m_{e_{n}}^{\prime\prime}l}\epsilon_{kn_{e_{n}}}+\ \sum_{J\neq 0}\dots
\end{array}
\end{equation}
$C_{j_{e_{2}}j_{e_{2}}J}^{\dots}$ are the standard Clebsch-Gordon co-efficients.\\
However, even in this ($J=0$) case we see that certainly there is a non-trivial intertwiner placed at the bivalent vertex $f(\alpha(s_{e_{m}}^{\delta},s_{e_{n}}^{\delta}))$ for all $j\ \in\ \{1,2\}$. Whence this state is orthogonal to $\vert{\bf s}_{\delta}(I=1; v, (e_{m}, i), (e_{n}, i))\rangle\ \forall\ i$.\\
Thus it only remains to be shown that 
\begin{equation}
\langle{\bf s}_{\delta}(I=1; v, (e_{m}, i=3), (e_{n}, i=3))\vert  {\bf s}_{\delta}(I=2; v, (e_{m}, j=3), (e_{n}, j=3))\rangle\ =\ 0
\end{equation}
Using eq.(\ref{J=0ins(2)state}) and a similar functional representation for the $\vert{\bf s}_{\delta}(1;e_{m},e_{n},i)\rangle$ state, the inner product reduces it to 
\begin{equation}
\begin{array}{lll}
\langle{\bf s}_{\delta}(I=1; v, (e_{m}, i=3), (e_{n}, i=3))\vert  {\bf s}_{\delta}(I=2; v, (e_{m}, j=3), (e_{n}, j=3))\rangle\\
\vspace*{0.1in}
\equiv\ \sum_{m_{e_{n}}^{\prime\prime}}m_{e_{n}}^{\prime\prime}\left[\int dh_{\alpha}\left(\pi_{j_{e_{n}}}(h_{\alpha(s_{e_{m}}^{\delta},s_{e_{n}}^{\delta})})_{p_{e_{n}}^{\prime}p_{e_{n}}^{\prime\prime}})^{\star} \pi_{j_{e_{n}}}(h_{\alpha(s_{e_{m}}^{\delta},s_{e_{n}}^{\delta})})_{m_{e_{n}}^{\prime}m_{e_{n}}^{\prime\prime}})\right)\right.\\
\hspace*{0.5in}\left.\int dh_{e_{n}-s_{e_{n}}^{\delta}}\left(\pi_{j_{e_{n}}}(h_{e_{n}-s_{e_{n}}^{\delta}})_{p_{e_{n}}^{\prime\prime}p}^{\star}\pi_{j_{e_{n}}}(h_{e_{n}-s_{e_{n}}^{\delta}})_{n_{e_{n}}^{\prime}n_{e_{n}}}\right)\right]
\end{array}
\end{equation}
Using Peter-Weyl and $\sum_{m=-j}^{j}m\ =\ 0$, one can show that 
\begin{equation}
\begin{array}{lll}
\langle{\bf s}_{\delta}(I=1; v, (e_{m}, i=3), (e_{n}, i=3))\vert  {\bf s}_{\delta}(I=2; v, (e_{m}, j=3), (e_{n}, j=3))\rangle\ =\ 0
\end{array}
\end{equation}
This finishes the proof.
\subsection{Action of two Hamiltonian constraints.}
We now compute the action of two successive finite triangulation Hamiltonian constraints on $\vert{\bf s}\rangle$ and analyze the structure of the resulting states in detail. As we saw above, action of Hamiltonian constraint is a sum of local operators, each of which involves a pair of edges which begin at $v$. The resulting cylindrical functions are based on graphs which have certain special kind of kink vertices on the edges involved. We would like to understand just such structures in detail when two constraints act on $\vert{\bf s}\rangle$. These structures will be crucial when we analyze the continuum limit of the commutator.\\

Once again, we quickly recall that given a spin-net state $\vert{\bf s}\rangle$ with a single non-degenerate vertex $v$, action of finite triangulation Hamiltonian constraint is given by,
\begin{equation}
\begin{array}{lll}
\delta\hat{H}_{T(\delta)}[N]\vert{\bf s}\rangle\ =\ N(x(v))\hat{H}_{T(\delta)}(v)\vert{\bf s}\rangle\\
=\ \frac{{\cal A}}{\delta}N(x(v))\sum_{I=1}^{4}\sum_{i=1}^{N_{v}}\sum_{e_{j}\vert j\neq i, j=1}^{N_{v}}(-1)^{I+1}\vert{\bf s}_{\delta}(I;v,e_{i},e_{j})\rangle
\end{array}
\end{equation}
where ${\cal A}=\frac{{\cal A}^{\prime}}{2}\ =\ \frac{3}{8\pi}$ is the numerical factor which results from the regularization of quantum shift and is independent of the state ${\bf s}$.\\
The action of two successive Hamiltonian constraints on a spin-network state $\vert{\bf s}\rangle$ can now be written as,
\begin{equation}
\begin{array}{lll}
\hat{H}_{T(\delta^{\prime})}[M]\ \hat{H}_{T(\delta)}[N]\ \vert{\bf s}\rangle\ =\\
\vspace*{0.1in}
\frac{{\cal A}\ N(v)}{\delta}\hat{H}_{T(\delta^{\prime})}[M]\ \sum_{i,j=1\vert i\neq j}^{N_{v}({\bf s})}\sum_{I=1}^{4}(-1)^{I+1}\vert{\bf s}_{\delta}(I;v,e_{i},e_{j})\rangle=\\
\vspace*{0.1in}
\sum_{i=1}^{N_{v}({\bf s})}{\cal A}^{2}\frac{M(v_{e_{i},\delta})\ N(v)}{\delta\delta^{\prime}}\\
\hspace*{0.4in}\sum_{j=1\vert j\neq i}^{N_{v}({\bf s})}\sum_{J=1}^{4}\sum_{I=1}^{2}(-1)^{(I+J)\textrm{mod}2}\sum_{i^{\prime}, j^{\prime}=1\vert i^{\prime}\neq j^{\prime}}^{N_{v_{e_{i},\delta}}({\bf s}_{\delta}(I,e_{i},e_{j}))}
\vert({\bf s}_{\delta}(I;v,e_{i},e_{j}))_{\delta^{\prime}}(J;v_{e_{i},\delta}, e_{i^{\prime}}^{\delta}, e_{j^{\prime}}^{\delta})\rangle &&\\
\vspace*{0.1in}
&+& \dots 
\end{array}
\end{equation}
We have suppressed all the terms whose Lapse dependence is given by $M(v)N(v)$ as such terms will vanish under anti-symmetrization. Such terms are indicated by $\dots$. More precisely  these terms are of the following type.\\
{\bf (1)} They contain states obtained by action of second Hamiltonian on $\vert{\bf s}_{\delta}(I;v,e_{i},e_{j})\rangle$ for $I\ \in\ \{3,4\}$.\\ 
{\bf (2)} States that arise from the action of second Hamiltonian constraint operator localized at the original vertex $v$, when it acts on $\vert{\bf s}_{\delta}(2;v,e_{i},e_{j})\rangle$.\footnote{These vertices are bi-valent in ${\bf s}_{\delta}(2,v,\dots)$ and do not vanish under the action of constraint of density weight two.}\\
As far as density two Hamiltonian constraint is concerned, there will also be additional contributions coming from 
action of second Hamiltonian on the bivalent kink vertices. These contributions would not be ultra-local in Lapse functions (Lapse dependence would be of the form $\left(N(v)M(f(s_{e_{n}}^{\delta})\ -\ N\leftrightarrow M\right)$ and hence would not vanish simply due to anti-symmetrization of Lapses. \emph{However such contributions would be absent in the case of density $\frac{4}{3}$ constraint for any physically reasonable quantization of inverse volume functional. Whence we do not consider such terms here. We note that this is where we are taking a key input from the computation we would actually like to do, that involving density $\frac{4}{3}$ constraint.}\footnote{The definition of habitat states considered here is such that, the states resulting from the action of second Hamiltonian on bivalent kink-vertices are orthogonal to all the cylindrical-networks in $[{\bf s}]$. Whence such states are not relevant in the continuum limit. However the proof of this statement is rather involved, and hence we take recourse to the assumption stated above.}
Had we worked with the Hamiltonian constraint with density weight $\frac{4}{3}$, then due to the bi-valent nature of such vertices, second Hamiltonian will plausibly annihilate such vertices due to the action of $\hat{q}^{\frac{1}{3}}$ operator. In any case, as such contributions are irrelevant for the commutator, we will neglect them from now on.\\
Whence the action of finite triangulation commutator on $\vert{\bf s}\rangle$ is given by,
\begin{equation}\label{oct26-1}
\begin{array}{lll}
[\ \hat{H}_{T(\delta^{\prime})}[M],\ \hat{H}_{T(\delta)}[N]\ ]\ \vert{\bf s}\rangle\ =\\
\vspace*{0.1in}
\sum_{i=1}^{N_{v}({\bf s})}{\cal A}^{2}\frac{1}{\delta\delta^{\prime}}\left(M(v_{e_{i},\delta})\ N(v)\ -\ M\leftrightarrow N\right)\\
\hspace*{0.4in}\sum_{j=1\vert j\neq i}^{N_{v}({\bf s})}\sum_{I=1}^{2}\sum_{J=1}^{4}(-1)^{(I+J)\textrm{mod}2}\sum_{i^{\prime}, j^{\prime}=1\vert i^{\prime}\neq j^{\prime}}^{N_{v_{e_{i},\delta}}({\bf s}_{\delta}(I,e_{i},e_{j}))}
\vert({\bf s}_{\delta}(I;v,e_{i},e_{j}))_{\delta^{\prime}}(J;v_{e_{i},\delta}, e_{i^{\prime}}^{\delta}, e_{j^{\prime}}^{\delta})\rangle
\end{array}
\end{equation}
We would now like to build our habitat states by starting with $\vert{\bf s}\rangle$ and summing over all singly-deformed states $\vert{\bf s}_{\delta}(I;v,e_{i},e_{j})\rangle$ and all doubly deformed states\\ $\vert({\bf s}_{\delta}(I;v,e_{i},e_{j}))_{\delta^{\prime}}(J;v_{e_{i},\delta};e_{i^{\prime}}^{\delta},e_{j^{\prime}}^{\delta})\rangle\ \forall\ I,J;(i,j);(i^{\prime},j^{\prime})$. However instead of taking into account all such possible states, each of our Habitat state will be constructed only by summing over all singly deformed states and certain specific class of doubly deformed states.
\subsection{Relevant doubly deformed states.} 
States obtained by action of two finite triangulation constraints can be classified into following three types.\\
\noindent{{\bf type (a)}} : The action of two Hamiltonian constraints involve the (ordered) pair $(e_{i_{0}}, e_{j_{0}}\vert 1\leq i_{0}\neq j_{0}\leq N_{v})$ , $(e_{i_{0} \delta}, e_{j_{0} \delta})$ respectively.\\
\noindent {{\bf type (b)}} : The action of two Hamiltonian constraints involve the (ordered) pair $(e_{i_{0}}, e_{j_{0}}\vert 1\leq i_{0}\neq j_{0}\leq N_{v})$ , $(e_{j_{0} \delta}, e_{i_{0} \delta})$ respectively.\\
\noindent {{\bf type (c)}} :  The action of two Hamiltonian constraints involve the (ordered) pair $(e_{i}, e_{j}\vert 1\leq i \neq j\leq N_{v})$ , $(e_{k \delta}, e_{l\delta}\vert 1\leq k\neq l\leq N_{v_{e_{i},\delta}})$, where at least one of $(k,\ l)$ is not among $\{i,j\}$.\\
Note that terms of the third type would have been absent had we worked with the Hamiltonian constraint defined in eq.(\ref{dec28-5}). There is also an intuitive reason why terms of type ({\bf c}) are ``undesirable" from the point of view of off-shell closure. As quantization of RHS would show, there are no cylindrical-network states obtained by the action of $\widehat{RHS}$ on $\vert{\bf s}\rangle$ whose intertwiner structure matches with the intertwiner structure of cylindrical-networks of type-{\bf \(c\)}. Whence our choice of the Habitat ${\cal V}_{hab}$ will be such that,
if $\Psi\ \in\ {\cal V}_{hab}$ then,
\begin{equation}\label{jan4-1}
\begin{array}{lll}
\lim_{\delta\rightarrow 0}\lim_{\delta^{\prime}\rightarrow 0}\\
\hspace*{0.5in}(\Psi\vert\ \left(\frac{{\cal A}M(v_{e,\delta})}{\delta^{\prime}} \frac{{\cal A} N(v)}{\delta}\ -\ N\leftrightarrow M\right)\sum_{I,J=1}^{2}(-1)^{(I+J)\textrm{mod}2}\vert({\bf s}_{\delta}(I;v,e,e^{\prime}))_{\delta^{\prime}}(J;v_{e,\delta}, e_{\delta}, e_{\delta^{\prime}})\rangle\ \neq\ 0
\end{array}
\end{equation}
only if the pairs of edges $(e, e^{\prime})$, $(e_{\delta}, e^{\prime}_{\delta})$ belong to {\bf case(a)} or {\bf case(b)} and is zero otherwise.\\
We put one further technical restriction on the class of doubly deformed states that we would use to construct the Habitat. This restriction can be relaxed, but the analysis would be far more involved. Perhaps more importantly, this restriction would be borne out naturally, had we to begin with started with a gauge-invariant vertex.\\
Recall that the action of finite triangulation Hamiltonian constraint on $\vert{\bf s}\rangle$ is given by
\begin{equation}
\hat{H}_{T(\delta)}(v)\vert{\bf s}\rangle\ =\ \sum_{m,n=1\vert m\neq n}^{N_{v}}\left[\vert{\bf s}_{\delta}(1;v,e_{m},e_{n})\rangle\ -\ \vert{\bf s}_{\delta}(2;v,e_{m},e_{n})\rangle\ +\ \sum_{I=3}^{4}(-1)^{I+1}\vert{\bf s}_{\delta}(I;v,e_{m},e_{n})\rangle\right]
\end{equation}
For any ordered pair of edges $(e_{m},e_{n})$, both the states $\vert{\bf s}_{\delta}(1;v,e_{m},e_{n})\rangle,\ \vert{\bf s}_{\delta}(2;v,e_{m},e_{n})\rangle$ has a unique non-degenerate vertex $v_{e_{m},\delta}$ which is $N_{v}+1$-valent and has, in addition to the edges $\{e_{i}^{\delta}\}_{i=1,\dots,N_{v}}$, an edge $s_{e_{m}}^{\delta}$ which is in the reducible representation $j_{e_{1}}\otimes\dots\otimes j_{e_{N_{v}}}$ of $SU(2)$. As $0\ \in\ \sum_{i=1}^{N_{v}}\vec{j}_{e_{i}}$, such a reducible representation can be decomposed into irreps labelled by $SU(2)$-spin $J$, a subset of which correspond to $J=0$. Each element of this subset is obtained by considering all possible recouping schemes for $j_{e_{1}},\dots,j_{e_{N_{v}}}$ such that the total spin adds up to zero. Whence, schematically we can write 
\begin{equation}
\begin{array}{lll}
\vert{\bf s}_{\delta}(\alpha\in\{1,2\};v,e_{m},e_{n})\rangle\ =\ \sum_{{\cal I}}\vert{\bf s}_{\delta}(\alpha\in\{1,2\};v,e_{m},e_{n};J_{s_{e_{m}^{\delta}}}=0,{\cal I})\rangle\ +\\
\vspace*{0.1in}
\hspace*{2.6in}\vert{\bf s}_{\delta}(\alpha\in\{1,2\};v,e_{m},e_{n};J_{s_{e_{m}^{\delta}}}\ \neq\ 0)\rangle
\end{array}
\end{equation}
Given an intertwiner ${\cal I}$, the doubly deformed states which are obtained by action of second Hamiltonian on 
$\vert{\bf s}_{\delta}(\alpha\in\{1,2\};v,e_{m},e_{n};J_{s_{e_{m}^{\delta}}}=0,{\cal I})\rangle$ such that they belong to {\bf type-(a)}, {\bf type-(b)} are given by the set,
\begin{equation}
{\cal Z}\ :=\ \left\{\vert({\bf s}_{\delta}(I;v,e_{i},e_{j},{\cal I}))_{\delta^{\prime}}(J;v_{e_{i},\delta}, e_{i}^{\delta}, e_{j}^{\delta})\rangle,\ \vert({\bf s}_{\delta}(I;v,e_{i},e_{j},{\cal I}))_{\delta^{\prime}}(J;v_{e_{i},\delta}, e_{j}^{\delta}, e_{i}^{\delta})\rangle\ \vert \forall\ (i,j),\forall J,\ I\ \in\ \{1,2\}\right\}
\end{equation}
We will use all the singly deformed states and the doubly deformed states which belong to ${\cal Z}$ to construct the Habitat. 

\section{Continuum limit of the Hamiltonian constraint}\label{con-limit}
In this section, we introduce a linear space of Habitat states with respect to  which the finite triangulation Hamiltonian admits a continuum limit in a operator topology that we specify below. We refer the reader to seminal works \cite{lm, lm2} where the ideas of introducing Habitats first arose, and to \cite{hlt1, tv} for more recent works, where similar ideas have been applied in the context of a $U(1)^{3}$ gauge theory. The work here leans heavily on the habitat structures introduced in \cite{hlt1}, \cite{tv} and we encourage the reader to refer those works for more details. It is important to note that such Habitats are essentially  kinematical arenas on which one can meaningfully ask questions associated to continuum limit of quantum constraints. This in turn implies that questions related to anomalies in the constraint algebra can be asked in such a setting. Generically there is no universal definition of a habitat for all the constraints in LQG and as such there is no canonical inner product on Habitats (for one counterexample see \cite{hlt2}). For a detailed critique of Habitat, we refer the reader to \cite{ttbook}.\\
The basic idea underlying notion of continuum limit on a subspace of space of distributions is rather simple. We say that the net $\delta\hat{H}_{T(\delta)}[N]$ converges to $\hat{H}[N]$ which is a linear operator from ${\cal V}_{hab}\rightarrow\ {\cal D}^{*}$  if, 
\begin{equation}
\lim_{\delta\rightarrow 0} \Psi\left(\delta\hat{H}_{T(\delta)}[N]\right)\vert{\bf s}\rangle\ =:\ \Psi^{\prime}\vert{\bf s}\rangle 
\end{equation}
$\forall\ \vert{\bf s}\rangle,\ \Psi\in {\cal V}_{hab}$, and  where the range state $\Psi^{\prime}\in {\cal D}^{*}$.\\
The topology on the space of operators in which $\hat{H}_{T(\delta)}[N]$ converges to $\hat{H}[N]$ is explained in section 3.2 of \cite{tv}.\\
We similarly compute the continuum limit of $\hat{H}_{T(\delta^{\prime})}[N]\hat{H}_{T(\delta)}[M]\ -\ M\leftrightarrow N$ and $\hat{D}_{T(\delta^{\prime})}[N]\hat{D}_{T(\delta)}[M]\ -\ M\leftrightarrow N$.
\begin{equation}\label{def-contlimit}
\begin{array}{lll}
\lim_{\delta^{\prime}\rightarrow 0}\lim_{\delta\rightarrow 0}\Psi\left(\hat{H}_{T(\delta^{\prime})}[N]\hat{H}_{T(\delta)}[M]\ -\ M\leftrightarrow N\right)\vert{\bf s}^{\prime}\rangle\ =\ \textrm{LHS}(\Psi,{\bf s}^{\prime})\\
\vspace*{0.1in}
(-3)\lim_{\delta^{\prime}\rightarrow 0}\lim_{\delta\rightarrow 0}\Psi\left(\hat{D}_{T(\delta^{\prime})}[N]\hat{D}_{T(\delta)}[M]\ -\ M\leftrightarrow N\right)\vert{\bf s}^{\prime}\rangle\ =\ \textrm{RHS}(\Psi,{\bf s}^{\prime})\\
\end{array}
\end{equation}
\emph{Our goal in this paper is to show that, for the finite triangulation Hamiltonian constraint defined in this paper, there exists a choice of $V_{hab}$ for which both $\textrm{LHS}, \textrm{RHS}$ are non-trivial (bi-linear) functions on ${\cal V}_{hab}, {\cal D}$ and they are in fact equal.}

We now introduce the space of (distributional) states ${\cal V}_{hab}$ on which these ideas are carried out.\\
Recall that ${\bf s}$ is our canonical  \emph{gauge-variant} spin-network state, with a single non-degenerate vertex $v$. We will define a set ${\cal S}$ of cylindrical-networks, such that the sum over this set will generate a class of Habitat states.\\
All the cylindrical networks which are used from now on will be denoted as ${\bf c}$.
\subsection{Construction of ${\cal S}$}
Given an intertwiner ${\cal I} : j_{e_{1}}\otimes\dots\otimes j_{e_{N_{v}}}\rightarrow\ 0$, define 
\begin{equation}\label{varioussetsofs}
\begin{array}{lll}
[{\bf s}]_{(I=1)}^{(i,j);{\cal I}}\ =\ \{{\bf c}^{\prime}\vert {\bf c}^{\prime} = {\bf s}_{\delta}(1;v,e_{i}, e_{j};J_{s_{e_{i}^{\delta}}}=0,{\cal I}),\ \forall\ \delta\leq \delta_{0}({\bf s})\}\\
\vspace*{0.1in}
[{\bf s}]_{(I=2)}^{(i,j),{\cal I}}\ =\ \{{\bf c}^{\prime}\vert {\bf c}^{\prime} = {\bf s}_{\delta}(2;v,e_{i}, e_{j};J_{s_{e_{i}}^{\delta}}=0,{\cal I}),\ \forall\ \delta\leq \delta_{0}({\bf s})\}\\
\vspace*{0.1in}
[{\bf s}]_{(I=3)}^{(i,j)}\ =\ \{{\bf c}^{\prime}\vert {\bf c}^{\prime} = {\bf s}(3;v,e_{i}, e_{j})\}\\
\vspace*{0.1in}
[{\bf s}]_{(I=4)}^{(i,j)}\ =\ \{{\bf c}^{\prime}\vert {\bf c}^{\prime} = {\bf s}_{\delta}(4;v,e_{i}, e_{j}),\ \forall\ \delta\leq \delta_{0}({\bf s})\}\\
\vspace*{0.1in}
[{\bf s}]_{I,J}^{(1) (i,j) {\cal I}}\ =\\
 \{{\bf c}^{\prime}\vert {\bf s}^{\prime} = ({\bf s}_{\delta}(I;v,e_{i}, e_{j};J_{s_{e_{i}}^{\delta}}=0,{\cal I}))_{\delta^{\prime}}(J; v_{e_{i},\delta}, e_{j \delta}, e_{j \delta}),\ 1\leq i\neq j\leq N_{v},\ \forall\ \delta\leq \delta_{0}({\bf s})\},\\
\hspace*{3.0in}\forall \delta^{\prime} \leq \delta^{\prime}_{0}(\delta_{0}), 1\leq I\ \leq 2,\ 1\leq J\ \leq 4\}\\
\vspace*{0.1in}
[{\bf s}]_{I,J}^{(2) (i,j) {\cal I}}\ =\\
 \{{\bf c}^{\prime}\vert {\bf s}^{\prime} = ({\bf s}_{\delta}(I;v,e_{i}, e_{j};J_{s_{e_{i}}^{\delta}}=0,{\cal I}))_{\delta^{\prime}}(J; v_{e_{j},\delta}, e_{i \delta}, e_{j \delta}),\ 1\leq i\neq j\leq N_{v},\ \forall\ \delta\leq \delta_{0}({\bf s})\},\\
 \hspace*{3.0in} \forall \delta^{\prime} \leq \delta^{\prime}_{0}(\delta_{0}), 1\leq I\leq 2,\ 1\ \leq J\leq 4\}
\end{array}
\end{equation}
As the dependence of sets defined above on the interwiner ${\cal I}$ is obvious, we will drop the superscript ${\cal I}$ from now on.\\
A couple of relevant unions of the sets defined above deserve an additional symbol.
\begin{equation}
[{\bf s}]^{(i,j)}\ =\ \cup_{I}[{\bf s}]_{(I)}^{(i,j)}\\
\end{equation}
All those  cylindrical-network states which are obtained by the action of finite triangulation Hamiltonian on ${\bf s}$ with insertion operators acting on order pair $(e_{i}, e_{j})$ belong to this set.
\begin{equation}
[{\bf s}]_{I,J}^{(m)}\ :=\ \cup_{(i,j)}[{\bf s}]_{I,J}^{(m) (i,j)}\\
\end{equation}
This set contains the set of all cylindrical-network states obtained by action of two Hamiltonians such that for $m=1$ if first Hamiltonian acts on $(e_{i},e_{j})$, then the second constraint acts on $(e_{i}^{\delta}, e_{j}^{\delta})$, and if $m=2$
corresponds to second  Hamiltonian acting on $(e_{j}^{\delta}, e_{i}^{\delta})$.\\
Finally let,
\begin{equation}
[{\bf s}]_{I,J}\ =\ \cup_{m}[{\bf s}]_{I,J}^{(m)}\ I\in\{1,2\},\ J\in\{1,\dots,4\}
\end{equation}
We now define a set of cylindrical networks which will generate the Habitat states.
\begin{equation}\label{defofcalS}
\begin{array}{lll}
{\cal S}({\bf s}, {\cal I})\ :=\ \cup_{I\in\{1,\dots,4\}}\cup_{1\leq i\neq j\neq  N_{v}}[{\bf s}]_{(I)}^{i,j}\ \bigcup\ \cup_{(I,J)}\cup_{i,j}\cup_{m}[{\bf s}]_{(I,J)}^{i,j; (m)}
\end{array}
\end{equation}
where we have re-instated the explicit dependence on ${\cal I}$ in the definition of ${\cal S}$.\\
The various sets introduced above rather nicely divide the states into mutually orthogonal sets. These orthogonality relations are best encapsulated by introducing
the following equivalence relation between cylindrical-networks.
\begin{displaymath}
{\bf c}\ \sim\ {\bf c}^{\prime}
\end{displaymath}
if for exactly one ordered pair $(i,j)$ precisely one of the following conditions is satisfied.
\begin{equation}
\begin{array}{lll}
\textrm{(a)}\ \ {\bf c},\ {\bf c}^{\prime}\ \in\ \cup_{I=1}^{2}[{\bf s}]_{I}^{(i,j)}\ \textrm{or}\\
\vspace*{0.1in}
\textrm{(b)}\ \ {\bf c},\ {\bf c}^{\prime}\ \in\ \cup_{i,j}[{\bf s}]_{3}^{(i,j)}\ \textrm{or}\\
\vspace*{0.1in}
\textrm{(c)}\ {\bf c},\ {\bf c}^{\prime}\ \in\ [{\bf s}]_{4}^{(i,j)}\ \textrm{or}\\
\textrm{(d)}\ {\bf c}, {\bf c}^{\prime}\ \in\ \cup_{J=1}^{4}\cup_{I=1}^{2}[{\bf s}]^{(m) (i,j)}_{I,J}
\end{array}
\end{equation}
A simple inspection should convince the reader that these conditions are mutually exclusive.\\
\begin{lemma}\label{orthogonality-3}
Let ${\bf c}, {\bf c}^{\prime}\ \in\ {\cal S}({\bf s}, {\cal I})$  then $\langle{\bf c}\vert{\bf c}^{\prime}\rangle\ \neq\ 0$ if and only if ${\bf c}\ \sim\ {\bf c}^{\prime}$.
\end{lemma}
\begin{proof}
Consider ${\bf c}\ \in \cup_{I=1}^{2}[{\bf s}]_{I}^{(i,j)}$ and let ${\bf c}^{\prime}\ \in\ {\cal S}({\bf s})\ -\ \cup_{I=1}^{2}[{\bf s}]_{I}^{(i,j)}$. This means that either ${\bf c}^{\prime}\ \in \cup_{I=1}^{2}[{\bf s}]_{I}^{(k,l)}$ where the (ordered) pair $(k,l)\ \neq\ (i,j)$ or ${\bf c}^{\prime}$ is a cylindrical-network of type ( (b), ( c ), (d) ) listed above.\\
In the latter case, as the graphs underlying ${\bf c}$ and ${\bf c}^{\prime}$ are different, the two states are orthogonal. Hence let, ${\bf c}^{\prime}\ \in\ \cup_{I=1}^{2}[{\bf s}]_{I}^{(k,l)}$. if $k\ \neq\ i$ then again as the underlying graphs are obtained by deforming the graph underlying ${\bf s}$ along different edges, the resulting cylindrical networks are orthogonal as can be easily verified.
Finally let us suppose that $k = i,\ \textrm{but} l\ \neq j$ then, again the graphs $\gamma({\bf c}),\ \gamma({\bf c}^{\prime})$ are distinct, as the location of $C^{0}$-kink vertex is distinct for both graphs. Hence even in this case the resulting states are orthogonal.\\
A similar analysis can be done when ${\bf c}\ \in\ {\cal S}({\bf s})\ -\ \cup_{I=1}^{2}[{\bf s}]_{I}^{(i,j)},\ {\bf c}^{\prime}\ \sim\ {\bf c}$. We do not give details for the remaining cases here, as the analysis is rather straightforward and without further subtleties.
\end{proof}
We now display two additional properties of the cylindrical networks underlying habitat states. These features are especially relevant for the commutator computation and is independent of the orthogonality of ``inequivalent" cylindrical networks that was proved in the lemma above.\\

\begin{lemma}\label{kinklemma}
Let ${\bf c}\ \in\ {\cal S}({\bf s}, {\cal I})$. Then, $\langle{\bf c}\vert{\bf s}_{\delta}(I;v,e_{i},e_{j})_{\delta^{\prime}}(J;v_{e_{i},\delta},e_{k}^{\delta},e_{l}^{\delta^{\prime}})\rangle\ =\ 0$ if as unordered pairs, $(i,j)\ \neq\ (k,l)$.\\
\end{lemma}
\begin{proof}
 Consider first $k\ \notin\ (i,j)$ and  $l=i$. Then final non-degenerate vertex would be placed at a co-ordinate distance $O(\delta^{\prime})$ away from $v_{e_{i},\delta}$ along $e_{k}^{\delta}$. As no ${\bf c}^{\prime}$ is based on a graph with such a placement of final non-degenerate vertex, whence the result follows.\\
On the other hand let, $k\ \in\ (i,j)$ but $l\ \notin\ (i,j)$.  Then the ``kink-structure" associated to $\vert{\bf s}_{\delta}(I;v,e_{i},e_{j})_{\delta^{\prime}}(J;v_{e_{i},\delta},e_{k}^{\delta},e_{l}^{\delta^{\prime}})\rangle$ is as follows.
There are two $C^{0}$ kink-vertices at $f(s_{e_{j}}^{\delta}),\ f(s_{e_{l}^{\delta}}^{\delta^{\prime}})$. However for $l\ \notin\ (i,j)$ there is no ${\bf c}\ \in\ {\cal S}({\bf s},{\cal I})$ whose one of the two $C^{0}$-kinks is  at $f(s_{e_{l}^{\delta}}^{\delta^{\prime}})$. This completes the proof.
\end{proof}
And finally,
\begin{lemma}\label{orthogonal-intertwiner}
If ${\bf c}\ \in\ {\cal S}({\bf s},{\cal I})$ and ${\bf c}^{\prime}\ \in\ {\cal S}({\bf s},{\cal I}^{\prime})$, then $\langle{\bf c}\vert{\bf c}^{\prime}\rangle\ =\ 0$ as long as ${\cal I}\ \neq\ {\cal I}^{\prime}$.
\end{lemma}
\begin{proof}
 As we know from eq.(\ref{orthoproperty}),
\begin{equation}
\langle{\bf s}_{\delta}(I;v,e_{m},e_{n};J_{s_{e_{m}}^{\delta}}=0,{\cal I})\vert 
{\bf s}_{\delta}(I^{\prime};v,e_{m^{\prime}},e_{n^{\prime}};J_{s_{e_{m^{\prime}}}^{\delta}}=0,{\cal I}^{\prime})\rangle\ \propto\ \delta_{I,I^{\prime}}\delta_{(m,n);(m^{\prime},n^{\prime})}
\end{equation}
it suffices to show that 
\begin{equation}
\langle{\bf s}_{\delta}(I;v,e_{m},e_{n};J_{s_{e_{m}}^{\delta}}=0,{\cal I})\vert 
{\bf s}_{\delta}(I;v,e_{m},e_{n};J_{s_{e_{m}}^{\delta}}=0,{\cal I}^{\prime})\rangle\ =\ 0
\end{equation}
for some $I,\ (m,n)$.\\
Consider $I=1$. Due to the insertion operator $\pi_{j_{e_{m}}}(\tau_{i})\otimes\pi_{j_{e_{n}}}(\tau_{i})$ ``placed at" $v_{e_{m},\delta}$, the state ${\bf s}_{\delta}(I;v,e_{m},e_{n};J_{s_{e_{m}}^{\delta}}=0,{\cal I}^{\prime})\rangle$ is a linear combination of spin-networks. The most succinct way to understand the resulting states is to consider the presentation of the intertwiner structure as given e.g. in \cite{thiemannalesci}, where the insertion operator $\pi_{j_{e_{m}}}(\tau_{i})\otimes\pi_{j_{e_{n}}}(\tau_{i})$ is presented via a ``fiducial spin-1 edge" between $e_{m},\ e_{n}$ with the two insertions $\pi_{j_{e_{m}}}(\tau_{i})$, $\pi_{j_{e_{n}}}(\tau_{i})$ being the two intertwines at the resulting (fiducial) trivalent vertices. However such an insertion clearly does not change the inter-twining data ${\cal I}^{\prime}$ that is already present at $v_{e_{m},\delta}$, whence the states are orthogonal as they would be in the absence of the insertion operators due to ${\cal I}\ \neq\ {\cal I}^{\prime}$. Hence the proof.
\end{proof}
\subsection{The Habitat space}\label{habitat}
We will now use the set ${\cal S}({\bf s},{\cal I})$ to define the habitat states such that the Hamiltonian constraint as well as the commutator of two (finite triangulation) constraints has a well defined continuum limit on such states. The key idea behind the habitat is given in \cite{lm}. However as our habitat is a non-trivial generalization of the Habitat spaces considered in the literature so far, we revisit the essential idea behind the defining such spaces.\\
The essential idea behind construction of Habitat is the following. We first recall that Habitats are subspaces of ${\cal D}^{\star}$. Suppose we wanted to evaluate continuum limit of a finite (one parameter family of) triangulation operators $\hat{O}_{T(\delta)}$ on ${\cal D}^{\star}$.\footnote{In the interest of pedagogy, we will not expand upon the precise definition of topology on the space of operators that we employ here etc.}
Starting with a spin-network ${\bf s}$ we consider all the ``deformed" spin-nets $\vert{\bf s}(\hat{O},\delta)\rangle$ that are obtained by the action of $\hat{O}_{T(\delta)}$. The states $\vert{\bf s}(\hat{O},\delta)\rangle$ have distinct graphs as well as distinct edge (and vertex) labels with respect to the original state. Each state in the Habitat is first and foremost a distribution obtained by summing over all the ``deformed spin-nets" $\vert{\bf s}(\hat{O},\delta)\rangle$. the coefficients of this linear sum should be chosen so as to ensure that the continuum limit of $\hat{O}_{T(\delta)}$ is well defined in the sense that if $\Psi$ is a state in the Habitat then
$\lim_{\delta\rightarrow 0}\Psi(\hat{O}_{T(\delta)}\vert{\bf s}^{\prime}\rangle)$ is well defined. As far as the toy models considered in the literature so far, it has been sufficient to choose the coefficients as obtained by starting with some function $f$ and evaluating it on the vertex set of $\vert{\bf s}(\hat{O},\delta)\rangle$. That is, formally habitat states look like
\begin{equation}
\Psi^{f} =\ \sum f(V({\bf s}(\hat{O},\delta)))\langle{\bf s}(\hat{O},\delta)\vert 
\end{equation}
where $\sum$ is over all ${\bf s}(\hat{O},\delta)$  and $V({\bf s}(\hat{O},\delta))$ is a vertex set associated to the deformed state and subsequently $f$ is a function on $\Sigma^{V({\bf s})}$. We will generalize this definition as follows.\\
To each deformed state $\vert{\bf s}(\hat{H},\delta)\rangle$, we will not only associate a unique vertex set $V({\bf s}(\hat{H},\delta)$, but also a tensor $\lambda({\bf s}(\hat{H},\delta))$ in the tensor algebra $T(su(2))$ over $su(2)$. Each Habitat state will be defined formally as,
\begin{equation}\label{formaldefofhabitat}
\Psi^{f,{\cal P}} =\ \sum {\cal P}\left(\lambda({\bf s}(\hat{H},\delta))\right)
f(V({\bf s}(\hat{H},\delta)))\langle{\bf s}(\hat{H},\delta)\vert 
\end{equation}
where $f$ is a function on the set of vertices and ${\cal P}$ is a function on $T(su(2))$ which associates to any element of $T(su(2))$ a fixed number by writing the element as some $M\times N$ matrix and evaluating a fixed matrix element. (Details will become clear as we work out the definitions explicitly).\\

We will now work out the above ideas in detail and obtain the Habitat.\\
Given $v,w,\ \in\ su(2)$, and two irreducible representations $j, j^{\prime}$ we will need to look at tensor (Kronecker) product of matrices of the type $\pi_{j}(v)\otimes\pi_{j^{\prime}}(w)$. These Matrices are $(2j+1)\cdot(2j^{\prime}+1)\times (2j+1)\cdot(2j^{\prime}+1)$ dimensional. We now associate elements of such matrices to any ${\bf c}\ \in\ {\cal S}({\bf s},{\cal I})$ as follows. (An assignment of such matrix elements for cylindrical networks belonging to ${\cal T}({\bf s},{\cal I})$ will be completed in section (\ref{RHS-section}))\\
As we saw in the previous section, an ordered pair of edges $(e_{m}, e_{n})\ \in E(\gamma({\bf s}))$ is uniquely associated to each ${\bf c}\ \in\ {\cal S}({\bf s},{\cal I})$, such that in fact if ${\bf c}\ \sim\ {\bf c}^{\prime}$ then the ordered pair associated to them is the same.\\
 To each $[{\bf c}]\ =\ \{{\bf c}^{\prime}\vert {\bf c}^{\prime}\ \sim\ {\bf c}\}$ we assign a matrix $\lambda([{\bf c}])\ :=\ \pi_{j_{e_{m}}}(v([{\bf c}]))\otimes\pi_{j_{e_{n}}}(w([{\bf c}]))$ with $v([{\bf c}]),\ w([{\bf c}])\in su(2)$ and their explicit form is given in section (\ref{labels-for-LHS}). 
\begin{equation}
\lambda({\bf c}^{\prime})\ :=\ \lambda([{\bf c}])
\end{equation}
$\forall\ {\bf c}^{\prime}\ \in\ [{\bf c}]$.\\ 
We will refer to $\lambda({\bf c})$ as vertex tensors. Their definition is motivated  by first analyzing the change in the intertwiner structure at $v$ in a  ``short distance limit" of the finite triangulation commutator. We refer to this limit as a Naive continuum limit, as it is the limit that would be true precisely if the inner product on the state space was a continuous inner product and not a singular inner product like the one on ${\cal H}_{kin}$. The ``naive continuum limit" computation as well as the definition of vertex tensors motivated by it are deferred to the next two sections.\\ 
 Continuing with the construction of Habitat, we note that, as each ${\bf c}$ has a unique pair $(e_{m},e_{n})$ associated to it, we can associate to each $[{\bf c}]$ a pair $(M_{m},N_{n})\ \in\ \{(1,\dots,(2j_{e_{m}}+1)\cdot(2j_{e_{n}}+1)\}$ such that $M_{m}\ \neq\ N_{n}\ \forall\ (e_{m},e_{n})$.\footnote{The Quantization of RHS presented here only gives a matching between LHS and RHS for such ``off-diagonal" entries. We defer the analysis of the diagonal case for later work.}

 The function ${\cal P}$ 
 abstractly introduced in  eq.(\ref{formaldefofhabitat}) is defined as 
\begin{equation}
{\cal P}(\lambda[{\bf c}])\ =\ \lambda([{\bf c}])_{M_{m}N_{n}}
\end{equation} 
We now define a finite set of matrix elements as
\begin{equation}
\begin{array}{lll}
\Lambda_{\vec{MN}}\ =\ \{\lambda([{\bf c}])_{(M_{m}N_{n})}\}
\end{array}
\end{equation}
 where we have chosen once and for all, a pair  $(M_{m},N_{n})$ for each $[{\bf c}], {\bf c}\ \in\ {\cal S}({\bf s},{\cal I})$.  For each such set of matrix elements, we obtain a habitat state defined as follows.
\begin{equation}\label{def-hab}
\Psi^{{\cal F}}_{[{\bf s}]}\ =\ \sum_{{\bf c}\in {\cal S}({\bf s})}\kappa({\bf c}){\cal F}({\bf c})\langle{\bf c}\vert\ +\ \frac{1}{12}\sum_{{\bf c}^{\prime}\in {\cal T}({\bf s})}\kappa({\bf c}^{\prime}){\cal F}({\bf c}^{\prime})\langle{\bf c^{\prime}}\vert
\end{equation}
where,
\noindent{\bf (i)} 
\begin{equation}
\begin{array}{lll}\label{defofkappac}
\kappa({\bf c})^{-1}\ =\ \vert\vert{\bf c}\vert\vert^{2}\\
\textrm{if}\ {\bf c}\ \in\ \cup_{I=1}^{2}\cup_{(i,j)}\left([{\bf s}]_{(I)}^{(i,j);{\cal I}}\ \cup\ \cup_{J=1}^{2}\left([{\bf s}]_{I,J}^{(1) (i,j);{\cal I}}\cup\ [{\bf s}]_{I,J}^{(2) (i,j);{\cal I}}\right)\right)\\
\vspace*{0.1in}
\textrm{and}\\
\vspace*{0.1in}
\kappa({\bf c})^{-1}\ =\ -\vert\vert{\bf c}\vert\vert^{2}\\
\textrm{if}\ {\bf c}\ \in\ \cup_{I=3}^{4}\cup_{(i,j)}\left([{\bf s}]_{(I)}^{(i,j);{\cal I}}\ \cup\ \cup_{J=3}^{4}\left([{\bf s}]_{I,J}^{(1) (i,j);{\cal I}}\cup\ [{\bf s}]_{I,J}^{(2) (i,j);{\cal I}}\right)\right)
\end{array}
\end{equation}
where $[{\bf s}]_{\dots}^{\dots}$ are defined in eq.(\ref{varioussetsofs}). The reason for choosing a relative minus sign for the coefficients $\kappa({\bf c})$  of various cylindrical-networks in a habitat state is purely because, thanks to this minus sign, as we will see, we obtain a well-defined continuum limit of $\hat{H}_{T(\delta)}[N]$ as well as that of finite triangulation commutator. To the best of our understanding, there is no other independent physical reason for this choice and as we insisted before, this is an unsatisfactory aspect of our work that needs a better resolution. In other words, we believe that the fact that $\hat{H}_{T(\delta)}^{\prime}[N]$ as defined in eq.(\ref{dec30-1}) does not have a well defined continuum limit is an issue tied to our choice of the habitat which we are unable to improve upon at this stage.\\
\noindent{\bf (ii)} ${\cal F}({\bf c})$ is defined as follows.\\
\begin{equation}\label{def-hab1}
\begin{array}{lll}
{\cal F}({\bf c})\ =\ f(v({\bf c}))\lambda({\bf c})_{M_{m}N_{n}}\ \vert\ \lambda({\bf c})_{M_{m}N_{n}}=\lambda([{\bf c}])_{M_{m}N_{n}}\ \in\ \Lambda_{\vec{MN}}
\end{array}
\end{equation}
where $v({\bf c})$ is the \emph{unique} non-degenerate vertex associated to ${\bf c}$. The function $f$ is a $C^{\infty}$ function on $\Sigma$.
Remaining ingredients of eq.(\ref{def-hab}) i.e. the set of cylindrical networks ${\cal T}({\bf s},{\cal I})$ is defined in section (\ref{setforRHS}). For now we merely forecast that the set ${\cal T}({\bf s},{\cal I})$ is orthogonal to the set ${\cal S}({\bf s},{\cal I})$.\\
We now show that in order to compute continuum limit of commutator on the Habitat, it suffices to look at $\Psi^{{\cal F}}_{[{\bf s}]} [\ \hat{H}_{T(\delta^{\prime})}[M],\ \hat{H}_{T(\delta)}[N]\ ]\ \vert{\bf s}\rangle$.
\begin{lemma}\label{habitat-only-on-s}
\begin{equation}
\begin{array}{lll}
\Psi^{{\cal F}}_{[{\bf s}]} [\ \hat{H}_{T(\delta^{\prime})}[M],\ \hat{H}_{T(\delta)}[N]\ ]\ \vert{\bf s}^{\prime}\rangle\ \neq 0
\end{array}
\end{equation}
iff ${\bf s}\ =\ {\bf s}^{\prime}$.
\end{lemma}
\begin{proof}
 Without loss of generality, we can focus on the case where both ${\bf s},\ {\bf s}^{\prime}$ have only one non-degenerate vertex $v$.
Consider a co-ordinate ball of radius $\epsilon\ >\ \delta_{0}$ around the vertex $v$. The spin-network states $\vert{\bf s}\rangle,\ \vert{\bf s}^{\prime}\rangle$ can be written as (open) spin-network states $\vert{\bf s}_{v}\rangle$, $\vert{\bf s}^{\prime}\rangle_{v}$ appropriately contracted with states $\vert{\bf s}\rangle_{V(\gamma)-v}$ not adjacent to $v$ (which lie outside the ball $B(v,\epsilon)$) (Such decomposition of spin-network states can be found in e.g. \cite{brunemann-rideout}).\\
 clearly if $\gamma({\bf s}_{V({\bf s})-v})\ \neq\ \gamma({\bf s}_{V({\bf s}^{\prime})-v})$ then 
\begin{displaymath}
\Psi^{{\cal F}}_{[{\bf s}]}[\ \hat{H}_{T(\delta^{\prime})}[M],\ \hat{H}_{T(\delta)}[N]\ ]\ \vert{\bf s}^{\prime}\rangle\ = 0,
\end{displaymath}
as the finite triangulation commutator at most affects $\vert {\bf s}_{v}\rangle,\ \vert{\bf s}_{v}^{\prime}\rangle$.\\
Now suppose the graphs $\gamma({\bf s}_{v})\ \neq\ \gamma({\bf s}^{\prime}_{v})$. 
Which implies there is at least one edge $e\ \in E({\bf s}_{v})$ which does not completely overlap with any edge in $E({\bf s}^{\prime}_{v})$. Let the only edges for which are distinct from each other in $E({\bf s}_{v})$ and $E({\bf s}^{\prime}_{v})$ are $e, e^{\prime}$ respectively.\\
Then we claim that in the limit $\delta, \delta^{\prime}\rightarrow\ 0$, there does not exist a  ${\bf c}\in\ [{\bf s}]$ such that ${\bf c}\ =\ {\bf s}^{\prime}_{\delta}(I, e_{m}, e_{n})_{\delta^{\prime}}(J, e_{m^{\prime}, \delta}, e_{n^{\prime}, \delta})$ for any $I, J, m, n, m^{\prime}, n^{\prime}$. This is rather easy to see. As in this limit, there is always a segment in $e^{\prime}$ which is ``left untouched" by the action of Hamiltonian constraint. Existence of such a segment in ${\bf \tilde{s}}\ =\ {\bf s}^{\prime}_{\delta}(I, e_{m}, e_{n})_{\delta^{\prime}}(J, e_{m^{\prime}, \delta}, e_{n^{\prime}, \delta})$ ensures that such a state will be orthogonal to all ${\bf c}\ \in\ [{\bf s}]$.
\end{proof}

\section{``Short distance analysis" of intertwiner dynamics}\label{naivcont}
In this section we analyze the commutator of two finite triangulation Hamiltonian constraints in the states produced by the action of product of two Hamiltonians in the explicit ``connection" representation. This will tell us,\\
\noindent{\bf (i)} How the structure of commutator at finite triangulation is closely tied to the old computation done by brugmann in Loop representation and,\\
\noindent{\bf (ii)} Meditating over the structure of computation provides  key hint which motivates the definition of vertex tensors.\\
We will compute the finite triangulation commutator for the {\bf case (b)} as defined in the main text. That is, if the action of the first Hamiltonian involves the (ordered) pair $(e_{m}, e_{n})$ then the action of second Hamiltonian involves the ``flipped pair" $(e_{n}^{\delta}, e_{m}^{\delta})$.\\

\subsection{Finite triangulation commutator again}
As explained above, our goal is to analyze the state\\ $\sum_{I,J=1}^{2}(-1)^{I+J\textrm{mod 2}}{\bf s}_{\delta}(I;v,e_{m},e_{n}))_{\delta^{\prime}}(J;v_{e_{m},\delta}, e_{n\ \delta}, e_{m\ \delta})$ in connection representation.
\begin{equation}
\begin{array}{lll}
\sum_{I,J=1}^{2}(-1)^{I+J\textrm{mod 2}}{\bf s}_{\delta}(I;v,e_{m},e_{j}))_{\delta^{\prime}}(J;v_{e_{i},\delta}, e_{j\ \delta}, e_{i\ \delta})\ =\\
\vert{\bf s}_{\delta}(1;v,e_{i},e_{j}))_{\delta^{\prime}}(1;v_{e_{i},\delta}, e_{j\ \delta}, e_{i\ \delta})\rangle - \vert{\bf s}_{\delta}(2;v,e_{i},e_{j}))_{\delta^{\prime}}(1;v_{e_{i},\delta}, e_{j\ \delta}, e_{i\ \delta})\rangle\\
\vspace*{0.1in}
\hspace*{1.4in}-\vert{\bf s}_{\delta}(1;v,e_{i},e_{j}))_{\delta^{\prime}}(2;v_{e_{i},\delta}, e_{j\ \delta}, e_{i\ \delta})\rangle\ + \vert{\bf s}_{\delta}(2;v,e_{i},e_{j}))_{\delta^{\prime}}(2;v_{e_{i},\delta}, e_{j\ \delta}, e_{i\ \delta})\rangle
\end{array}
\end{equation}
We first write all four states in the right hand side of the above equation in connection representation.\\
Note that in the functional representation,
\begin{equation}\label{actionofHinfunctional}
\begin{array}{lll}
\vert{\bf s}_{\delta}(1;v,e_{m},e_{n}, J_{s_{e_{m}}}^{\delta}=0;{\cal I})\rangle\ =\
\left(\dots\otimes\dots\otimes\pi_{j_{e_{m}}}(\tau_{i}\cdot h_{e_{m}-s_{e_{m}}})\otimes\dots\otimes\pi_{j_{e_{n}}}(\tau_{i}\cdot h_{\phi^{\delta}_{\hat{e}_{m}}(0)}\cdot e_{n})\otimes\dots\right)\\
\vspace*{0.1in}
\vert{\bf s}_{\delta}(2;v,e_{m},e_{n}, J_{s_{e_{m}}}^{\delta}=0;{\cal I})\rangle\ =\\
\left(\dots\otimes\dots\otimes\pi_{j_{e_{m}}}(\tau_{i}\cdot h_{e_{m}-s_{e_{m}}})\otimes\dots\otimes\pi_{j_{e_{n}}}(h_{\phi^{\delta}_{\hat{e}_{m}(0)}\cdot s_{e_{n}}}\cdot\tau_{i}\cdot h_{e_{n}-s_{e_{n}}})\otimes\dots\right)
\end{array}
\end{equation}
In the above equations we have only indicated the deformed edges corresponding to $e_{m},\ e_{n}$ as they are the only ones on which insertion operators are acting. The complete expression for the state in the functional form is not relevant for us here.\\
It is now straightforward to deduce the functional representation of four states\\ $\vert{\bf s}_{\delta}(I;v,e_{m},e_{n};{\cal I})_{\delta^{\prime}}(J;v_{e_{m},\delta},e_{n}^{\delta}, e_{m}^{\delta})\rangle$ by applying Hamiltonian constraint on states in eq.(\ref{actionofHinfunctional}).\\
We finally get,
\begin{equation}
\begin{array}{lll}
\vert{\bf s}_{\delta}(1;v,e_{m},e_{n},J_{s_{e_{m}}^{\delta}=0},{\cal I})_{\delta^{\prime}}(1;v_{e_{m},\delta}, e_{n\ \delta}, e_{m\ \delta})\rangle\ =\\ 
{\cal I}\cdot\left(\dots\otimes\pi_{j_{e_{m}}}(\tau_{j}\cdot\tau_{i}\cdot h_{\phi^{\delta^{\prime}}_{\hat{e_{n}^{\delta}}(0)}(e_{m}-s_{e_{m}}^{\delta})})\otimes\dots\otimes\pi_{j_{e_{n}}}(\tau_{j}\cdot\tau_{i}\cdot h_{\phi^{\delta}_{\hat{e}_{m}(0)}(e_{n})-s_{e_{n}^{\delta}}^{\delta^{\prime}}})\otimes\dots\right)
\end{array}
\end{equation}
Similarly, the functional form for $\vert{\bf s}_{\delta}(I;v,e_{m},e_{n},J_{s_{e_{m}}^{\delta}=0},{\cal I})_{\delta^{\prime}}(J;v_{e_{m},\delta}, e_{n\ \delta}, e_{m\ \delta})\rangle,\ \{I,J\ \neq\ 1\}$ are given by,
\begin{equation}
\begin{array}{lll}
\vert{\bf s}_{\delta}(1;v,e_{m},e_{n},J_{s_{e_{m}}^{\delta}=0},{\cal I})_{\delta^{\prime}}(2;v_{e_{m},\delta}, e_{n\ \delta}, e_{m\ \delta})\rangle\ =\\
\vspace*{0.1in}
{\cal I}\cdot\left(\dots\otimes\pi_{j_{e_{m}}}(\tau_{i}\cdot h_{\phi^{\delta^{\prime}}_{\hat{e_{n}}^{\delta}(0)}(s_{e_{m}}^{\delta^{\prime}})}\cdot h_{s_{e_{m}}^{\delta^{\prime}}}^{-1}\cdot\tau_{j}\cdot h_{s_{e_{m}}^{\delta^{\prime}}})\otimes\dots\otimes\pi_{j_{e_{n}}}(\tau_{j}\cdot\tau_{i}\cdot h_{\phi^{\delta}_{\hat{e}_{m}(0)}(e_{n})-s_{e_{n}^{\delta}}^{\delta^{\prime}}})\otimes\dots\right)
\end{array}
\end{equation}
\begin{equation}
\begin{array}{lll}
\vert{\bf s}_{\delta}(2;v,e_{m},e_{n},J_{s_{e_{m}}^{\delta}=0},{\cal I})_{\delta^{\prime}}(1;v_{e_{m},\delta}, e_{n\ \delta}, e_{m\ \delta})\rangle=\\
\vspace*{0.1in}
=\ {\cal I}\cdot\left(\dots\otimes\pi_{j_{e_{m}}}(\tau_{j}\cdot\tau_{i}\cdot h_{\phi^{\delta^{\prime}}_{\hat{e_{n}^{\delta}}(0)}(e_{m}-s_{e_{m}}^{\delta})}\otimes\dots\otimes\pi_{j_{e_{n}}}(\tau_{j}\cdot h_{\phi^{\delta}_{\hat{e}_{m}(0)}(e_{n})-s_{e_{n}^{\delta}}^{\delta^{\prime}}}\cdot h_{s_{e_{n}}^{\delta}}^{-1}\cdot\tau_{i}\cdot h_{e_{n}})\otimes\dots\right)
\end{array}
\end{equation}
\begin{equation}
\begin{array}{lll}
\vert{\bf s}_{\delta}(2;v,e_{m},e_{n},J_{s_{e_{m}}^{\delta}=0},{\cal I})_{\delta^{\prime}}(2;v_{e_{m},\delta}, e_{n\ \delta}, e_{m\ \delta})\rangle=\\
\vspace*{0.1in}
{\cal I}\cdot\left(\dots\otimes\pi_{j_{e_{m}}}(\tau_{i}\cdot h_{\phi^{\delta^{\prime}}_{\hat{e_{n}}^{\delta}(0)}(s_{e_{m}}^{\delta^{\prime}})}\cdot h_{s_{e_{m}}^{\delta^{\prime}}}^{-1}\cdot\tau_{j}\cdot h_{s_{e_{m}}^{\delta^{\prime}}})\otimes\dots\otimes\pi_{j_{e_{n}}}(\tau_{j}\cdot h_{\phi^{\delta}_{\hat{e}_{m}(0)}(e_{n})-s_{e_{n}^{\delta}}^{\delta^{\prime}}}\cdot h_{s_{e_{n}}^{\delta}}^{-1}\cdot\tau_{i}\cdot h_{e_{n}})\otimes\dots\right)
\end{array}
\end{equation}
The functional form for remaining states which arise in this commutator computation,\\ $\vert{\bf s}_{\delta}(I;v,e_{m},e_{n},J_{s_{e_{m}}^{\delta}=0},{\cal I})_{\delta^{\prime}}(J;v_{e_{m},\delta}, e_{n\ \delta}, e_{m\ \delta})\rangle,\ J\in\ \{3,4\}$ can be written down similarly.\\
We now consider the nature of insertion operators when a ``naive" continuum limit is performed. In the limit, when $\delta^{\prime},\delta\ \rightarrow\ 0$.\\
the leading non-vanishing term is obtained by replacing e.g. the holonomy $ h_{\phi^{\delta^{\prime}}_{\hat{e_{n}^{\delta}}(0)}(e_{m}-s_{e_{m}}^{\delta})}$ by $\sim\  \textrm{O}(\delta\delta^{\prime})A^{i}(v)\tau_{i}$. Whence the naive continuum limit will lead to the following structure of insertion operator for each of the four states
\begin{equation}
\begin{array}{lll}
\vert{\bf s}_{\delta}(1;v,e_{m},e_{n},J_{s_{e_{m}}^{\delta}=0},{\cal I})_{\delta^{\prime}}(1;v_{e_{m},\delta}, e_{n\ \delta}, e_{m\ \delta})\rangle\equiv\ \pi_{j_{e_{m}}}(\tau_{i_{1}}\cdot\tau_{j}\cdot\tau_{k})\otimes\ \pi_{j_{e_{n}}}(\tau_{i_{2}}\cdot\tau_{j}\cdot\tau_{k})\\
\vspace*{0.1in}
\vert{\bf s}_{\delta}(1;v,e_{m},e_{n},J_{s_{e_{m}}^{\delta}=0},{\cal I})_{\delta^{\prime}}(2;v_{e_{m},\delta}, e_{n\ \delta}, e_{m\ \delta})\rangle\equiv\ \pi_{j_{e_{m}}}(\tau_{k}\cdot\tau_{i_{1}}\cdot\tau_{j})\otimes\ \pi_{j_{e_{n}}}(\tau_{i_{2}}\cdot\tau_{j}\cdot\tau_{k})\\
\vspace*{0.1in}
\vert{\bf s}_{\delta}(2;v,e_{m},e_{n},J_{s_{e_{m}}^{\delta}=0},{\cal I})_{\delta^{\prime}}(1;v_{e_{m},\delta}, e_{n\ \delta}, e_{m\ \delta})\rangle\equiv\ \pi_{j_{e_{m}}}(\tau_{i_{1}}\cdot\tau_{j}\cdot\tau_{k})\otimes\ \pi_{j_{e_{n}}}(\tau_{j}\cdot\tau_{i_{2}}\cdot\tau_{k})\\
\vspace*{0.1in}
\vert{\bf s}_{\delta}(2;v,e_{m},e_{n},J_{s_{e_{m}}^{\delta}=0},{\cal I})_{\delta^{\prime}}(2;v_{e_{m},\delta}, e_{n\ \delta}, e_{m\ \delta})\rangle\equiv\ \pi_{j_{e_{m}}}(\tau_{k}\cdot\tau_{i_{1}}\cdot\tau_{j})\otimes\ \pi_{j_{e_{n}}}(\tau_{j}\cdot\tau_{i_{2}}\cdot\tau_{k})\\
\end{array}
\end{equation}
where, in the naive continuum limit computation insertion matrices $\tau_{i_{1}},\ \tau_{i_{2}}$ are contracted with $A^{i_{1}},\ A^{i_{2}}$ respectively.\\
In fact, such a continuum limit of a single Hamiltonian constraint leads to the following structure of insertion operators in the leading term.
\begin{equation}\label{dec25-3}
\begin{array}{lll}
\vert{\bf s}_{\delta}(1;v,e_{m},e_{n})\rangle\ \equiv\ \left(\pi_{j_{e_{m}}}(\tau_{j}\cdot\tau_{i_{1}})\otimes\pi_{j_{e_{n}}}(\tau_{j}\cdot\tau_{i_{2}})\right)\\
\vspace*{0.1in}
\vert{\bf s}_{\delta}(2;v,e_{m},e_{n})\rangle\ \equiv\ \left(\pi_{j_{e_{m}}}(\tau_{j}\cdot\tau_{i_{1}})\otimes\pi_{j_{e_{n}}}(\tau_{i_{2}}\cdot\tau_{j})\right)
\end{array}
\end{equation}
where in the ``naive continuum limit", $\tau_{i_{1}},\tau_{i_{2}}$ will be dotted with $A_{i_{1}},\ A_{i_{2}}$ respectively.\\
Based on the above observation, we define the ``vertex smooth functions" for the habitat states which in addition to the dependence on placement of vertex as in the case of Lewandowski-Marolf habitat, precisely captures the information contained in the ordering of insertion operators. This is achieved by first introducing vertex tensors $\lambda({\bf c})$  for all the cylindrical-networks which are in the image of finite triangulation Hamiltonian constraint.

\subsection{Vertex tensors associated to ${\cal S}({\bf s},{\cal I})$}\label{labels-for-LHS}

We now define $\lambda\ =\ \lambda([{\bf c}])$ for all $[{\bf c}]\ \in\ {\cal S}({\bf s},{\cal I})$ as follows.\\
\noindent {\bf (i)} : ${\bf\tilde{s}}\in [{\bf s}]_{(1)}\cup [{\bf s}]_{(3)}$
\begin{equation}\label{22/10-1}
\begin{array}{lll}
\lambda({\bf\tilde{s}})_{MN}\ =\ \sum_{i,j}\left(\pi_{j_{e_{m}}}(\tau_{i}\cdot\tau_{j})\otimes\pi_{j_{e_{n}}}(\tau_{i}\cdot\tau_{j})\right)\\
\hspace*{2.9in}\textrm{if}\ {\bf\tilde{s}}=s_{\delta}(I=\{1,3\}; v, e_{m}, e_{n})
\end{array}
\end{equation}
The ordering of insertion operators is motivated by looking at the structures arising in naive continuum limit (in this case, eq.(\ref{dec25-3})). However as we do not have any smooth gauge fields contracting free indices $i_{1}, i_{2}$, our vertex tensors are defined by choosing $i_{1}=i_{2}$ and summing over $i_{2}$.\footnote{It is possible to start with a vertex tensor which is in fact a map from the tensor algebra over $su(2)$ to a $su(2)^{\star}\otimes su(2)^{\star}$ as 
$\lambda({\bf\tilde{s}})^{i_{1}i_{2}}\ :=\ \sum_{j}\left(\pi_{j_{e_{m}}}(\tau_{i_{1}}\cdot\tau_{j})\otimes\pi_{j_{e_{n}}}(\tau_{i_{2}}\cdot\tau_{j})\right)$. However as this is our first stab at construction of a Habitat which  permits computation of constraint algebra, we leave more refined constructions for future work.}
Similarly, once again taking a cue from the second line in eq.(\ref{dec25-3})
we have for,\\
\noindent {\bf (ii)} : ${\bf\tilde{s}}\in [{\bf s}]_{(2)}\cup [{\bf s}]_{(4)}$
\begin{equation}\label{22/10-2}
\begin{array}{lll}
\lambda({\bf\tilde{s}})\ =\ \sum_{i,j}\left(\pi_{j_{e_{m}}}(\tau_{i}\cdot\tau_{j})\otimes\pi_{j_{e_{n}}}(\tau_{j}\cdot\tau_{i})\right)\\
\hspace*{2.9in}\textrm{if}\ {\bf\tilde{s}}=s_{\delta}(I=\{2,4\}; v, e_{m}, e_{n})
\end{array}
\end{equation}
where the placement of $\tau_{j}$ in the tensor product is different as compared to 
eq.(\ref{22/10-1}).\\
\noindent{\bf (iv)} : 
${\bf\tilde{s}}\in [{\bf s}]^{(1)}_{1,1}$
\begin{equation}\label{22/10-3}
\begin{array}{lll}
\lambda({\bf\tilde{s}})_{MN}\ =\ \sum_{i,j,k}\left(\pi_{j_{e_{m}}}(\tau_{i}\cdot\tau_{j}\cdot\tau_{k})\otimes\pi_{j_{e_{n}}}(\tau_{i}\cdot\tau_{j}\cdot\tau_{k})\right)_{MN}\\
\hspace*{2.9in}\textrm{if}\ {\bf\tilde{s}}=\left({\bf s}_{\delta}(1; v, e_{m}, e_{n})\right)_{\delta^{\prime}}(1;v_{e_{m} \delta},e_{n}^{\delta}, e_{m}^{\delta})
\end{array}
\end{equation}

\underline{case-4} : 
${\bf\tilde{s}}\in [{\bf s}]^{(1)}_{1,2}$
\begin{equation}\label{22/10-4}
\begin{array}{lll}
\lambda({\bf\tilde{s}})\ =\ \sum_{i,j,k}\left(\pi_{j_{e_{m}}}(\tau_{k}\cdot\tau_{i}\cdot\tau_{j})\otimes\pi_{j_{e_{n}}}(\tau_{i}\cdot\tau_{j}\cdot\tau_{k})\right)\\
\hspace*{2.9in}\textrm{if}\ {\bf\tilde{s}}=\left({\bf s}_{\delta}(1; v, e_{m}, e_{n})\right)_{\delta^{\prime}}(2;v_{e_{m} \delta},e_{n}^{\delta}, e_{m}^{\delta})
\end{array}
\end{equation}

\underline{case-5} : 
${\bf\tilde{s}}\in [{\bf s}]^{(1)}_{2,1}$
\begin{equation}\label{22/10-5}
\begin{array}{lll}
\lambda({\bf\tilde{s}})\ =\ \sum_{i,j,k}\left(\pi_{j_{e_{m}}}(\tau_{i}\cdot\tau_{j}\cdot\tau_{k})\otimes\pi_{j_{e_{n}}}(\tau_{j}\cdot\tau_{i}\cdot\tau_{k})\right)\\
\hspace*{2.9in}\textrm{if}\ {\bf\tilde{s}}=\left({\bf s}_{\delta}(2; v, e_{m}, e_{n})\right)_{\delta^{\prime}}(1;v_{e_{m} \delta},e_{n}^{\delta}, e_{m}^{\delta})
\end{array}
\end{equation}

\underline{case-6} : 
${\bf\tilde{s}}\in [{\bf s}]^{(1)}_{2,2}$
\begin{equation}\label{22/10-6}
\begin{array}{lll}
\lambda({\bf\tilde{s}})\ =\ \sum_{i,j,k}\left(\pi_{j_{e_{m}}}(\tau_{k}\cdot\tau_{i}\cdot\tau_{j})\otimes\pi_{j_{e_{n}}}(\tau_{j}\cdot\tau_{i}\cdot\tau_{k})\right)\\
\hspace*{2.9in}\textrm{if}\ {\bf\tilde{s}}=\left({\bf s}_{\delta}(2; v, e_{m}, e_{n})\right)_{\delta^{\prime}}(2;v_{e_{m} \delta},e_{n}^{\delta}, e_{m}^{\delta})
\end{array}
\end{equation}
\underline{case-7} : 
${\bf\tilde{s}}\in [{\bf s}]^{(2)}_{1,1}\cup [{\bf s}]^{(2)}_{1,3}$.\\
In this case, $\lambda({\bf\tilde{s}})$ is same as that given in, eq.(\ref{22/10-3}).
\begin{equation}\label{22/10-3prime}
\begin{array}{lll}
\lambda({\bf\tilde{s}})\ =\ \sum_{i,j,k}\left(\pi_{j_{e_{m}}}(\tau_{i}\cdot\tau_{j}\cdot\tau_{k})\otimes\pi_{j_{e_{n}}}(\tau_{i}\cdot\tau_{j}\cdot\tau_{k})\right)\\
\hspace*{2.9in}\textrm{if}\ {\bf\tilde{s}}=\left({\bf s}_{\delta}(1; v, e_{m}, e_{n})\right)_{\delta^{\prime}}(1;v_{e_{m} \delta},e_{m}^{\delta}, e_{n}^{\delta})
\end{array}
\end{equation}

\underline{case-8} : 
${\bf\tilde{s}}\in [{\bf s}]^{(2)}_{1,2}\cup [{\bf s}]^{(2)}_{1,4}$.\\
\begin{equation}\label{22/10-7}
\begin{array}{lll}
\lambda({\bf\tilde{s}})\ =\ \sum_{i,j,k}\left(\pi_{j_{e_{m}}}(\tau_{i}\cdot\tau_{j}\cdot\tau_{k})\otimes\pi_{j_{e_{n}}}(\tau_{k}\cdot\tau_{i}\cdot\tau_{j})\right)\\
\hspace*{2.9in}\textrm{if}\ {\bf\tilde{s}}=\left({\bf s}_{\delta}(1; v, e_{m}, e_{n})\right)_{\delta^{\prime}}(2;v_{e_{m} \delta},e_{m}^{\delta}, e_{n}^{\delta})
\end{array}
\end{equation}

\underline{case-9} : 
${\bf\tilde{s}}\in [{\bf s}]^{(2)}_{2,1}\cup [{\bf s}]^{(2)}_{2,3}$.\\
\begin{equation}\label{22/10-8}
\begin{array}{lll}
\lambda({\bf\tilde{s}})\ =\ \sum_{i,j,k}\left(\pi_{j_{e_{m}}}(\tau_{i}\cdot\tau_{j}\cdot\tau_{k})\otimes\pi_{j_{e_{n}}}(\tau_{j}\cdot\tau_{i}\cdot\tau_{k})\right)\\
\hspace*{2.9in}\textrm{if}\ {\bf\tilde{s}}=\left({\bf s}_{\delta}(2; v, e_{m}, e_{n})\right)_{\delta^{\prime}}(1;v_{e_{m} \delta},e_{n}^{\delta}, e_{m}^{\delta})
\end{array}
\end{equation}

\underline{case-10} : 
${\bf\tilde{s}}\in [{\bf s}]^{(2)}_{2,2}\cup [{\bf s}]^{(2)}_{2,4}$.\\
\begin{equation}\label{22/10-9}
\begin{array}{lll}
\lambda({\bf\tilde{s}})\ =\ \sum_{i,j,k}\left(\pi_{j_{e_{m}}}(\tau_{i}\cdot\tau_{j}\cdot\tau_{k})\otimes\pi_{j_{e_{n}}}(\tau_{k}\cdot\tau_{j}\cdot\tau_{i})\right)\\
\hspace*{2.9in}\textrm{if}\ {\bf\tilde{s}}=\left({\bf s}_{\delta}(2; v, e_{m}, e_{n})\right)_{\delta^{\prime}}(2;v_{e_{m} \delta},e_{m}^{\delta}, e_{n}^{\delta})
\end{array}
\end{equation}
With the above definition of habitat states in place, we can evaluate the continuum limit of $\delta^{2}\hat{H}_{T(\delta)}[N]$. 
\begin{equation}
\begin{array}{lll}
\lim_{\delta^{2}\rightarrow 0}\Psi^{{\cal F}}_{[{\bf s}]}(\delta\hat{H}_{T(\delta)}[N])\vert{\bf s}\rangle =\\
\vspace*{0.1in}
\frac{{\cal A}}{2\delta}\sum_{m=1}^{N_{v}}\sum_{n=1\vert n\neq m}^{N_{v}}\Psi^{{\cal F}}_{[{\bf s}]}\sum_{I=1}^{4}(-1)^{I+1}\vert{\bf s}_{\delta}(I;v,e_{m}, e_{n})\rangle=\\
\frac{{\cal A}}{2\delta}\\
\sum_{m=1}^{N_{v}}\sum_{n=1\vert n\neq m}^{N_{v}}\left[\Psi^{{\cal F}}_{[{\bf s}]}\left(\vert{\bf s}_{\delta}(1;v,e_{m}, e_{n})\rangle-\vert{\bf s}_{\delta}(2;v,e_{m}, e_{n})\rangle\right)+\Psi^{{\cal F}}_{[{\bf s}]}\left(\vert{\bf s}_{\delta}(3;v,e_{m}, e_{n})\rangle-\vert{\bf s}_{\delta}(4;v,e_{m}, e_{n})\rangle\right)\right]\\
\vspace*{0.1in}
=\frac{{\cal A}}{2\delta}\sum_{m=1}^{N_{v}}\sum_{n=1\vert n\neq m}^{N_{v}}\left(f(v+\delta\hat{e}_{m}(0))-f(v)\right)\left(\pi_{j_{e_{m}}}(\tau_{i}\cdot\tau_{j})\otimes\pi_{j_{e_{n}}}([\tau_{i},\tau_{j}])\right)_{M_{m}N_{n}}\\
\vspace*{0.1in}
=\frac{{\cal A}}{2}\sum_{m=1}^{N_{v}}\sum_{n=1\vert n\neq m}^{N_{v}}\hat{e}^{a}_{m}(0)\partial_{a}f(v)\left(\pi_{j_{e_{m}}}(\tau_{k})\otimes\pi_{j_{e_{n}}}(\tau_{k})\right)_{M_{m}N_{n}}
\end{array}
\end{equation}

\section{Continuum limit of LHS on the proposed Habitat.}\label{conlimitofLHS}
We now have all the ingredients to finally compute $\Psi^{{\cal F}}_{[{\bf s}]}\left([\hat{H}_{T(\delta^{\prime})}[N]\hat{H}_{T(\delta)}[M]\ -\ N\leftrightarrow\ M\right)\vert{\bf s}\rangle$.
As
\begin{equation}
\begin{array}{lll}
\lim_{\delta\rightarrow,\ \delta^{\prime}\rightarrow 0}\Psi^{{\cal F}}_{[{\bf s}]}\left([\hat{H}_{T(\delta^{\prime})}[N]\hat{H}_{T(\delta)}[M]\ -\ N\leftrightarrow\ M\right)\vert{\bf s}\rangle\ =\\
\lim_{\delta\rightarrow,\ \delta^{\prime}\rightarrow 0}\Psi^{{\cal F}}_{[{\bf s}]}\left(\left(\sum_{I=1}^{2}\sum_{J=1}^{2} + \sum_{I=1}^{2}\sum_{J=3}^{4}\right)\sum_{m,n}\left[M(v)N(v_{e_{m},\delta})\ -\ M\leftrightarrow N\right]\right.\\
\hspace*{1.6in}\left.\sum_{m^{\prime},n^{\prime}}(-1)^{I+J(\textrm{mod}\ 2)}\vert\left({\bf s}_{\delta}(I;v,e_{m},e_{n})\right)_{\delta^{\prime}}(J;v_{e_{m},\delta},e_{m^{\prime}}^{\delta},e_{n^{\prime}}^{\delta})\rangle\right)
\end{array}
\end{equation}
We have restricted the sum over $I$ to $\{1,2\}$ as for $I\in\{3,4\}$ the non-degenerate vertex is still located at $v$, whence such terms won't contribute to the commutator.\\
For a fixed $(I,J,m,n)$, depending on the ordered pair $(m^{\prime}, n^{\prime})$ 
the above matrix elements corresponds to case-{\bf (a)},case-{\bf (b)} or case-{\bf  (c)} as detailed below eq.(\ref{oct26-1}).\\
Before proceeding notice that 
\begin{equation}
\begin{array}{lll}
\Psi^{{\cal F}}_{[{\bf s}]}\left(\vert\left({\bf s}_{\delta}(I;v,e_{m},e_{n})\right)_{\delta^{\prime}}(J;v_{e_{m},\delta},e_{m^{\prime}}^{\delta},e_{n^{\prime}}^{\delta})\rangle\right) =\\
\vspace*{0.1in}
\hspace*{0.7in}\Psi^{{\cal F}}_{[{\bf s}]}\left(\vert\left({\bf s}_{\delta}(I;v,e_{m},e_{n};J_{s_{e_{m}}^{\delta}}=0,{\cal I})\right)_{\delta^{\prime}}(J;v_{e_{m},\delta},e_{m^{\prime}}^{\delta},e_{n^{\prime}}^{\delta})\rangle\right)=\\
\vspace*{0.1in}
=\ \sum_{{\bf c}\in[{\bf s}]}{\cal F}({\bf c})\langle{\bf c}\vert\left({\bf s}_{\delta}(I;v,e_{m},e_{n};J_{s_{e_{m}}^{\delta}}=0,{\cal I})\right)_{\delta^{\prime}}(J;v_{e_{m},\delta},e_{m^{\prime}}^{\delta},e_{n^{\prime}}^{\delta})\rangle\\
=\textrm{either}\\
f\left(v(\left({\bf s}_{\delta}(I;v,e_{m},e_{n};J_{s_{e_{m}}^{\delta}}=0,{\cal I})\right)_{\delta^{\prime}}(J;v_{e_{m},\delta},e_{m^{\prime}}^{\delta},e_{n^{\prime}}^{\delta}))\right)\\
\vspace*{0.1in}
\hspace*{1.5in}\lambda(\left({\bf s}_{\delta}(I;v,e_{m},e_{n};J_{s_{e_{m}}^{\delta}}=0,{\cal I})\right)_{\delta^{\prime}}\left(J;v_{e_{m},\delta},e_{m^{\prime}}^{\delta},e_{n^{\prime}}^{\delta})\right))\\
\textrm{or}\\
\Psi^{{\cal F}}_{[{\bf s}]}\left(\vert\left({\bf s}_{\delta}(I;v,e_{m},e_{n})\right)_{\delta^{\prime}}(J;v_{e_{m},\delta},e_{m^{\prime}}^{\delta},e_{n^{\prime}}^{\delta})\rangle\right) = 0
\end{array}
\end{equation}

\underline{{\bf case (a)
}} :\\
In this case LHS is given by,
\begin{equation}
\begin{array}{lll}
\lim_{\delta\rightarrow,\ \delta^{\prime}\rightarrow 0}\Psi^{{\cal F}}_{[{\bf s}]}
\left(\sum_{I=1}^{2}\sum_{J=1}^{2} + \sum_{I=1}^{2}\sum_{J=3}^{4}\right)\sum_{m,n}\\\hspace*{1.5in}(-1)^{I+J(\textrm{mod}\ 2)}\left[M(v)N(v_{e_{m},\delta})\ -\ M\leftrightarrow N\right]\vert\left({\bf s}_{\delta}(I;v,e_{m},e_{n})\right)_{\delta^{\prime}}(J;v_{e_{m},\delta},e_{m}^{\delta},e_{n}^{\delta})\rangle\\
=\ \lim_{\delta\rightarrow,\ \delta^{\prime}\rightarrow 0}\left(\sum_{I=1}^{2}\sum_{J=1}^{2} + \sum_{I=1}^{2}\sum_{J=3}^{4}\right)\sum_{m,n}\\\hspace*{1.5in}(-1)^{I+J(\textrm{mod}\ 2)}\left[M(v)N(v_{e_{m},\delta})\ -\ M\leftrightarrow N\right]{\cal F}(\left({\bf s}_{\delta}(I;v,e_{m},e_{n})\right)_{\delta^{\prime}}(J;v_{e_{m},\delta},e_{m}^{\delta},e_{n}^{\delta}))\\
\end{array}
\end{equation}
where we have used the definition of habitat given in eq.(\ref{def-hab}).\\
We 
we can now use eqs. ( (\ref{def-hab1}), (\ref{22/10-3prime}), (\ref{22/10-7}), (\ref{22/10-8}), (\ref{22/10-9}) ) and evaluate LHS for {\bf case-(a)} as,
\begin{equation}
\begin{array}{lll}
\textrm{LHS}=\\
\lim_{\delta\rightarrow 0}\lim_{\delta^{\prime}\rightarrow 0}\sum_{m,n}\left[M(v)N(v_{e_{m},\delta})\ -\ M\leftrightarrow N\right]\left(f((v_{e_{m},\delta})_{e_{m}^{\delta},\delta^{\prime}})-f(v_{e_{m},\delta})\right)\\
\vspace*{0.1in}
\hspace*{0.5in}\sum_{i,j,k}\left(\left(\pi_{j_{e_{m}}}(\tau_{i}\cdot\tau_{j}\cdot\tau_{k})\otimes\pi_{j_{e_{n}}}(\tau_{i}\cdot\tau_{j}\cdot\tau_{k})\right)_{M_{m}N_{n}}\ -\ \left(\pi_{j_{e_{m}}}(\tau_{k}\cdot\tau_{i}\cdot\tau_{j})\otimes\pi_{j_{e_{n}}}(\tau_{i}\cdot\tau_{j}\cdot\tau_{k})\right)_{M_{m}N_{n}}\ -\right.\\
\hspace*{1.0in}\left. \left(\pi_{j_{e_{m}}}(\tau_{i}\cdot\tau_{j}\cdot\tau_{k})\otimes\pi_{j_{e_{n}}}(\tau_{j}\cdot\tau_{i}\cdot\tau_{k})\right)_{M_{m}N_{n}}\ +\ \left(\pi_{j_{e_{m}}}(\tau_{k}\cdot\tau_{i}\cdot\tau_{j})\otimes\pi_{j_{e_{n}}}(\tau_{j}\cdot\tau_{i}\cdot\tau_{k})\right)_{M_{m}N_{n}}\right)\\
=\lim_{\delta\rightarrow 0}\lim_{\delta^{\prime}\rightarrow 0}\\
\hspace*{0.5in}\sum_{m,n}\left[M(v)N(v_{e_{m},\delta})\ -\ M\leftrightarrow N\right]\left(f((v_{e_{m},\delta})_{e_{m}^{\delta},\delta^{\prime}})-f(v_{e_{m},\delta})\right)\left(\pi_{j_{e_{m}}}(\tau_{i})\otimes\pi_{j_{e_{n}}}(\tau_{i})\right)_{M_{m}N_{n}}\\
=\ \sum_{m,n}\hat{e}^{a}_{m}(0)\left(M\partial_{a}N-N\partial_{a}M\right)\ \hat{e}^{b}_{m}(0)\partial_{b}f(v)\ \left(\pi_{j_{e_{m}}}(\tau_{i})\otimes\pi_{j_{e_{n}}}(\tau_{i})\right)_{M_{m}N_{n}}
\end{array}
\end{equation}
\emph{The relative minus sign between $f((v_{e_{m},\delta})_{e_{m}^{\delta},\delta^{\prime}}), f(v_{e_{m},\delta})$ is due to choice of $\kappa({\bf c})$ for various states as in eq.(\ref{defofkappac}).}\\
We have glossed over an important sublety in the above computation. It is related to the fact that due to non-trivial density weight of the Lapse, its evaluation at a point requires choice of co-ordinate. This issue was rather beautifully tackled in \cite{tv}, and we summarize the main aspect of their construction which is relevant for our purposes.\footnote{We are grateful to Madhavan Varadarajan for making us aware of these subtleties which turn out to be key players when one checks for diffeomorphism covariance of the constraint. \cite{mad2}}
Our discussion here is brief, lacking the precision of \cite{tv},  and for more details we refer the reader to section 4 of the above mentioned paper.\\
Let us assume that in the neighbor-hood the vertex $v$, we have fixed a co-ordinate system $\{x\}_{v}$ once and for all. For small enough $\delta$, the vertices $v_{e,\delta}$ lie inside such a neighborhood. The co-ordinate charts $\{x^{\prime}\}_{v_{e,\delta}}$ around such vertices can be obtained by rigidly translating the system $\{x\}_{v}$. Such a construction of co-ordinate patches around all the vertices of ${\bf c}$ ensures that the Jacobian $\frac{\partial x^{\prime a}_{\delta}}{\partial x^{b}}$ has a smooth  and non-singular limit as $\delta\rightarrow 0$ and moreover, due the the fact that different co-ordinate charts are related to each other by rigid translations, we have that
\begin{equation}\label{dec27-11}
\frac{\partial x^{\prime a}_{\delta}}{\partial x^{b}}\vert_{\delta=0}\ =\ \delta^{a}_{b}
\end{equation}
Thanks to eq.(\ref{dec27-11}), even though at finite triangulation, one has 
\begin{equation}
\begin{array}{lll}
N(v_{e_{m},\delta})\ :=\ N(\{x^{\prime}\}_{v_{e_{m},\delta}})\ = \vert\frac{\partial x_{v_{e_{m},\delta}}}{\partial x_{v}}\vert^{-1}N(\{x\}(v_{e_{m},\delta}))
\end{array}
\end{equation}
In the limit $\delta\rightarrow 0$, the Jacobian equals unity.\\
We also note that in arriving at the continuum limit, we have used eq.(\ref{deformedtangent}). At this point it is perhaps pertinent to note that in \cite{tv}, the authors also chose the deformations in such a way that if the deformation was along edge $e_{m}$, they had, to leading order in $\delta$, $\hat{e}_{n}^{\delta}(0)\ =\ -\hat{e}_{m}^{\delta}(0)$. It is tedious but a straight-forward exercise to show that the anomaly-freedom demonstrated in this paper will go through even for such ``conical" deformations as long as we are consistent with our prescription of deformation for both the Hamiltonian constraint and the Electric Vector constraint.

\underline{{\bf case (b)}} :\\
In this case LHS is given by,
\begin{equation}
\begin{array}{lll}
\lim_{\delta\rightarrow,\ \delta^{\prime}\rightarrow 0}\Psi^{{\cal F}}_{[{\bf s}]}
\left(\sum_{I=1}^{2}\sum_{J=1}^{2} +\ \sum_{I=1}^{2}\sum_{J=3}^{4}\right)\sum_{m,n}\\\hspace*{1.5in}(-1)^{I+J(\textrm{mod}\ 2)}\left[M(v)N(v_{e_{m},\delta})\ -\ M\leftrightarrow N\right]\vert\left({\bf s}_{\delta}(I;v,e_{m},e_{n})\right)_{\delta^{\prime}}(J;v_{e_{m},\delta},e_{n}^{\delta},e_{m}^{\delta})\rangle\\
=\ \lim_{\delta\rightarrow,\ \delta^{\prime}\rightarrow 0}\left(\sum_{I=1}^{2}\sum_{J=1}^{2}\ +\ \sum_{I=1}^{2}\sum_{J=3}^{4}\right)\sum_{m,n}\\\hspace*{1.5in}(-1)^{I+J(\textrm{mod}\ 2)}\left[M(v)N(v_{e_{m},\delta})\ -\ M\leftrightarrow N\right]{\cal F}(\left({\bf s}_{\delta}(I;v,e_{m},e_{n})\right)_{\delta^{\prime}}(J;v_{e_{m},\delta},e_{n}^{\delta},e_{m}^{\delta}))\\
\end{array}
\end{equation}
We can now use eqs. ( (\ref{def-hab1}), (\ref{22/10-3prime}), (\ref{22/10-7}), (\ref{22/10-8}), (\ref{22/10-9}) ) and evaluate LHS for {\bf case-(b)} as,
\begin{equation}
\begin{array}{lll}
\textrm{LHS}=\\
\lim_{\delta\rightarrow 0}\lim_{\delta^{\prime}\rightarrow 0}\sum_{m,n}\left[M(v)N(v_{e_{m},\delta})\ -\ M\leftrightarrow N\right]f((v_{e_{m},\delta})_{e_{n}^{\delta},\delta^{\prime}})\\
\vspace*{0.1in}
\hspace*{0.5in}\sum_{i,j,k}\left(\left(\pi_{j_{e_{m}}}(\tau_{i}\cdot\tau_{j}\cdot\tau_{k})\otimes\pi_{j_{e_{n}}}(\tau_{i}\cdot\tau_{j}\cdot\tau_{k})\right)_{M_{m}N_{n}}\ -\ \left(\pi_{j_{e_{m}}}(\tau_{k}\cdot\tau_{i}\cdot\tau_{j})\otimes\pi_{j_{e_{n}}}(\tau_{i}\cdot\tau_{j}\cdot\tau_{k})\right)_{M_{m}N_{n}}\ -\right.\\
\hspace*{1.0in}\left. \left(\pi_{j_{e_{m}}}(\tau_{i}\cdot\tau_{j}\cdot\tau_{k})\otimes\pi_{j_{e_{n}}}(\tau_{j}\cdot\tau_{i}\cdot\tau_{k})\right)_{M_{m}N_{n}}\ +\ \left(\pi_{j_{e_{m}}}(\tau_{k}\cdot\tau_{i}\cdot\tau_{j})\otimes\pi_{j_{e_{n}}}(\tau_{j}\cdot\tau_{i}\cdot\tau_{k})\right)_{M_{m}N_{n}}\right)\\
=\lim_{\delta\rightarrow 0}\lim_{\delta^{\prime}\rightarrow 0}\sum_{m,n}\left[M(v)N(v_{e_{m},\delta})\ -\ M\leftrightarrow N\right]\left(f((v_{e_{m},\delta})_{e_{n}^{\delta},\delta^{\prime}})-f(v_{e_{m},\delta})\right)\left(\pi_{j_{e_{m}}}(\tau_{i})\otimes\pi_{j_{e_{n}}}(\tau_{i})\right)_{M_{m}N_{n}}
\end{array}
\end{equation}
In this case the ordered continuum limit ($\delta^{\prime}\rightarrow\ 0$ , $\delta\rightarrow\ 0$) is given by
\begin{equation}
\begin{array}{lll}
\textrm{LHS}\ =\ \sum_{m,n\vert m\neq n}\hat{e}^{a}_{m}(0)\left[M(v)\partial_{a}N(v) - N(v)\partial_{a}M(v)\right]\hat{e}^{b}_{n}(0)\partial_{b}f(v)\left(\pi_{j_{e_{m}}}(\tau_{i})\otimes\pi_{j_{e_{n}}}(\tau_{i})\right)_{M_{m}N_{n}}
\end{array}
\end{equation}
\underline{{\bf case( c )}}:\\
As proven in lemma (\ref{kinklemma}),  $\langle\left({\bf s}_{\delta}(I;v,e_{m}, e_{n})_{\delta^{\prime}}
(J;v_{e_{m},\delta},e_{k}^{\delta}e_{l}^{\delta})\right)\vert{\bf c}^{\prime}\rangle = 0$ $\forall\ {\bf c}^{\prime}$ if $(k,l)\ \neq\ (m,n)$. Hence in this case, $\textrm{LHS}\ =\ 0.$\\

\subsection{Continuum limit of LHS : Final result}
We can use the three cases analyzed above to finally analyze the continuum limit of RHS in the weak-* topology, defined using ${\cal D}_{kin}\times\ {\cal V}_{hab}$ as
\begin{equation}\label{mainresultforlhs}
\begin{array}{lll}
\lim_{\delta\rightarrow,\ \delta^{\prime}\rightarrow 0}\Psi^{{\cal F}}_{[{\bf s}]}\left(\hat{H}_{T(\delta^{\prime})}[N]\hat{H}_{T(\delta)}[M]\ -\ N\leftrightarrow\ M\right)\vert{\bf s}^{\prime}\rangle\ =\\
\vspace*{0.1in}
\delta_{{\bf s},{\bf s}^{\prime}}{\cal A}^{2}\left[\ \sum_{m,n\vert m\neq n}\left(\ \hat{e}^{a}_{m}(0)\left[\ M(v)\partial_{a}N(v) - N(v)\partial_{a}M(v)\ \right]\hat{e}^{b}_{n}(0)\partial_{b}f(v)\ +\right.\right.\\
\left.\hspace*{2.5in}\hat{e}^{a}_{m}(0)\left[\ M(v)\partial_{a}N(v) - N(v)\partial_{a}M(v)\ \right]
\hat{e}^{b}_{m}(0)\partial_{b}f(v)\ \right)\\
\hspace*{3.7in}\left.\left(\pi_{j_{e_{m}}}(\tau_{i})\otimes\pi_{j_{e_{n}}}(\tau_{i})\right)_{M_{m}N_{n}}\right]
\end{array}
\end{equation}
where we have restored the factor of ${\cal A}^{2}=(\frac{3}{8\pi})^{2}$ in the final step.

\section{Analysis of the R.H.S}\label{analysisofRHS}
We now turn our attention to the quantization of the right hand side in light of the remarkable identity for the R.H.S that was obtained in \cite{tv}. Thus our starting point 
in the quantum theory is 
\begin{equation}\label{remarkablerhs}
\begin{array}{lll}
\textrm{R.H.S}\ =\ -3\sum_{i=1}^{3}\{D[M_{i}], D[N_{i}]\}\\
\end{array}
\end{equation}
where $N^{a}_{i}\ =\ N \tilde{E}^{a}_{i}$ and
\begin{equation}
\begin{array}{lll}
D[M_{i}]\ =\ \int_{\Sigma}M(x)\tilde{E}^{a}_{i}(x)F_{ab}^{j}(x)\tilde{E}^{b}_{j}(x)\\
\vspace*{0.1in}
=\ H_{\textrm{diff}}[M_{i}]\ +\ {\cal G}[\Lambda_{i}=NE^{a}_{i}A_{a}]
\end{array}
\end{equation}
We now use the quantization of the vector constraint $D[\vec{W}]$ for a c-no. shift $\vec{W}$ as given in \cite{aldiff} and quantize $D[N_{i}]$. 
The basic idea behind the quantization of the vector constraint is as follows.
Consider our canonical gauge variant spin-network state $\vert{\bf s}\rangle$.
Let $\vec{W}$ be a c-no. shift such that $\textrm{Supp}(\vec{W})\cap V({\bf s})\ =\ \{v\}$. As shown in \cite{aldiff}, there exists a quantization of $D[\vec{W}]$ on ${\cal H}_{kin}$ such that the action of the quantum operator on $\vert{\bf s}\rangle$ is given by,
\begin{equation}
\begin{array}{lll}
\hat{D}_{T(\delta)}[\vec{W}]\vert{\bf s}\rangle\ =\ \frac{1}{\delta}\left[\vert\overline{\phi}^{\delta}_{\vec{W}}\circ{\bf s}\rangle\ -\ \vert{\bf s}\rangle\right]\\
\textrm{where $\vert\overline{\phi}^{\delta}_{\vec{W}}\circ{\bf s}\rangle$ is a spin-network state.}
\\
\vspace*{0.1in}
T_{\overline{\phi}^{\delta}_{\vec{W}}\circ{\bf s}}\ :=\\
\left(\pi_{j_{e_{1}}}(h_{s_{\vec{W}}^{\delta}}\circ h_{\phi^{\delta}_{\vec{W}}\cdot e_{1}})\otimes\dots\otimes\pi_{j_{e_{N_{v}}}}(h_{s_{\vec{W}}^{\delta}}\circ h_{\phi^{\delta}_{\vec{W}}\cdot e_{N_{v}}})\right)_{m_{e_{1}}\dots m_{e_{N_{v}}},n_{e_{1}}\dots n_{e_{N_{v}}}}\cdot M^{m_{e_{1}}\dots m_{e_{N_{v}}},n_{e_{1}}\dots n_{e_{N_{v}}}}
\end{array}
\end{equation}
In the first line, we have denoted the action of $\hat{D}_{T(\delta)}[\vec{W}]$ by $\overline{\phi}_{\vec{W}}^{\delta}$. The bar on $\phi$ implies that the action is not given merely by a diffeomorphism but is a convolution of spatial diffeomorphism and $SU(2)$ gauge transformations. This can be most easily explained by looking at the second line where we see that in the functional representation, under the action of $\hat{D}_{T(\delta)}[\vec{W}]$ the transformed state is given by ``displacing" each edge beginning at $v$ along a diffeomorphism generated by $\vec{W}$, and attaching a segment $s^{\delta}_{\vec{W}}$ of co-ordinate length of $O(\delta)$ that is along the integral curves of $\vec{W}$ between $v$ and the beginning point of the displaced edge. We refer the reader to the figure \ref{fig5} for a succinct illustration of these ideas and to \cite{aldiff} for further details.\\
\begin{figure}
\begin{center}
\includegraphics[height=2in, width=5in]{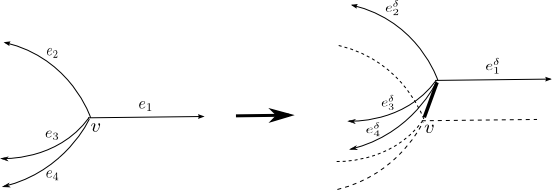}
\caption{Action of $\hat{D}_{T(\delta)}[\vec{W}]$ on a spin-net. The state moves along integral curves of  $\vec{W}$, and in this figure we have defined $e_{\delta}:=\phi^{\delta}_{\vec{W}}(e)$. The displaced vertex $\phi^{\delta}_{\vec{W}}(v)$ is attached to the original vertex by a segment along $\vec{W}$ and which is in $j_{e_{1}}\otimes\dots j_{e_{4}}$.}
\label{fig5}
\end{center}
\end{figure}

The special case when $\vec{W}$ is tangential to one of the edges, say $e_{1}$ at $v$ is worth mentioning. In this case, as the reader can easily verify herself, the action of finite triangulation vector constraint is given by,
\begin{equation}\label{shiftalongedge}
\begin{array}{lll}
T_{\overline{\phi}^{\delta}_{\vec{W}}\circ{\bf s}}\ :=\\
\left(\pi_{j_{e_{1}}}(h_{e_{1}})\otimes\dots\otimes\pi_{j_{e_{N_{v}}}}(h_{s_{\vec{W}}^{\delta}}\circ h_{\phi^{\delta}_{\vec{W}}\cdot e_{N_{v}}})\right)_{m_{e_{1}}\dots m_{e_{N_{v}}},n_{e_{1}}\dots n_{e_{N_{v}}}}\cdot M^{m_{e_{1}}\dots m_{e_{N_{v}}},n_{e_{1}}\dots n_{e_{N_{v}}}}
\end{array}
\end{equation}
We can now mimic this procedure to quantize $D[N_{i}\ =\ NE_{k}]$. The difference however is that $N\tilde{E}_{k}$ is no longer a c-no. shift, and its quantization as explained in section (\ref{qshift}) results in a ``singular" vector field with support in the neighborhood of $v$.\footnote{Recall that we are assuming that all our Lapse functions have support only in the nbd. of $v$ and vanish on all remaining vertices of ${\bf s}$.} We recall that the definition of quantum shift from section (\ref{qshift}), 
\begin{equation}
\begin{array}{lll}
\widehat{N^{a}_{k}}(v)\vert_{\delta}\vert{\bf s}\rangle :=\ N(x(v))\frac{1}{\frac{4\pi\delta^{2}}{3}}\sum_{e\in E(\gamma)\vert b(e)=v}\hat{e}^{a}(0)\hat{\tau}_{k}\vert_{v}\vert{\bf s}\rangle
\end{array}
\end{equation}
Three key structural differences between the quantum shift and a classical c-no. shift are worth mentioning.\\
\noindent{\bf (1)} The quantum shift is singular and is non-vanishing in the closed ball of co-ordinate radius $\delta$ around $v$, outside of which it vanishes.\\
\noindent{\bf (2)} Second difference is the presence of insertion operator $\hat{\tau}_{k}$ in the case of quantum shift which alters the intertwiner structure of ${\bf s}$ at $v$, and finally\\
\noindent{\bf (3)} The third difference is the fact that $\vec{W}$ is a vector field whereas $N\tilde{E}_{i}$ is a sensitized vector field with density weight -1.\\
 A careful mimicking of quantization of $D[\vec{W}]$ as applied to $D[N\tilde{E}_{k}]$ then results in the following finite-triangulation operator on ${\cal H}_{kin}$.
\begin{equation}
\begin{array}{lll}
\hat{D}_{T(\delta)}[N_{k}]T_{{\bf s}}\ =\\
\vspace*{0.1in}
\frac{1}{\frac{4\pi\delta^{2}}{3}}\\
\sum_{i=1}^{N_{v}}\left[\left(\pi_{j_{e_{1}}}(h_{s_{e_{i}}^{\delta}}\cdot h_{\phi^{\delta}_{\hat{e}_{i}(0)}\cdot s_{e_{1}}^{\delta}}\cdot h_{e_{1}-s_{e_{1}}^{\delta}})
\otimes\dots\otimes\pi_{j_{e_{i}}}(\tau_{k}\cdot h_{e_{i}})\otimes\right.\right.\\
\hspace*{1.9in}\dots\otimes
\left.\left.\pi_{j_{e_{N_{v}}}}(h_{s_{e_{i}}^{\delta}}\cdot h_{\phi^{\delta}_{\hat{e}_{i}(0)}\cdot s_{e_{N_{v}}}^{\delta}}\cdot h_{e_{N_{v}}-s_{e_{N_{v}}}^{\delta}})\right)\cdot M\ -\ T_{{\bf s}}\right]
\end{array}
\end{equation}
We encourage the reader to understand the equation in light of eq.(\ref{shiftalongedge}) with $\vec{W}$ replaced by $\hat{N}_{k}$. The extra factor of $\frac{1}{\delta}$ as compared to the factor in eq.(\ref{shiftalongedge}) is precisely due to the fact that the $N\tilde{E}_{i}$ is a densitized vector field.\\
We denote the deformed state as,
\begin{equation}\label{eq:12.7}
\begin{array}{lll}
\left(\pi_{j_{e_{1}}}(h_{s_{e_{i}}^{\delta}}\cdot h_{\phi^{\delta}_{\hat{e}_{i}(0)}\cdot s_{e_{1}}^{\delta}}\cdot h_{e_{1}-s_{e_{1}}^{\delta}})
\otimes\dots\otimes\pi_{j_{e_{i}}}(\tau_{k}\cdot h_{e_{i}})\otimes\right.\\
\hspace*{1.9in}\dots\otimes
\left.\pi_{j_{e_{N_{v}}}}(h_{s_{e_{i}}^{\delta}}\cdot h_{\phi^{\delta}_{\hat{e}_{i}(0)}\cdot s_{e_{N_{v}}}^{\delta}}\cdot h_{e_{N_{v}}-s_{e_{N_{v}}}^{\delta}})\right)\cdot M\ =:\ \vert\overline{\phi}^{\delta}_{\hat{e}_{i}(0)}\cdot{\bf s}(v,e_{i},k)\rangle\ 
\end{array}
\end{equation}
\noindent{\bf (i)} ${\bf s}(v, e_{i},k)$ denotes a cylindrical-network obtained from ${\bf s}$ by insertion of $\tau_{k}$ along edge $e_{i}$ with $f(e_{i})\ =\ v$.\\
\noindent{\bf (ii)} As before $\overline{\phi}^{\delta}_{\hat{e}_{i}(0)}$ denotes a (one parameter family of) singular transformation along $\hat{e}_{i}(0)$ which is generated by the Vector constraint, as opposed to the Diffeomorphism constraint.\\
\footnote{Note that Eq.(\ref{eq:12.7}), can in fact be re-written by using the basic intertwining property of $\pi_{j}(\tau_{k})$,
\begin{equation}
\pi_{j_{e}}(\tau_{k}\cdot h_{s_{e}^{\delta}})_{m_{e}m_{e}^{\prime}}\ =\ 
\pi_{j_{e}}(\tau_{k})_{m_{e}\tilde{m}_{e}}\pi_{j_{e}}(h_{s_{e}^{\delta}})_{\tilde{m}_{e}m_{e}^{\prime}}\ =\ \left(\pi_{j_{e}}(h_{s_{e}})_{m_{e}\tilde{m}_{e}}\otimes\pi_{j=1}(h_{s_{e}})_{kl}\right)\cdot\pi_{j_{e}}(\tau_{l})_{\tilde{m}_{e}m_{e}^{\prime}}
\end{equation}
Using the above equation, we can re-write eq.(\ref{eq:12.7}) in such a way that the insertion operator is ``placed at" $v_{e,\delta}$ in the functional representation.
\begin{equation}
\begin{array}{lll}
\vert\overline{\phi}^{\delta}_{\hat{e}_{i}(0)}\cdot{\bf s}(v,e_{i},k)\rangle\ =\\ 
\left(\pi_{j_{e_{1}}}(h_{s_{e_{i}}^{\delta}}\cdot h_{\phi^{\delta}_{\hat{e}_{i}(0)}\cdot s_{e_{1}}^{\delta}}\cdot h_{e_{1}-s_{e_{1}}^{\delta}})
\otimes\dots\otimes\left[(\pi_{j_{e_{i}}}(h_{s_{e_{i}}^{\delta}})\otimes\pi_{j=1}(h_{s_{e_{i}}^{\delta}})_{kl})\cdot\pi_{j_{e_{i}}}(\tau_{l})\right]\otimes\right.\\
\hspace*{1.9in}\dots\otimes
\left.\pi_{j_{e_{N_{v}}}}(h_{s_{e_{i}}^{\delta}}\cdot h_{\phi^{\delta}_{\hat{e}_{i}(0)}\cdot s_{e_{N_{v}}}^{\delta}}\cdot h_{e_{N_{v}}-s_{e_{N_{v}}}^{\delta}})\right)\cdot M\\
\end{array}
\end{equation}
However, we will not have an occasion to use this form of the state in this paper.}
The resulting state is depicted in the figure above.\\
\begin{figure}
\begin{center}
\includegraphics[height=2in, width=5in]{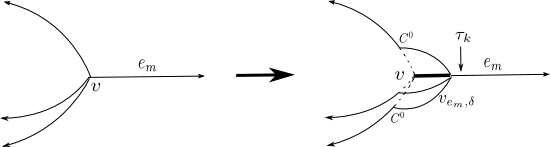}
\caption{Action of $\hat{D}_{T(\delta)}[N_{k}]$ on a spin-net.Only deformation along $e_{m}$ is displayed. $\tau_{k}$ is inserted along $e_{m}$ at the displaced vertex and all the bivalent vertices are of differentiability class $C^{0}$.}
\label{fig6}
\end{center}
\end{figure}

Precisely as for $\vert{\bf s}_{\delta}(2;v,e_{m},e_{n})$, we note that the cylindrical-network state $\vert\overline{\phi}^{\delta}_{\hat{e}_{i}(0)}\cdot{\bf s}(v,e_{i},k)\rangle$ has an edge $s_{e_{i}}^{\delta}$ which is in the reducible representation of $SU(2)$. Exactly as for $\vert{\bf s}_{\delta}(2;v,e_{m},e_{n})$ we can write $\vert\overline{\phi}^{\delta}_{\hat{e}_{i}(0)}\cdot{\bf s}(v,e_{i},k)\rangle$ as a linear combination of cylindrical-networks with $J_{s_{e_{i}^{\delta}}}\ \in\ \{0,\dots,\vec{j_{e_{1}}}+\dots\vec{j_{e_{N_{v}}}}\}$. 
\begin{equation}\label{cylnet-J=0}
\begin{array}{lll}
\vert\overline{\phi}^{\delta}_{\hat{e}_{i}(0)}\cdot{\bf s}(v,e_{i},k)\rangle\ =\ 
\vert\overline{\phi}^{\delta}_{\hat{e}_{i}(0)}\cdot{\bf s}(v,e_{i},k;J_{s_{e_{i}}^{\delta}}=0)\rangle\ +\ \sum_{J_{s_{e_{i}}^{\delta}}\neq\ 0}\vert\overline{\phi}^{\delta}_{\hat{e}_{i}(0)}\cdot{\bf s}(v,e_{i},k;J_{s_{e_{i}}^{\delta}})\rangle
\end{array}
\end{equation}
The $J_{s_{e_{i}}^{\delta}}\ =\ 0$ state can be written as a linear combination over various Intertwiners ${\cal I} : j_{e_{1}}\otimes\dots j_{e_{N_{v}}}\rightarrow\ \mathbb{C}$.\\
\begin{equation}\label{dec23-1}
\begin{array}{lll}
\vert\overline{\phi}^{\delta}_{\hat{e}_{i}(0)}\cdot{\bf s}(v,e_{i},k;J_{s_{e_{i}}^{\delta}}=0)\rangle\ =\ \sum_{{\cal I}}\vert\overline{\phi}^{\delta}_{\hat{e}_{i}(0)}\cdot{\bf s}(v,e_{i},k;J_{s_{e_{i}}^{\delta}}=0,{\cal I})\rangle
\end{array}
\end{equation}
As for the Hamiltonian constraint, our interest is in the rescaled Vector constraint operator $\delta\hat{D}_{T(\delta)}[N_{k}]$. 
\begin{equation}
\delta\hat{D}_{T(\delta)}[N_{k}]\vert{\bf s}\rangle\ =\ N(x(v))\frac{1}{\frac{4\pi}{3}}\frac{1}{\delta}\sum_{i=1}^{N_{v}}\left(\vert\overline{\phi}^{\delta}_{\hat{e}_{i}(0)}\cdot{\bf s}(v,e_{i},k)\rangle\ -\ \vert{\bf s}(v,e_{i},k)\rangle\right)
\end{equation}
In  the spirit of the nomenclature of \cite{tv}, we will refer to $\delta\hat{D}_{T(\delta)}[N_{k}]$ as electric Vector constraint.\\
We can now evaluate the action of two successive Electric Vector constraints on ${\bf s}$ as follows.

\begin{equation}\label{2singulardiffs-1}
\begin{array}{lll}
\hat{D}_{T(\delta^{\prime})}[M_{k}]\hat{D}_{T(\delta)}[N_{k}]\ \vert{\bf s}\rangle\ =\\ 
\vspace*{0.1in}
N(x(v))\frac{1}{\frac{4\pi}{3}}\frac{1}{\delta}\sum_{m=1}^{N_{v}}\hat{D}_{T(\delta^{\prime})}[M_{k}]\left(\vert\overline{\phi}^{\delta}_{\hat{e}_{m}(0)}\cdot{\bf s}(v,e_{m},k)\rangle\ -\ \vert{\bf s}(v,e_{m},k)\rangle\right)\\
\vspace*{0.1in}
=\ N(x(v))\left(\frac{1}{\frac{4\pi}{3}}\right)^{2}\frac{1}{\delta\delta^{\prime}}
\sum_{m=1}^{N_{v}}\sum_{n=1}^{N_{v}+1}M(x(v_{e_{m},\delta}))\left[\vert\overline{\phi}^{\delta^{\prime}}_{\hat{e}_{n}^{\delta}(0)}\left(\overline{\phi}^{\delta}_{\hat{e}_{m}(0)}\cdot{\bf s}(v,e_{m},k)\right)(v_{e_{m},\delta},e_{n}^{\delta},k)\rangle\right.\\
\hspace*{3.0in} \left.-\ 
\vert\left(\overline{\phi}^{\delta}_{\hat{e}_{m}(0)}\cdot{\bf s}(v,e_{m},k)\right)(v_{e_{m},\delta},e_{n}^{\delta},k)\rangle\right]\ +\ \dots
\end{array}
\end{equation}
where $\dots$ indicate states whose coefficients include $N(x(v))M(x(v))$ and which will drop out in the (finite-triangulation) commutator.\\
Once again we note that there in an implicit assumption in the above computation, namely that the action of $\hat{D}_{T(\delta^{\prime})}[M_{k}]$ is trivial on bivalent kink vertices. As before, this assumption would only be true for density $\frac{1}{6}$ constraint and as emphasized above, since our ultimate goal is to check off-shell closure for such constraints, we have used this as a key input in our analysis.\\
Whence the finite triangulation commutator between $D[N_{k}]$ and $D[M_{k}]$ is given by,
\begin{equation}
\begin{array}{lll}
\left(\hat{D}_{T(\delta^{\prime})}[M_{k}]\hat{D}_{T(\delta)}[N_{k}]\ -\ \hat{D}_{T(\delta^{\prime})}[N_{k}]\hat{D}_{T(\delta)}[M_{k}]\right)
 \vert{\bf s}\rangle\ =\\ 
\vspace*{0.1in}
=\ \left(\frac{1}{\frac{4\pi}{3}}\right)^{2}\frac{1}{\delta\delta^{\prime}}
\sum_{m=1}^{N_{v}}\sum_{n=1}^{N_{v}+1}\left(N(x(v))M(x(v_{e_{m},\delta}))-M\leftrightarrow N\right)\left[\vert\overline{\phi}^{\delta^{\prime}}_{\hat{e}_{n}^{\delta}(0)}\left(\overline{\phi}^{\delta}_{\hat{e}_{m}(0)}\cdot{\bf s}(v,e_{m},k)\right)(v_{e_{m},\delta},e_{n}^{\delta},k)\rangle\right.\\
\hspace*{4.0in}\left. -\ 
\vert\left(\overline{\phi}^{\delta}_{\hat{e}_{m}(0)}\cdot{\bf s}(v,e_{m},k)\right)(v_{e_{m},\delta},e_{n}^{\delta},k)\rangle\right]
\end{array}
\end{equation}
We can now use the decomposition of $\vert\overline{\phi}^{\delta}_{\hat{e}_{i}(0)}\cdot{\bf s}(v,e_{i},k)\rangle$ given in eq.(\ref{cylnet-J=0}) and use the fact that for $J_{s_{e_{m}}^{\delta}}\ =\ 0$ state, the vertex $v_{e_{m},\delta}$ is $N_{v}$ valent to obtain,
\begin{equation}\label{dec24-1}
\begin{array}{lll}
\left(\hat{D}_{T(\delta^{\prime})}[M_{k}]\hat{D}_{T(\delta)}[N_{k}]\ -\ \hat{D}_{T(\delta^{\prime})}[N_{k}]\hat{D}_{T(\delta)}[M_{k}]\right)
 \vert{\bf s}\rangle\ =\\ 
\vspace*{0.1in}
\left(\frac{1}{\frac{4\pi}{3}}\right)^{2}\frac{1}{\delta\delta^{\prime}}
\sum_{m=1}^{N_{v}}\sum_{n=1}^{N_{v}}\left(N(x(v))M(x(v_{e_{m},\delta}))-M\leftrightarrow N\right)\\
\hspace*{1.3in}\left[\vert\overline{\phi}^{\delta^{\prime}}_{\hat{e}_{n}^{\delta}(0)}\left(\overline{\phi}^{\delta}_{\hat{e}_{m}(0)}\cdot{\bf s}(v,e_{m},k;J_{s_{e_{m}}}^{\delta}=0)\right)(v_{e_{m},\delta},e_{n}^{\delta},k)\rangle\right.\\
\hspace*{3.5in}\left. -\ 
\vert\left(\overline{\phi}^{\delta}_{\hat{e}_{m}(0)}\cdot{\bf s}(v,e_{m},k;J_{s_{e_{m}}}^{\delta}=0)\right)(v_{e_{m},\delta},e_{n}^{\delta},k)\rangle\right]\\
\vspace*{0.1in}
\hspace*{4.7in}+\ \dots
\end{array}
\end{equation}
where $\dots$ in last line indicate all the states with $J_{s_{e_{m}}^{\delta}}\ \neq\ 0$.\\
Given an intertwiner ${\cal I}\ :\ j_{e_{1}}\otimes\dots\otimes j_{e_{N_{v}}}\rightarrow\ \mathbb{C}$, we can further expand the above equation as,
 \begin{equation}\label{dec24-2}
\begin{array}{lll}
\left(\hat{D}_{T(\delta^{\prime})}[M_{k}]\hat{D}_{T(\delta)}[N_{k}]\ -\ \hat{D}_{T(\delta^{\prime})}[N_{k}]\hat{D}_{T(\delta)}[M_{k}]\right)
 \vert{\bf s}\rangle\ =\\ 
\left(\frac{1}{\frac{4\pi}{3}}\right)^{2}\frac{1}{\delta\delta^{\prime}}
\sum_{m=1}^{N_{v}}\sum_{n=1}^{N_{v}}\left(N(x(v))M(x(v_{e_{m},\delta}))-M\leftrightarrow N\right)\\
\hspace*{1.3in}\left[\vert\overline{\phi}^{\delta^{\prime}}_{\hat{e}_{n}^{\delta}(0)}\left(\overline{\phi}^{\delta}_{\hat{e}_{m}(0)}\cdot{\bf s}(v,e_{m},k;J_{s_{e_{m}}}^{\delta}=0,{\cal I})\right)(v_{e_{m},\delta},e_{n}^{\delta},k)\rangle\right.\\
\hspace*{3.5in}\left. -\ 
\vert\left(\overline{\phi}^{\delta}_{\hat{e}_{m}(0)}\cdot{\bf s}(v,e_{m},k;J_{s_{e_{m}}}^{\delta}=0,{\cal I})\right)(v_{e_{m},\delta},e_{n}^{\delta},k)\rangle\right]\\
\vspace*{0.1in}
\hspace*{4.7in}+\ \dots
\end{array}
\end{equation}
where the $\dots$ not indicate all the remaining terms which arise from the decomposition
\begin{equation}
\begin{array}{lll}
\vert\overline{\phi}^{\delta}_{\hat{e}_{m}(0)}\cdot{\bf s}(v,e_{m},k)\rangle\ =\\ 
\vspace*{0.1in}
\vert\overline{\phi}^{\delta}_{\hat{e}_{m}(0)}\cdot{\bf s}(v,e_{m},k;J_{s_{e_{m}}}^{\delta}=0,{\cal I})\rangle\ +\ \\
\vspace*{0.1in}
\left[\sum_{{\cal I}^{\prime}\neq{\cal I}}\vert\overline{\phi}^{\delta}_{\hat{e}_{m}(0)}\cdot{\bf s}(v,e_{m},k;J_{s_{e_{m}}}^{\delta}=0,{\cal I}^{\prime})\rangle\ +\ \sum_{J_{s_{e_{m}}}^{\delta}\ \neq\ 0}\vert\overline{\phi}^{\delta}_{\hat{e}_{m}(0)}\cdot{\bf s}(v,e_{m},k;J_{s_{e_{m}}}^{\delta}\neq\ 0)\rangle\right]
\end{array}
\end{equation}
As we will see such states will not contribute in the continuum limit Habitat states have trivial projection along them.\\

Readers familiar with the constructions in \cite{tv} will notice that the singular deformations have analogous structure in $SU(2)$ theory as that in $U(1)^{3}$ theory. Additional complications simply arise due to insertion operators i.e. non-diagonal nature of inverse metric.\\

\section{Continuum limit of RHS}\label{RHS-section}
In this section we will compute
\begin{equation}\label{contlimitrhs}
\begin{array}{lll}
\textrm{RHS}\ =\ (-3)\lim_{\delta,\delta^{\prime}\rightarrow\ 0}\Psi^{{\cal F}}_{[{\bf s}]}\sum_{k=1}^{3}\left(\hat{D}_{T(\delta^{\prime})}[M_{k}]\hat{D}_{T(\delta)}[N_{k}]\ -\ M\leftrightarrow N\right)\vert{\bf s}^{\prime}\rangle
\end{array}
\end{equation}
$\forall\ \vert{\bf s}^{\prime}\rangle$.\\
Our analysis will proceed as follows. Recall that, given an intertwiner ${\cal I}\ :\ j_{e_{1}}\otimes\dots\otimes j_{e_{N_{v}}}$
\begin{displaymath}
[{\bf s}]\ :=\ \{{\bf c}\vert {\bf c}\in\ {\cal S}({\bf s}, {\cal I})\ \cup\ {\cal T}({\bf s},{\cal I})\}
\end{displaymath}
Where we have not defined ${\cal T}({\bf s},{\cal I})$ yet, and ${\cal S}({\bf s},{\cal I})$ is defined in  eq.(\ref{defofcalS}).\\
We will now show that\\
\noindent{\bf (1)} $\langle{\bf c}\vert\sum_{k=1}^{3}\left(\hat{D}_{T(\delta^{\prime})}[M_{k}]\hat{D}_{T(\delta)}[N_{k}]\ -\ M\leftrightarrow N\right)\vert{\bf s}^{\prime}\rangle\ =\ 0$ $\forall\ {\bf c}\ \in\ {\cal S}({\bf s}), \forall\ {\bf s}^{\prime}$.\\
\noindent{\bf (2)} We then construct the set ${\cal T}({\bf s},{\cal I})$, and complete the definition of $[{\bf s}]$.\\
\noindent{\bf (3)} We define the vertex tensors associated to ${\bf c}\ \in\ {\cal T}({\bf s},{\cal I})$ and this completes the definition of ${\cal F}$.\\
Using the above results, continuum limit of RHS is computed and is shown to match with the continuum limit of LHS.

\underline{{\bf Lemma}} : Consider a (gauge-variant) Spin-network ${\bf s}^{\prime}$ with a single, non-degenerate, $N_{v}$-valent vertex $v$. Let ${\bf s}$ be our canonical spin-network.
Then,
\begin{equation}\label{eq:dec22-1}
\langle{\bf c}\vert\sum_{k=1}^{3}\left(\hat{D}_{T(\delta^{\prime})}[M_{k}]\hat{D}_{T(\delta)}[N_{k}]\ -\ M\leftrightarrow N\right)\vert{\bf s}^{\prime}\rangle\ =\ 0
\end{equation}
$\forall\ {\bf s}^{\prime},\ {\bf c}\ \in\ {\cal S}({\bf s})$.\\
\begin{proof}
We first consider ${\bf s}^{\prime}\ =\ {\bf s}$. In this case, it is immediate that the action of $\hat{D}_{T(\delta^{\prime})}[M_{k}]\hat{D}_{T(\delta)}[N_{k}]$ will create $2(N_{v}-1)$ $C^{0}$ bivalent Kink vertices. As there is \emph{no} cylindrical-network in ${\cal S}({\bf s})$ whose underlying graph contains $2(N_{v}-1)$ $C^{0}$ kink vertices, it is clear that eq.(\ref{eq:dec22-1}) holds.\\
Now let ${\bf s}^{\prime}\ \neq\ {\bf s}$. In this case eq.(\ref{eq:dec22-1}) continues to hold simply due to the argument given above. However 
even without appealing to the detailed nature of Bi-valent vertices created by electric Diffeomorphism constraint eq.(\ref{eq:dec22-1}) is satisfied in this case. 
This can be seen as follows.\\
It is safe to consider the scenario where ${\bf s}^{\prime}$ has a single non-degenerate $N_{v}$ valent vertex with incident spins being some permutation of $(j_{e_{1}},\dots,j_{e_{N_{v}}})$. In this case the only way ${\bf s}^{\prime}\ \neq\ {\bf s}$ is if $\gamma({\bf s}^{\prime})\ \neq\ \gamma({\bf s})$. Hence there exists at least one edge $e\in\ E(\gamma({\bf s})),\ e^{\prime}\ \in\ E(\gamma({\bf s}^{\prime}))$ such that there is a ball of radius $\epsilon$ centered at $v$ such that inside the ball, $e\ \neq\ e^{\prime}$. As we are looking at the limit $\delta^{\prime},\delta\rightarrow\ 0$, consider 
$\delta^{\prime},\delta\ \ll\ \epsilon$. In order for $\langle{\bf c}\vert\sum_{k=1}^{3}\left(\hat{D}_{T(\delta^{\prime})}[M_{k}]\hat{D}_{T(\delta)}[N_{k}]\ -\ M\leftrightarrow N\right)\vert{\bf s}^{\prime}\rangle$ to not vanish, there must exist atleast one ${\bf c}\ \in\ {\cal S}({\bf s})$ whose graph coincides with  
$\gamma(\overline{\phi}^{\delta}_{\hat{e}_{m}(0)}({\bf s}(v;e_{m},k)))$ for some $m$. Now note that  the non-degenerate vertex for such a state is placed by $v_{e_{m},\delta}$ and the singular diffeomorphism $\overline{\phi}^{\delta}$ deforms the graph $\gamma({\bf s}^{\prime}$ in the neighborhood (of $v$) of radius $O(\delta)$. Whence for $\delta\ \ll\ \epsilon$, there is at least one edge segment 
belonging to $e^{\prime}\ \in\ {\bf s}^{\prime}$ which lies inside ${\cal B}(v,\epsilon)$ and is left unchanged by $\overline{\phi}^{\delta}$. Considering all the ${\bf c}\ \in\ {\cal S}({\bf s})$ which have the non-degenerate vertex at the same location $v_{e_{m},\delta}$, we notice that all such ${\bf c}$'s will have graphs which are deformed relative to $\gamma({\bf s})$, only in the nbd. of $v$ of radius $O(\delta)$, Whence for all such states, there is a segment in $e$ which is with in 
${\cal B}(v,\epsilon)$ but is un-displaced by Hamiltonian constraint deformations.
Thus for all such states, the underlying graph \emph{can not} coincide with the graph underlying $\overline{\phi}^{\delta}_{\hat{e}_{m}(0)}({\bf s}(v;e_{m},k))$ for any $m$. This completes the proof.
\end{proof}

\subsection{Construction of ${\cal T}({\bf s},{\cal I})$}\label{setforRHS}
In order to take the continuum limit of finite triangulation LHS, we first need to complete the definition of habitat state. As we saw in section(\ref{habitat}), habitat states are defined as 
\begin{equation}
\begin{array}{lll}
\Psi^{{\cal F}}_{[{\bf s}]}\ :=\ \sum_{{\bf c}\in{\cal S}({\bf s},{\cal I})}\kappa({\bf c}){\cal F}({\bf c})\langle{\bf c}\vert\ -\frac{1}{12} \sum_{{\bf c}^{\prime}\in{\cal T}({\bf s},{\cal I})}\kappa({\bf c}^{\prime}){\cal F}({\bf c}^{\prime})\langle{\bf c}^{\prime}\vert
\end{array}
\end{equation}
We first define the set ${\cal T}({\bf s},{\cal I})$, and then introduce the vertex tensors for ${\bf c}\ \in\ {\cal T}({\bf s},{\cal I})$ which completes the definition of vertex function ${\cal F}$.\\
The construction of this set is similar to the construction of ${\cal S}({\bf s},{\cal I})$. We consider all possible cylindrical networks obtained by the action of 
$\hat{D}_{T(\delta)}$, $\hat{D}_{T(\delta^{\prime})}\hat{D}_{T(\delta)}$ on $\vert{\bf s}\rangle$. The cylindrical networks so obtained can be re-expressed as a linear combination of states as in eq.(\ref{dec23-1}). Given an intertwiner, ${\cal I}$, we will use the states $\vert\overline{\phi}^{\delta}_{\hat{e}_{m}(0)}({\bf s}(k;e_{m},v;{\cal I}))\rangle$ and the states obtained by action of $\hat{D}_{T(\delta^{\prime})}$ on such states to generate ${\cal T}({\bf s},{\cal I})$.\\ 
Whence consider following sets of cylindrical-networks for each $k\ \in\ \{1,2,3\}$.
\begin{equation}\label{dec25-6}
\begin{array}{lll}
({\bf s})^{k}\ :=\ \{{\bf c}\vert{\bf c}\ =\ {\bf s}(k;e,v)\ \forall\ e\vert_{b(e)=v}\}\\
\vspace*{0.1in}
[{\bf s}]_{\delta}^{k,(m)}\ =\\
\{{\bf c}\vert{\bf c}\ :=\ \overline{\phi}^{\delta}_{\hat{e}_{m}(0)}\left({\bf s}(k;e_{m},v;J_{s_{e_{m}}^{\delta}}=0;{\cal I})\right),\ e_{m}\vert_{b(e_{m})=v},\ \delta\leq \delta_{0}({\bf s})\}\\
\vspace*{0.1in}
[[{\bf s}]_{\delta}^{k,(m)}]^{k,(n)}\ :=\\
\{{\bf c}\vert{\bf c}\ =\ \overline{\phi}^{\delta}_{\hat{e}_{m}(0)}\left({\bf s}(k;e_{m},v;J_{s_{e_{m}}^{\delta}}=0,{\cal I})\right)(k;e_{n}^{\delta},v),\ e_{m}\vert_{b(e_{m})=v}, \delta^{\prime}\ll\delta\leq \delta_{0}({\bf s}), n\neq m\}\\
\vspace*{0.1in}
[[{\bf s}]_{\delta}^{k,(m)}]^{k,(m)}\ :=\\
\{{\bf c}\vert{\bf c}\ =\ \overline{\phi}^{\delta}_{\hat{e}_{m}(0)}\left({\bf s}(k;e_{m},v;J_{s_{e_{m}}^{\delta}}=0,{\cal I})\right)(k;e_{m}^{\delta},v),\ e_{m}\vert_{b(e_{m})=v}, \delta^{\prime}\ll\delta\leq \delta_{0}({\bf s})\}\\
\vspace*{0.1in}
([[{\bf s}]_{\delta}^{k,(m)}]_{\delta^{\prime}}^{k,(n)})\ =\\
\{{\bf c}\vert{\bf c}\ =\ \overline{\phi}^{\delta^{\prime}}_{\hat{e_{n}^{\delta}}(0)}\left(\overline{\phi}^{\delta}_{\hat{e}_{m}(0)}\left({\bf s}(k;e_{m},v;J_{s_{e_{m}}^{\delta}}=0;{\cal I})\right)(k;e_{n}^{\delta},v)\right),\ e_{m}\vert_{b(e_{m})=v}, \delta^{\prime}\ll\delta\leq \delta_{0}({\bf s}), e_{n}^{\delta}\neq e_{m}^{\delta}\}\\
\vspace*{0.1in}
([[{\bf s}]_{\delta}^{k,(m)}]_{\delta^{\prime}}^{k,(m)})\ =\\
\{{\bf c}\vert{\bf c}\ =\ \overline{\phi}^{\delta^{\prime}}_{\hat{e_{m}^{\delta}}(0)}\left(\overline{\phi}^{\delta}_{\hat{e}_{m}(0)}\left({\bf s}(k;e_{m},v;J_{s_{e_{m}}^{\delta}}=0;{\cal I})\right)(k;e_{m}^{\delta},v)\right),\ e_{m}\vert_{b(e_{m})=v}, \delta^{\prime}\ll\delta\leq \delta_{0}({\bf s})\}
\end{array}
\end{equation}
The superscript $k,(m)$ denotes that we have an insertion along edge $e_{m}$ at $b(e_{m})$ of insertion operator $\hat{\tau}_{k}$. The subscript $\delta$ implies that the states are further deformed by $\overline{\phi}^{\delta}_{\hat{e}_{m}(0)}$. 

${\cal T}({\bf s},{\cal I})$ is defined as,
\begin{equation}
\begin{array}{lll}
{\cal T}({\bf s},{\cal I})\ :=\\
\cup_{k}\left[\cup_{m}\cup_{n\vert n\neq m}[[{\bf s}]_{\delta}^{k,(m)}]^{k,(n)}\ \cup\ \cup_{m}[[{\bf s}]_{\delta}^{k,(m)}]^{k,(m)}\ \bigcup\ 
\cup_{m}\cup_{n\vert n\neq m}([[{\bf s}]_{\delta}^{k,(m)}]_{\delta^{\prime}}^{k,(n)})\ \cup\ ([[{\bf s}]_{\delta}^{k,(m)}]_{\delta^{\prime}}^{k,(m)})\right]
\end{array}
\end{equation}
\begin{lemma}

\begin{equation}
\begin{array}{lll}
\langle{\bf c}\vert{\bf c}^{\prime}\rangle\ =\ 0
\end{array}
\end{equation}
$\forall\ {\bf c}\ \in\ {\cal S}({\bf s},{\cal I}),\ {\bf c}^{\prime}\ \in\ {\cal T}({\bf s},{\cal I})$.
\end{lemma}
\begin{proof}
 Let $\vert{\bf c}^{\prime}\rangle\ =\ \sum_{I}a_{I}\vert{\bf s}_{I}\rangle,\ \vert{\bf c}^{\prime}\rangle\ =\ \sum_{J}b_{J}\vert{\bf s}^{\prime}_{J}\rangle$, then if $\langle{\bf c}\vert{\bf c}^{\prime}\rangle\ \neq\ 0$ implies that there exists at least one ${\bf s}^{\prime}_{J}$ such that $\langle{\bf c}\vert{\bf s}_{I}^{\prime}\rangle\ \neq 0$. However this is impossible as the graph underlying ${\bf s}_{J}^{\prime}$ has $N_{v}-1$ $C^{0}$ bivalent vertices, and there exists no ${\bf s}_{I}$ in the linear sum of ${\bf c}$ which has $N_{v}-1$ $C^{0}$ bivalent vertices. Whence any cylindrical-network state belonging to ${\cal S}$ is orthogonal to any cylindrical-network state belonging to ${\cal T}({\bf s})$.
\end{proof}

\subsection{Vertex function associated to ${\cal T}({\bf s},{\cal I})$}
In this section we finally give the complete definition of habitat states by defining ${\cal F}({\bf c})\ \forall\ {\bf c}\ \in\ {\cal T}({\bf c},{\cal I})$. 
As we saw in section (\ref{habitat}), We defined ${\cal F}({\bf c})\ \forall\ {\bf c}\ \in\ {\cal S}$ as follows.
\begin{equation}
{\cal F}({\bf c})\ :=\ f(v({\bf c}))\lambda({\bf c})_{MN}
\end{equation}
Where, given any ordered pair of (colored) edges $(e_{m}, e_{n})\ \in\ (E(\gamma({\bf s})),E(\gamma({\bf s})))$ we associated to it a fixed pair $(M_{m},N_{n})\vert_{(M_{m}\neq N_{n})}\ \in\ \{(1,\dots,2j_{e_{m}}+1);(1,\dots,2j_{e_{n}}+1)\}$. We use precisely these pairs in the construction of vertex tensors for all the cylindrical networks contained in ${\cal T}({\bf s},{\cal I})$.\\
Recall that we associated to each ${\bf c}\ \in\ {\cal S}({\bf s},{\cal I})$ a 
 vertex tensor (an element of $T^{N}(su(2)^{\star})$) by examining the naive continuum limit of the finite triangulation commutator. Each of this vertex tensor is just a matrix element $(\pi_{j_{e}}(v)\otimes\pi_{j_{e^{\prime}}}(w))_{M_{e}N_{e^{\prime}}}$ with $v,w\ \in\ su(2)$. In order to be able to compare the continuum limit of LHS with that of RHS, we need to choose vertex tensors of the same type for ${\bf c}\ \in\ {\cal T}({\bf s},{\cal I})$. This requirement motivates the following definition of vertex tensors for cylindrical networks in ${\cal T}({\bf s},{\cal I})$. 
\begin{equation}\label{dec25-7}
\begin{array}{lll}
\lambda({\bf c})\ =\ \sum_{i}\left(\pi_{j_{e_{m}}}(\tau_{i}\cdot\tau_{j})\otimes\pi_{j_{e_{n}}}(\tau_{i}\cdot\tau_{j})\right)\\
\hspace*{2.9in}\textrm{if}\ {\bf c}\ \in\ [[{\bf s}]_{\delta}^{k,(m)}]^{k,(n)}\\
\hspace*{2.9in} \textrm{or if}\ {\bf c}\ \in\ ([[{\bf s}]_{\delta}^{k,(m)}]_{\delta^{\prime}}^{k,(n)})\\
\vspace*{0.2in}
\sum_{n,n\neq m}\lambda({\bf c})\ :=\ \sum_{n\vert n\neq m}\sum_{i}\left(\pi_{j_{e_{m}}}(\tau_{i}\cdot\tau_{j}\cdot\tau_{j})\otimes\pi_{j_{e_{n}}}(\tau_{i})\right)\\
\hspace*{2.9in}\textrm{if}\ {\bf c}\ \in\ [[{\bf s}]_{\delta}^{k,(m)}]^{k,(m)}\\
\vspace*{0.1in}
\hspace*{2.9in}\textrm{or if}\ {\bf c}\ \in\ ([[{\bf s}]_{\delta}^{k,(m)}]_{\delta^{\prime}}^{k,(m)})
\end{array}
\end{equation}
Careful reader can convince herself that the ``naive continuum limit" of LHS will produce combination of insertion operators which occurs in these vertex tensors.\\
We can now complete the definition of states in the Habitat space by specifying evaluation of ${\cal F}$ on ${\bf c}\ \in\ {\cal T}({\bf s},{\cal I})$.
\begin{equation}\label{dec25-5}
\begin{array}{lll}
{\cal F}({\bf c})\ =\ f(v({\bf c}))\lambda({\bf c})_{M_{m}N_{n}}\\
\hspace*{2.9in}\textrm{if}\ {\bf c}\ \in\ [[{\bf s}]_{\delta}^{k,(m)}]^{k,(n)}\\
\hspace*{2.9in} \textrm{or if}\ {\bf c}\ \in\ ([[{\bf s}]_{\delta}^{k,(m)}]_{\delta^{\prime}}^{k,(n)})\\
\vspace*{0.1in}
{\cal F}({\bf c})\ =\ f(v({\bf c}))\sum_{n,n\neq m}\lambda({\bf c})_{M_{m}N_{n}}\\
\hspace*{2.9in}\textrm{if}\ {\bf c}\ \in\ [[{\bf s}]_{\delta}^{k,(m)}]^{k,(m)}\\
\vspace*{0.1in}
\hspace*{2.9in}\textrm{or if}\ {\bf c}\ \in\ ([[{\bf s}]_{\delta}^{k,(m)}]_{\delta^{\prime}}^{k,(m)})
\end{array}
\end{equation}
where $v({\bf c})$ is the unique non-degenerate vertex associated to ${\bf c}$.

\subsection{Continuum Limit of RHS}
In this section we finally like to compute $\lim_{\delta\rightarrow 0}\lim_{\delta^{\prime}\rightarrow 0}\Psi^{{\cal F}}_{[{\bf s}]}\left(\hat{D}_{T(\delta^{\prime})}[M_{i}]\hat{D}_{T(\delta)}[N_{i}]\right)\vert{\bf s}\rangle$.\\
We believe enough has been spoken already about how such a computation is done, so without boring the reader further we straight away compute the continuum limit.
Using eq.(\ref{dec24-2}) and the fact that 
\begin{equation}
\begin{array}{lll}
\Psi^{{\cal F}}_{[{\bf s}]}\left(\vert\phi^{\delta}_{\hat{e}_{m}(0)}\cdot{\bf s}(v,e_{m},k;J_{s_{e_{m}}^{\delta}}\ =\ 0,{\cal I}^{\prime})\rangle\right) =\ 0\ \forall\ {\cal I}^{\prime}\ \neq\ {\cal I}\\
\vspace*{0.1in}
\Psi^{{\cal F}}_{[{\bf s}]}\left(\vert\phi^{\delta}_{\hat{e}_{m}(0)}\cdot{\bf s}(v,e_{m},k;J_{s_{e_{m}}^{\delta}}\ \neq\ 0)\rangle\right)\ =\ 0
\end{array}
\end{equation}
It is easy to see that the relevant terms that survive in the continuum limit will be,
\begin{equation}
\begin{array}{lll}
(-3)\sum_{k}\left[\hat{D}_{T(\delta^{\prime})}[M_{k}]\hat{D}_{T(\delta)}[N_{k}]\ -\ M\leftrightarrow N\right]\vert{\bf s}\rangle=\\
\vspace*{0.1in}
\left(\frac{1}{\frac{4\pi}{3}}\right)^{2}\frac{1}{\delta\delta^{\prime}}
\sum_{m=1}^{N_{v}}\sum_{n=1}^{N_{v}+1}\left(N(x(v))M(x(v_{e_{m},\delta}))-M\leftrightarrow N\right)\\
\left[\vert\overline{\phi}^{\delta^{\prime}}_{\hat{e}_{n}^{\delta}(0)}\left(\overline{\phi}^{\delta}_{\hat{e}_{m}(0)}\cdot{\bf s}(v,e_{m},k)\right)(v_{e_{m},\delta},e_{n}^{\delta},k)\rangle-\ 
\vert\left(\overline{\phi}^{\delta}_{\hat{e}_{m}(0)}\cdot{\bf s}(v,e_{m},k)\right)(v_{e_{m},\delta},e_{n}^{\delta},k)\rangle\right]\\
\vspace*{0.1in}
\thickapprox\ \left(\frac{1}{\frac{4\pi}{3}}\right)^{2}\frac{1}{\delta\delta^{\prime}}
\sum_{m=1}^{N_{v}}\sum_{n=1}^{N_{v}}\left(N(x(v))M(x(v_{e_{m},\delta}))-M\leftrightarrow N\right)\\
\left[\vert\overline{\phi}^{\delta^{\prime}}_{\hat{e}_{n}^{\delta}(0)}\left(\overline{\phi}^{\delta}_{\hat{e}_{m}(0)}\cdot{\bf s}(v,e_{m},k;J_{s_{e_{m}}^{\delta}}=0,{\cal I})\right)(v_{e_{m},\delta},e_{n}^{\delta},k)\rangle-\ 
\vert\left(\overline{\phi}^{\delta}_{\hat{e}_{m}(0)}\cdot{\bf s}(v,e_{m},k;J_{s_{e_{m}}^{\delta}}=0,{\cal I})\right)(v_{e_{m},\delta},e_{n}^{\delta},k)\rangle\right]\\
\vspace*{0.1in}
=\sum_{m=1}^{N_{v}}\left(N(x(v))M(x(v_{e_{m},\delta}))-M\leftrightarrow N\right)\\
\left[\vert\overline{\phi}^{\delta^{\prime}}_{\hat{e}_{m}^{\delta}(0)}\left(\overline{\phi}^{\delta}_{\hat{e}_{m}(0)}\cdot{\bf s}(v,e_{m},k;J_{s_{e_{m}}^{\delta}}=0,{\cal I})\right)(v_{e_{m},\delta},e_{m}^{\delta},k)\rangle-\ 
\vert\left(\overline{\phi}^{\delta}_{\hat{e}_{m}(0)}\cdot{\bf s}(v,e_{m},k;J_{s_{e_{m}}^{\delta}}=0,{\cal I})\right)(v_{e_{m},\delta},e_{m}^{\delta},k)\rangle\right]\\
\vspace*{0.1in}
+\sum_{m=1}^{N_{v}}\sum_{n=1\vert n\neq m}^{N_{v}}\left(N(x(v))M(x(v_{e_{m},\delta}))-M\leftrightarrow N\right)\\
\left[\vert\overline{\phi}^{\delta^{\prime}}_{\hat{e}_{n}^{\delta}(0)}\left(\overline{\phi}^{\delta}_{\hat{e}_{m}(0)}\cdot{\bf s}(v,e_{m},k;J_{s_{e_{m}}^{\delta}}=0,{\cal I})\right)(v_{e_{m},\delta},e_{n}^{\delta},k)\rangle-\ 
\vert\left(\overline{\phi}^{\delta}_{\hat{e}_{m}(0)}\cdot{\bf s}(v,e_{m},k;J_{s_{e_{m}}^{\delta}}=0,{\cal I})\right)(v_{e_{m},\delta},e_{n}^{\delta},k)\rangle\right]
\end{array}
\end{equation}
We can now use eqs. (\ref{dec25-6}),(\ref{dec25-7}) and (\ref{dec25-5}) to conclude that
\begin{equation}\label{dec25-10}
\begin{array}{lll}
(\frac{3}{4\pi^{2}})(-3)\lim_{\delta,\delta^{\prime}\rightarrow 0}\sum_{k}\Psi^{{\cal F}}_{[{\bf s}]}\left(\left[\hat{D}_{T(\delta^{\prime})}[M_{k}]\hat{D}_{T(\delta)}[N_{k}]\ -\ M\leftrightarrow N\right]\vert{\bf s}\rangle\right)=\\
\vspace*{0.1in}
\lim_{\delta\rightarrow 0}\lim_{\delta^{\prime}\rightarrow 0}\\
\sum_{m=1}^{N_{v}}\left(N(x(v))M(x(v_{e_{m},\delta}))-M\leftrightarrow N\right)\sum_{n=1\vert n\neq m}^{N_{v}}\sum_{n\vert n\neq m}\sum_{i}\left(\pi_{j_{e_{m}}}(\tau_{i}\cdot\tau_{j}\cdot\tau_{j})\otimes\pi_{j_{e_{n}}}(\tau_{i})\right)_{M_{m}N_{n}}\\
\left(f(v+\delta\hat{e}_{m}(0)+\delta^{\prime}\hat{e}_{m}^{\delta}(0))\ -\ (f(v+\delta\hat{e}_{m}(0))\right)\\
\vspace*{0.1in}
+ \sum_{m=1}^{N_{v}}\sum_{n\vert n\neq m}^{N_{v}}\left(N(x(v))M(x(v_{e_{m},\delta}))-M\leftrightarrow N\right)\left(\pi_{j_{e_{m}}}(\tau_{i}\cdot\tau_{j})\otimes\pi_{j_{e_{n}}}(\tau_{i}\cdot\tau_{j})\right)_{M_{m}N_{n}}\\
\left(f(v+\delta\hat{e}_{m}(0)+\delta^{\prime}\hat{e}_{n}^{\delta}(0))\ -\ (f(v+\delta\hat{e}_{m}(0))\right)
\end{array}
\end{equation}
Note that 
\begin{equation}\label{dec25-8}
\begin{array}{lll}
\left(\pi_{j_{e_{m}}}(\tau_{i}\cdot\tau_{j}\cdot\tau_{j})\otimes\pi_{j_{e_{n}}}(\tau_{i})\right)_{M_{m}N_{n}}\ =\ \left(\pi_{j_{e_{m}}}(\tau_{i})\otimes\pi_{j_{e_{n}}}(\tau_{i})\right)_{M_{m}N_{n}}\\
\vspace*{0.1in}
\left(\pi_{j_{e_{m}}}(\tau_{i}\cdot\tau_{j})\otimes\pi_{j_{e_{n}}}(\tau_{i}\cdot\tau_{j})\right)_{M_{m}N_{n}}\ =\\
\left(\pi_{j_{e_{m}}}(\tau_{k})\otimes\pi_{j_{e_{n}}}(\tau_{k})\right)_{M_{m}N_{n}}
+\ \left(\pi_{j_{e_{m}}}(\tau_{i}\cdot\tau_{j})\otimes\pi_{j_{e_{n}}}(\delta_{ij}{\bf 1})\right)_{M_{m}N_{n}}\\
= \left(\pi_{j_{e_{m}}}(\tau_{k})\otimes\pi_{j_{e_{n}}}(\tau_{k})\right)_{M_{m}N_{n}}
\end{array}
\end{equation}
where in the last line we have used the fact that as we are only considering non-diagonal elements as the vertex tensors (i.e. $M_{m}\ \neq\ N_{n}$), 
\begin{displaymath}
\left(\pi_{j_{e_{m}}}(\tau_{i}\cdot\tau_{j})\otimes\pi_{j_{e_{n}}}(\delta_{ij}{\bf 1})\right)_{M_{m}N_{n}}\ =\ \left(\pi_{j_{e_{m}}}({\bf 1})\otimes\pi_{j_{e_{n}}}({\bf 1})\right)_{M_{m}N_{n}}\ =\ 0
\end{displaymath}
Using eq.(\ref{dec25-8}) in eq.(\ref{dec25-10}), we get
\begin{equation}
\begin{array}{lll}
(\frac{3}{4\pi^{2}})(-3)\lim_{\delta,\delta^{\prime}\rightarrow 0}\sum_{k}\Psi^{{\cal F}}_{[{\bf s}]}\left(\left[\hat{D}_{T(\delta^{\prime})}[M_{k}]\hat{D}_{T(\delta)}[N_{k}]\ -\ M\leftrightarrow N\right]\vert{\bf s}\rangle\right)=\\
\frac{1}{4}(\frac{3}{4\pi^{2}})\sum_{m=1}^{N_{v}}\sum_{n=1\vert n\neq m}^{N_{v}}\hat{e}^{a}_{m}(0)\left(M\partial_{a}N\ -\ N\partial_{a}M\right)\left(\pi_{j_{e_{m}}}(\tau_{i})\otimes\pi_{j_{e_{n}}}(\tau_{i})\right)_{M_{m}N_{n}}\hat{e}^{b}_{m}(0)\partial_{b}f(v)+\\
\vspace*{0.1in}
\frac{1}{4}(\frac{3}{4\pi^{2}})\sum_{m=1}^{N_{v}}\sum_{n=1\vert n\neq m}^{N_{v}}\hat{e}^{a}_{m}(0)\left(M\partial_{a}N\ -\ N\partial_{a}M\right)\left(\pi_{j_{e_{m}}}(\tau_{i})\otimes\pi_{j_{e_{n}}}(\tau_{i})\right)_{M_{m}N_{n}}\hat{e}^{b}_{n}(0)\partial_{b}f(v)
\end{array}
\end{equation}
Comparing the above equation with eq. (\ref{mainresultforlhs}) we arrive at our main result : Continuum limit of $\left[((\delta^{\prime})^{2}\hat{H}_{T(\delta^{\prime})}[N])(\delta^{2}\hat{H}_{T(\delta)}[M])\ -\ M\leftrightarrow\ N\right]$ on the Habitat generated by $\Psi^{{\cal F}}_{[{\bf s}]}$ equals continuum limit of $-3\sum_{i}\left[((\delta^{\prime})^{2}\hat{D}_{T(\delta^{\prime})}[N^{i}])(\delta^{2}\hat{D}_{T(\delta)}[M^{i}])\ -\ M\leftrightarrow\ N\right]$.

\section{Discussions and Open issues}\label{conclusions}
Dynamical equations of canonical Loop Quantum Gravity are still poorly understood. Despite the remarkable initial results of Thiemann, there is no consensus on a physically viable definition of Hamiltonian constraint in LQG. There are in our opinion two main reasons for this. One being that out of infinitely many possible choices for the Hamiltonian constraint, it is not clear which (if any) is a physically viable choice. Second reason is that there are strong indications that the original definition of Hamiltonian constraint while internally consistent (a stunning feat in itself) is perhaps anomalous but the anomaly remains hidden due to its specific density weight.\\
In this work we hope to have shown that continuing along the recent directions provided by the analysis of several toy models and Diffeomorphism constraint within LQG, there is a continuum Hamiltonian constraint operator which has the right seeds to generate an anomaly free Quantum Dirac algebra.\\
More in detail we started with a classical Hamiltonian constraint of density weight two. We quantized it (at finite regularization) on the spin-network states of the theory such that it captures two essential ideas which have arose out of the recent analysis of the toy models.\\ 
\noindent {\bf (i)} its action on spin-network states mimics in a precise sense the action of classical Hamiltonian constraint on holonomy functionals.\\
\noindent{\bf (ii)} The new vertices created the constraint operator are non-degenerate in the sense that they would not be annihilated by invest volume operator.\\
Second criteria implies that the results displayed in this paper may remain robust upon use of classical constraint of density weight $\frac{4}{3}$.\\
We then showed that even though the quantum constraint so defined is too singular for the definition of a continuum limit, a rescaled (rescaled by regularization scale $\delta$) constraint operator is just singular enough so that its limit exists on a generalization of Lewandowski-Marolf Habitat space.\\
We then proved that on ${\cal V}_{hab}$, $\hat{H}[N]$ satisfies eq.(\ref{finaleqn}). This is our main result.\\

We will now discuss several unpleasant aspects of our work which we hope to address in future investigations.\\
It is not clear to us if the finite triangulation constraint operator $\hat{H}_{T(\delta)}[N]$ is (or is not) gauge-invariant.\footnote{We have been unable to convincingly verify the behavior of $\hat{H}_{T(\delta)}[N]$ under gauge transformations.} The issue is complicated by the fact that we have modified the original constraint by ``tensoring" with higher order (in $\delta$) operators. 
This is not entirely unprecedented as e.g. at finite regularization, the Hamiltonian constraint defined in the old loop representation is also not gauge invariant, as the insertion operators act at different points along the loop,\cite{brugmann}. However what is essential is that the continuum Hamiltonian constraint be gauge-invariant. We leave the detailed analysis of this question for future work. 
\linebreak
Second aspect of our work which needs further scrutiny is related to the introduction of vertex tensors. Unlike previous works where Habitat  states were used to probe the vertex structure of spin-networks. The habitat states introduced in this paper only probes the vertices of cylindrical-networks. As cylindrical networks are finite linear combination of spin-nets, they are in a certain sense obtained by coarse graining over the intertwiner data associated to spin-nets. We have tried to recover this information by associating to each cylindrical-network a vertex tensor, an element in $T(su(2))$. These vertex tensors are deduced by analyzing a heuristic short-distance limit of the finite triangulation commutator. We firmly believe that a more detailed analysis of the action of the Hamiltonian constraint, where each cylindrical network is decomposed into spin-network state will naturally bypass this need to introduce vertex tensors. The information related to intertwiner dynamics will be captured in the coefficients of the Matrix elements of the type $\Psi^{f}_{[{\bf s}]}\left(\hat{H}_{T(\delta^{\prime})}[N]\hat{H}_{T(\delta)}[M]\ -\ N\leftrightarrow M\right)\vert{\bf s}\rangle$, where $\Psi^{f}_{[{\bf s}]}$ is once again some habitat state but which will be closer to the Habitat states in \cite{tv},\cite{hlt1} in the sense that they will only need to probe the vertices of kinematical states.\\
However we certainly hope that our work here at the very least provides Pointers to a strong plausibility that Euclidean canonical LQG will be Anomaly-free and space-time covariant.\\
\section{Acknowledgements}
The author is deeply indebted to Abhay Ashtekar for permission to use his results regarding the action of the Hamiltonian constraint on flux field prior to publication, stimulating discussions regarding the status of Hamiltonian constraint in LQG, as well as for his constant encouragement. We are also extremely grateful to Miguel Campiglia for discussions, insightful remarks, and for urging the author to write up these results. We are  grateful to Kinjal Banerjee, Ghanshyam Date, Marc Geiler, Norbert Bodendorfer, Sandipan Sengupta  and Casey Tomlin for useful discussions, and to Casey Tomlin for patiently answering author's questions related to \cite{tv}. 
A special thanks to Madhavan Varadarajan for sharing his remarkable insights related to Hamiltonian constraint and dynamical issues in LQG over the years and to Soumya Paul for making all the figures. We also thank Ivan Agullo and Sumati Surya  for encouragement. We thank the Institute of gravity and cosmos at Penn state university for their hospitality where part of this work was done. The author is supported by Ramanujan Fellowship of the Department of Science and Technology.

\appendix

\section{Action of $H[N]$ on Triad fields}\label{HactiononE}
In this appendix, we derive eq.(\ref{july7(7)}) from eq.(\ref{july7(6)}) exactly following \cite{abhaycalc}.\\

\begin{equation}\label{eqjuly7(3)}
\begin{array}{lll}
2\partial_{a}(N\tilde{E}^{a}_{i} \tilde{E}^{b}_{j}\epsilon^{ijk})\ =\ 2\partial_{a}(N_{0}e^{a}_{i}\tilde{E}^{b}_{j}\epsilon^{ijk})\\
\vspace*{0.1in}
=\ 2(\partial_{a}N_{0})\ e^{a}_{jk}\tilde{E}^{b}_{j}\ +\ 2N_{0}e^{a}_{jk}\partial_{a}\sqrt{q}\ e^{b}_{j} + 2N_{0}\tilde{E}^{b}_{j}\partial_{a}e^{a}_{jk}\\
\vspace*{0.1in}
=\ 2e^{a}_{jk}\tilde{E}^{b}_{j}(\partial_{a}N_{0}) + 2N_{0}e^{a}_{jk}e^{b}_{j}\partial_{a}\sqrt{q}\\
\vspace*{0.1in}
\hspace*{0.5in}+\ \sqrt{q}\ N_{0}\left(e^{a}_{jk}\partial_{a}e^{b}_{j}\ -\ e^{a}_{j}\partial_{a}e^{b}_{jk}\right) + 2 N_{0}\tilde{E}^{b}_{j}\partial_{a}e^{a}_{jk}
\end{array}
\end{equation}
In the above equation $N_{0}\ =\ N\sqrt{q}$, $\tilde{E}^{a}_{i}\ =\ \sqrt{q}e^{a}_{i}$ and in the last line we have used 
\begin{equation}
\begin{array}{lll}
2e^{a}_{jk}\partial_{a}e^{b}_{j}\ =\ 2\epsilon^{ijk}e^{a}_{i}\partial_{a}e^{b}_{j}\ =\\
\vspace*{0.1in}
\epsilon^{ijk}\left(e^{a}_{i}\partial_{a}e^{b}_{j}\ -\ e^{a}_{j}\partial_{a}e^{b}_{i}\right)\ =\ \left(e^{a}_{jk}\partial_{a}e^{b}_{j} - e^{a}_{j}\partial_{a}e^{b}_{jk}\right)
\end{array}
\end{equation}
We can further massage the RHS of (\ref{eqjuly7(3)}) as follows.
\begin{equation}\label{eqjuly7(5)}
\begin{array}{lll}
2\partial_{a}(N\tilde{E}^{a}_{i} \tilde{E}^{b}_{j}\epsilon^{ijk})\ =\\
\vspace*{0.1in}
2e^{a}_{jk}\tilde{E}^{b}_{j}(\partial_{a}N_{0}) + 2N_{0}e^{a}_{jk}e^{b}_{j}\partial_{a}\sqrt{q}\\
\vspace*{0.1in}
\hspace*{0.6in} + \left(\sqrt{q}N_{0}e^{a}_{jk}\partial_{a}e^{b}_{j} - e^{a}_{j}\partial_{a}N_{0}e^{b}_{jk}\right) + e^{b}_{jk}\tilde{E}^{a}_{j}\partial_{a}N_{0} + 2N_{0}\tilde{E}^{b}_{j}\partial_{a}e^{a}_{jk}\\
\vspace*{0.1in}
=\ e^{a}_{jk}\tilde{E}^{b}_{j}(\partial_{a}N_{0})\ +\ 2N_{0}e^{a}_{jk}e^{b}_{j}\partial_{a}\sqrt{q}\\
 \vspace*{0.1in}
\hspace*{0.6in} + {\cal L}_{N_{0}e_{jk}}e^{b}_{j} + N_{0}\tilde{E}^{b}_{j}\partial_{a}e^{a}_{jk}\ +\ N_{0}\sqrt{q}e^{b}_{j}\partial_{a}e^{a}_{jk}\\
= e^{b}_{j}{\cal L}_{N_{0}e_{jk}}\sqrt{q}\ +\ \sqrt{q}{\cal L}_{Ne_{jk}}e^{b}_{j}\ +\ N_{0}e^{b}_{j}\partial_{a}\tilde{E}^{a}_{jk}\\
\vspace*{0.1in}
=\ {\cal L}_{NE_{jk}}\tilde{E}^{b}_{j}\ +\ N_{0}e^{b}_{j}\partial_{a}\tilde{E}^{a}_{jk}
\end{array}
\end{equation}

Whence we finally have,
\begin{equation}
\begin{array}{lll}
\{H[N], \tilde{E}^{b}_{k}(x)\} =\ 2\partial_{a}(N\tilde{E}^{a}_{i} \tilde{E}^{b}_{j}\epsilon^{ijk})\ +\ 2\epsilon^{kmn}\ A_{a}^{m}\ \omega^{n}\\
\vspace*{0.1in}
\hspace*{0.5in}=\ {\cal L}_{NE_{jk}}\tilde{E}^{b}_{j}\ +\ N_{0}e^{b}_{j}\partial_{a}\tilde{E}^{a}_{jk} + 2 \epsilon^{kmn}\ A_{a}^{m}\ \epsilon^{ijn}N\tilde{E}^{a}_{i}\tilde{E}^{b}_{j}\\
\vspace*{0.1in}
\hspace*{0.5in}=\ {\cal L}_{NE_{jk}}\tilde{E}^{b}_{j}\ +\ N_{0}e^{b}_{j}{\cal D}_{a}\tilde{E}^{a}_{jk} +  \epsilon^{kmn}\ A_{a}^{m}\ \epsilon^{ijn}N\tilde{E}^{a}_{i}\tilde{E}^{b}_{j}
\end{array}
\end{equation}
Where, in the last line we have combined one $\epsilon^{kmn}\ A_{a}^{m}\ \epsilon^{ijn}N\tilde{E}^{a}_{i}\tilde{E}^{b}_{j}$ with $N_{0}e^{b}_{j}{\cal D}_{a}\tilde{E}^{a}_{jk}$ to get $N_{0}e^{b}_{j}{\cal D}_{a}\tilde{E}^{a}_{jk}$. This is eq.(\ref{july7(7)}).

\section{Commutator of classical vector fields}\label{commutator}
In this section we compute the commutator of $[\ X_{H[N]},\ X_{H[M]}\ ]$. This exercise is of interest in its own right as it sheds light on how our use of the two triads in $H$ on unequal footing (one being used as a functional derivative, whereas other is used to define a Shift field) reproduces the vector field, which can be thought of as a vector field associated to diffeomorphism constraint when the Shift field is dynamical.\\
We start with 
\begin{equation}
\begin{array}{lll}
X_{H[N]}\ =\int d^{3}x\ \textrm{Tr}\left([\tau^{k},L_{NE^{k}}A_{a}]\cdot\frac{\delta}{\delta A_{a}}\right)(x)\\
\vspace*{0.1in}
=\int d^{3}x\ f_{a}^{IJ}[N,A,E;x)\frac{\delta}{\delta A_{a}^{IJ}(x)}
\end{array}
\end{equation}
where $f_{a}^{IJ}[N,A,E;x)\ =\ [\tau^{k},L_{NE^{k}}A_{a}]^{IJ}$.\\
Thus
\begin{equation}
\begin{array}{lll}
[\ X_{H[N]}, X_{H[M]}\ ]\ =\ \int d^{3}x\ d^{3}y\left[f_{a}^{IJ}[N;x)\frac{\delta}{\delta A_{a}^{IJ}(x)}f_{b}^{KL}[M;y)\ -\ N\leftrightarrow M\right]\frac{\delta}{\delta A_{b}^{KL}(y)}
\end{array}
\end{equation}
A simple exercise shows that 
\begin{equation}
\begin{array}{lll}
\frac{\delta}{\delta A_{a}^{IJ}}f_{b}^{KL}[M;y)\ =\\
\vspace*{0.1in}
(\tau^{k})^{K}_{K_{1}}\left[M E^{c}_{k}(y)\delta^{a}_{b}\delta^{K_{1}L}_{IJ}\partial_{c}^{y}\delta^{3}(y,x)\ +\ \delta^{a}_{c}\delta^{K_{1}L}_{IJ}\delta^{3}(y,x)\partial_{b}^{y}\left(ME^{c}_{k}(y)\right)\right]\\
-(\tau^{k})^{L}_{K_{1}}\left[M E^{c}_{k}(y)\delta^{a}_{b}\delta^{KK_{1}}_{IJ}\partial_{c}^{y}\delta^{3}(y,x)\ +\ \delta^{a}_{c}\delta^{K K_{1}}_{IJ}\delta^{3}(y,x)\partial_{b}^{y}\left(ME^{c}_{k}(y)\right)\right]
\end{array}
\end{equation}
where $\partial_{a}^{y}\ :=\ \frac{\partial}{\partial y^{b}}$ and $\delta^{IJ}_{KL}\ =\ \delta^{I}_{K}\delta^{J}_{L}$.
We can use the above equation to compute\\ $\int d^{3}xd^{3}y f_{a}^{IJ}[N;x)\frac{\delta}{\delta A_{a}^{IJ}(x)}f_{b}^{KL}[M;y)$ and after some algebra find that
\begin{equation}
\begin{array}{lll}
\int d^{3}xd^{3}y f_{a}^{IJ}[N;x)\frac{\delta}{\delta A_{a}^{IJ}(x)}f_{b}^{KL}[M;y)\ =\\
\int d^{3}x\left[(\tau^{k}\cdot\tau^{m}\cdot L_{ME_{k}}L_{NE_{m}}A_{b})^{KL}\ +\ (L_{ME_{k}}L_{NE_{m}}A_{b}\cdot\tau_{m}\cdot\tau_{k})^{KL}\right]\\
\vspace*{0.1in}
-\int d^{3}x\left[(\tau_{k}\cdot\cdot L_{ME_{k}}L_{NE_{m}}A_{b}\cdot\tau_{m})^{KL}\ +\ (\tau_{m}\cdot L_{ME_{k}}L_{NE_{m}}A_{b}\cdot\tau_{k})^{KL}\right]
\end{array}
\end{equation}
Whence finally we have 
\begin{equation}
\begin{array}{lll}
[X_{H[N]}, X_{H[M]}]\ =\\
\int d^{3}x\left[(\tau^{k}\cdot\tau^{m}\cdot L_{[ME_{k},ME_{m}]}A_{b})^{KL}\ +\ (L_{[ME_{k},NE_{m}]}A_{b}\cdot\tau_{m}\cdot\tau_{k})^{KL}\right]\\
\vspace*{0.1in}
-\int d^{3}x\left[(\tau_{k}\cdot\cdot L_{[ME_{k},NE_{m}]}A_{b}\cdot\tau_{m})^{KL}\ +\ (\tau_{m}\cdot L_{[ME_{k},NE_{m}]}A_{b}\cdot\tau_{k})^{KL}\right]\\
=\int d^{3}x\left[\left(\tau_{k}\cdot[\tau_{m},L_{[ME_{k},NE_{m}]}A_{b}]\right)^{KL}\ +\ \left([\tau_{m},L_{[ME_{k},NE_{m}]}A_{b}]\cdot\tau_{k}\right)^{KL}\right]\frac{\delta}{\delta A_{b}^{KL}}\\
=\int d^{3}x \{\tau_{k},[\tau_{m},L_{[ME_{k},NE_{m}]}A_{b}]\}^{KL}\frac{\delta}{\delta A_{b}^{KL}}
\end{array}
\end{equation}
where $\{A,B\}$ stands for anti-commutator between $su(2)$ elements.\\
A simple gymnastics with Pauli matrices leads to
\begin{equation}
\begin{array}{lll}
[X_{H[N]}, X_{H[M]}]\ =\ -2\sum_{k}\int d^{3}x \textrm{Tr}\left(L_{[ME_{k},NE_{k}]}A_{b}\cdot\frac{\delta}{\delta A_{b}}\right)
\end{array}
\end{equation}

On the other hand, the right hand side of $\{H[N], H[M]\}$ is given by,
\begin{equation}\label{eq:dec25-1}
\begin{array}{lll}
\textrm{RHS}\ =\ \int d^{3}x\ [N\partial_{a}M\ -\ M\partial_{a}N]\ E^{a}_{m}E^{b}_{m}\ H_{b}(x)\\
\vspace*{0.1in}
=\ \int d^{3}x\ [NE^{m}, ME^{m}]^{a}(x)\ H_{a}(x)\\
\vspace*{0.1in}
=\ \sum_{m}\int d^{3}x L_{[NE^{m}, ME^{m}]}A_{a}^{i}(x)\ \cdot\ E^{a}_{i}(x)
\end{array}
\end{equation}
where we have used the classical identity,
\begin{equation}
\begin{array}{lll}
H_{diff}[U]\ =\ \int U^{a}F_{ab}^{i}E^{b}_{i}\ =\ \int L_{U}A_{a}^{i}\cdot E^{a}_{i}\\
\implies\ X_{H_{diff}[U]}\ =\ -\frac{1}{3}\int\textrm{Tr}\left(L_{U}A_{a}\cdot\frac{\delta}{\delta A_{a}}\right)
\end{array}
\end{equation}
We now notice that applying the same logic, that we applied to $H[N]$ in associating a vector field to it, that is, interpreting $[NE^{m}, ME^{m}]$ in (\ref{eq:dec25-1}) as a shift field and $E^{a}_{i}$ as functional derivative, we get 
\begin{equation}
\begin{array}{lll}
X_{\textrm{RHS}}\ =\ \sum_{m}\int d^{3}x \textrm{Tr}\left(L_{[NE^{m}, ME^{m}]}A_{a}(x)\cdot \frac{\delta}{\delta A_{a}}\right)\ =\ \frac{3}{2}[X_{H[N]}, X_{H[M]}]
\end{array}
\end{equation}

\end{document}

@article{improvedlqc,
  title = {Quantum nature of the big bang: Improved dynamics},
  month = {Oct},
  doi = {10.1103/PhysRevD.74.084003},
  author = {Ashtekar, Abhay and Pawlowski, Tomasz and Singh, Parampreet},
  year = {2006},
  issue = {8},
  url = {http://link.aps.org/doi/10.1103/PhysRevD.74.084003},
  numpages = {23},
  journal = {Phys. Rev. D},
  publisher = {American Physical Society},
  pages = {084003},
  volume = {74}
}

@article{madtom,
  title = {Towards an Anomaly-Free Quantum Dynamics for a Weak-Coupling Limit of Euclidean Gravity},
  author = {Tomlin, Casey and Varadarajan, Madhavan},
  year = {2012},
  url = {},
  journal = {},
  publisher = {},
}

@article{lam^2t,
author = {Abhay Ashtekar and Jerzy Lewandowski and Donald Marolf and Jos\'{e} Mour\~{a}o and Thomas Thiemann},
collaboration = {},
title = {Quantization of diffeomorphism invariant theories of connections with local degrees of freedom},
publisher = {AIP},
year = {1995},
journal = {Journal of Mathematical Physics},
volume = {36},
number = {11},
pages = {6456-6493},
keywords = {DEGREES OF FREEDOM; FIELD THEORIES; GENERAL RELATIVITY THEORY; GRAVITATIONAL FIELDS; LIE GROUPS; QUANTIZATION; QUANTUM GRAVITY},
url = {http://link.aip.org/link/?JMP/36/6456/1},
doi = {10.1063/1.531252}
}

@article{volconsistency,
  author={K Giesel and T Thiemann},
  title={Consistency check on volume and triad operator quantization in loop quantum gravity: I},
  journal={Classical and Quantum Gravity},
  volume={23},
  number={18},
  pages={5667},
  url={http://stacks.iop.org/0264-9381/23/i=18/a=011},
  year={2006}}

@article{qsd4,
  author={Thomas Thiemann},
  title={Quantum Spin Dynamics ({QSD}): {IV}. 2+1 {E}uclidean quantum gravity as a model to test 3+1 {L}orentzian quantum gravity},
  journal={Classical and Quantum Gravity},
  volume={15},
  number={5},
  pages={1249},
  url={http://stacks.iop.org/0264-9381/15/i=5/a=011},
  year={1998}
}

@article{diffcons,
  author={Alok Laddha and Madhavan Varadarajan},
  title={The diffeomorphism constraint operator in loop quantum gravity},
  journal={Classical and Quantum Gravity},
  volume={28},
  number={19},
  pages={195010},
  url={http://stacks.iop.org/0264-9381/28/i=19/a=195010},
  year={2011}
}

@article{nicolai,
  author={Hermann Nicolai and Kasper Peeters and Marija Zamaklar},
  title={Loop quantum gravity: an outside view},
  journal={Classical and Quantum Gravity},
  volume={22},
  number={19},
  pages={R193},
  url={http://stacks.iop.org/0264-9381/22/i=19/a=R01},
  year={2005}}
	
@article{qsd2,
  author={T Thiemann},
  title={Quantum spin dynamics (QSD): II. The kernel of the Wheeler - DeWitt constraint operator},
  journal={Classical and Quantum Gravity},
  volume={15},
  number={4},
  pages={875},
  url={http://stacks.iop.org/0264-9381/15/i=4/a=012},
  year={1998},
  abstract={We determine the complete and rigorous kernel of the Wheeler - DeWitt constraint operator for four-dimensional, Lorentzian, non-perturbative, canonical vacuum quantum gravity in the continuum. We do this for the non-symmetric version of the operator constructed previously in this series. We also construct a symmetric, regulated constraint operator. For the regulated Euclidean Wheeler - DeWitt operator as well as for the regulated generator of the Wick transform from the Euclidean to the Lorentzian regime we prove existence of self-adjoint extensions and based on these we propose a method of proof of self-adjoint extensions for the regulated Lorentzian operator. Both constraint operators evaluated at unit lapse as well as the generator of the Wick transform can be shown to have regulator-independent and symmetric duals on the diffeomorphism-invariant Hilbert space. Finally, we comment on the status of the Wick rotation transform in the light of the present results and give an intuitive description of the action of the Hamiltonian constraint.}
}

@article{qsd3,
  author={T Thiemann},
  title={Quantum spin dynamics (QSD): III. Quantum constraint algebra and physical scalar product in quantum general relativity},
  journal={Classical and Quantum Gravity},
  volume={15},
  number={5},
  pages={1207},
  url={http://stacks.iop.org/0264-9381/15/i=5/a=010},
  year={1998},
  abstract={This paper deals with several technical issues of non-perturbative four-dimensional Lorentzian canonical quantum gravity in the continuum that arose in connection with the recently constructed Wheeler-DeWitt quantum constraint operator. * The Wheeler-DeWitt constraint of quantum general relativity mixes the diffeomorphism superselection sectors for diffeomorphism-invariant theories of connections that were previously discussed in the literature. From it one can construct diffeomorphism-invariant operators which do not necessarily commute with the Hamiltonian constraint but which still mix those sectors and which, at the diffeomorphism-invariant level, encode physical information. Thus, if one adopts, as before in the literature, the strategy to solve the diffeomorphism constraint before the Hamiltonian constraint then those sectors become spurious. * The inner product for diffeomorphism-invariant states can be fixed by requiring that diffeomorphism group averaging is a partial isometry. * The established non-anomalous constraint algebra is clarified by computing commutators of duals of constraint operators. * The full classical constraint algebra is faithfully implemented on the diffeomorphism-invariant Hilbert space in an appropriate sense. * The Hilbert space of diffeomorphism-invariant states can be made separable if a natural new superselection principle is satisfied. * We propose a natural physical scalar product for quantum general relativity by extending the group-average approach to the case of non-self-adjoint constraint operators like the Wheeler-DeWitt constraint. * Equipped with this inner product, the construction of physical observables is straightforward.}
}

@article{aajurekreview,
  author={Abhay Ashtekar and Jerzy Lewandowski},
  title={Background independent quantum gravity: a status report},
  journal={Classical and Quantum Gravity},
  volume={21},
  number={15},
  pages={R53},
  url={http://stacks.iop.org/0264-9381/21/i=15/a=R01},
  year={2004}
}

@article{qsd1,
  author={Thomas Thiemann},
  title={Quantum spin dynamics ({QSD})},
  journal={Classical and Quantum Gravity},
  volume={15},
  number={4},
  pages={839},
  url={http://stacks.iop.org/0264-9381/15/i=4/a=011},
  year={1998}
}

@ARTICLE{lm1,
  author = {Jerzy Lewandowski and Donald Marolf},
  title = {Loop constraints: A habitat and their algebra},
  journal = {International Journal of Modern Physics D},
  volume = {7},
  pages = {299},
  url = {http://www.citebase.org/abstract?id=oai:arXiv.org:gr-qc/9710016},
  year = {1998}
}

@ARTICLE{lm2,
  author = {Rodolfo Gambini and Jerzy Lewandowski and Donald Marolf and Jorge Pullin},
  title = {On the Consistency of the Constraint Algebra in Spin Network Quantum Gravity},
  journal = {International Journal of Modern Physics D},
  volume = {7},
  pages = {97},
  url = {http://dx.doi.org/10.1142/S0218271898000103},
  year = {1998}
}

@article{pftham,
  volume = {83},
  journal = {Phys. Rev. D},
  author = {Laddha, Alok and Varadarajan, Madhavan},
  month = {Jan},
  url = {http://link.aps.org/doi/10.1103/PhysRevD.83.025019},
  doi = {10.1103/PhysRevD.83.025019},
  year = {2011},
  title = {Hamiltonian constraint in polymer parametrized field theory},
  issue = {2},
  publisher = {American Physical Society},
  numpages = {27},
  pages = {025019}
}

@article{smolin,
  author={Lee Smolin},
  title={The ${G}_{\mathrm{Newton}}\rightarrow 0$ limit of {E}uclidean quantum gravity},
  journal={Classical and Quantum Gravity},
  volume={9},
  number={4},
  pages={883},
  url={http://stacks.iop.org/0264-9381/9/i=4/a=007},
  year={1992}
}

@article{hat,
  author = {Henderson, Adam and Laddha, Alok and Tomlin, Casey},
  year = {2012},
  title = {Constraint Algebra in {LQG} Reloaded: Toy Model of a {U}(1)$^3$ Gauge Theory},
  eprint = {arXiv:1204.0211v1}
}

@article{mefereduardo,
  author = {Barbero, Fernando and Varadarajan, Madhavan and Villase\~{n}or, Eduardo},
  year = {2012},
  title = {in preparation},
}

@article{meinprep,
  author = {Varadarajan, Madhavan},
  year = {2012},
  title = {in preparation},
}

@article{fereduardo,
  author = {Barbero, Fernando and Villase\~{n}or, Eduardo},
  year = {2012},
  title = {in preparation},
}

@Book{ttbook,
author = {Thiemann, Thomas},
title = {Modern Canonical Quantum General Relativity},
publisher = {Cambridge University Press},
year = {2007},
series = {Cambridge Monographs on Mathematical Physics}
}

@Book{loopcoord,
author = {Gambini, Rodolfo and Pullin, Jorge},
title = {Loops, Knots, Gauge Theories, and Quantum Gravity},
publisher = {Cambridge University Press},
year = {2000},
series = {Cambridge Monographs on Mathematical Physics}
}

@article{aajurekarea,
  author={Abhay Ashtekar and Jerzy Lewandowski},
  title={Quantum theory of geometry: {I}. {A}rea operators},
  journal={Classical and Quantum Gravity},
  volume={14},
  number={1A},
  pages={A55},
  url={http://stacks.iop.org/0264-9381/14/i=1A/a=006},
  year={1997}
}

@article{RG,
author = {Norbert Grot and Carlo Rovelli},
title = {Moduli-space structure of knots with intersections},
publisher = {AIP},
year = {1996},
journal = {Journal of Mathematical Physics},
volume = {37},
number = {6},
pages = {3014-3021},
keywords = {QUANTUM GRAVITY; TOPOLOGY; MATHEMATICAL SPACE; LOOPS; MATHEMATICAL PHYSICS},
url = {http://link.aip.org/link/?JMP/37/3014/1},
doi = {10.1063/1.531527}
}

@ARTICLE{aajurekvol,
  author = {Abhay Ashtekar and Jerzy Lewandowski},
  title = {Quantum Theory of Geometry {II}: {V}olume operators},
  journal = {Advances in Theoretical and Mathematical Physics},
  volume = {1},
  pages = {388},
  url = {http://www.citebase.org/abstract?id=oai:arXiv.org:gr-qc/9711031},
  year = {1998}
}

@ARTICLE{hkt,
   author = {{Hojman}, S.~A. and {Kucha{\v r}}, K. and {Teitelboim}, C.},
    title = "{Geometrodynamics regained}",
  journal = {Annals of Physics},
     year = 1976,
    month = jan,
   volume = 96,
    pages = {88-135},
      doi = {10.1016/0003-4916(76)90112-3},
   adsurl = {http://adsabs.harvard.edu/abs/1976AnPhy..96...88H},
  adsnote = {Provided by the SAO/NASA Astrophysics Data System}
}

@ARTICLE{mad2,
   author = {madhavan Varadarajan},
    title = "{Towards an Anomaly-free Quantum Dynamics for a Weak Coupling Limit of Euclidean Gravity : Diffeomorphism Covariance }",
  journal = {In preparation},
     year = 2012,
    
}

@ARTICLE{abhay,
   author = {Abhay Ashtekar},
    title = "{Geometric interpretation of Hamiltonian constraint }",
  journal = {Unpublished notes},
     year = 2012,
    
}

\end{document}

We now show how this action on a spin-net is a precise correspondence with the geometric action derived above. 
Writing down the spin-network state in connection representation. Let $v$ be a vertex of the graph underlying ${\bf s}$ then,
\begin{equation}
\begin{array}{lll}
T_{s}(A)\ =\ {\cal I}_{v}\cdot\left(\otimes \pi_{j_{e}}h_{e}\right)_{m_{e_{1}}\textrm{. . .}m_{e_{N}} , n_{e_{1}}\textrm{. . .}n_{e_{N}}}\ \bigotimes\ T_{s-\{v\}}(A)
\end{array}
\end{equation}
Given two edges, $e_{1}, e_{2}$ such that $e_{1}\cap e_{2}\ =\ v$, 
it is straightforward to evaluate the action of say, $\hat{H}_{e_{1},e_{2},\delta}(v)$ on 
$T_{s}(A)$.
\begin{equation}
\begin{array}{lll}
\hat{H}_{e_{1}, e_{2}, \delta}(v)T_{s}(A)\ =\  [\ \widehat{\pi_{j_{e_{1}}}(h_{\alpha(e_{1},e_{2},\delta)}}, \widehat{\pi_{j_{e_{1}}}(\tau_{j})}\ ]\otimes \widehat{\pi_{j_{e_{2}}}(\tau_{j})}T_{s}(A)\\
\vspace*{0.1in}
={\cal I}_{v}\cdot\left(\pi_{j_{e_{1}}}\left(h_{\alpha(e_{1},e_{2},\delta)}}\circ\tau_{j}\circ h_{e_{1}}\right)_{m_{e_{1}}, n_{e_{1}}}\ \otimes\ \pi_{j_{e_{2}}}(\tau_{j}\circ h_{e_{2}})_{m_{e_{2}}, n_{e_{2}}}\ \left(\bigotimes_{e\neq e_{1}, e_{2}}\pi_{j_{e}}(h_{e})\right)_{m_{e_{3}}\textrm{É}m_{e_{N}},n_{e_{3}}\textrm{É}n_{e_{N}}}\right)\\
\vspace*{0.1in}
-  {\cal I}_{v}\cdot\left(\pi_{j_{e_{1}}}\left(\tau_{j}\circ h_{\alpha(e_{1},e_{2},\delta)}\circ h_{e_{1}}\right)_{m_{e_{1}}, n_{e_{1}}}\ \otimes\ \pi_{j_{e_{2}}}(\tau_{j}\circ h_{e_{2}})_{m_{e_{2}}, n_{e_{2}}}\ \left(\bigotimes_{e\neq e_{1}, e_{2}}\pi_{j_{e}}(h_{e})\right)_{m_{e_{3}}\textrm{É}m_{e_{N}},n_{e_{3}}\textrm{É}n_{e_{N}}}\right)\\
\end{array}
\end{equation}
We now show how the above operator action faithfully captures the classical action as detailed in (\ref{eq:regconstraintfinal}).\\

\section{Subtlties with inverse volume factor}
So far in this paper, we quantized the Hamiltonian constraint with lapse being a scalar density of weight $-1$. However as we briefly recalled in the introduction, the correct density weighted lapse, so that the resulting quantum constraint admits a continuum limit is a Lapse with weight $\frac{1}{3}$. We now show how the constraint obtained above can be ``amended" so that it can be smeared with respect to such a Lapse.\\
Thus in the classical theory, such a constraint is given by,
\begin{equation}
\begin{array}{lll}
H[N]\ =\ \int_{\Sigma} N(x)\epsilon^{ijk}F_{ab}^{i}E^{a}_{j}E^{b}_{k}(x)\frac{1}{q^{\frac{1}{3}}}(x)
\end{array}
\end{equation}
As shown in \cite{tv}, a minor modification of the Thiemann-trick can be used to reexpress the inverse volume factor $\frac{1}{q^{\frac{1}{3}}}$ in terms of Poisson brackets. This in turn implies that, $\frac{1}{q^{\frac{1}{3}}}$ can be ``quantized separately" independent of the quantization of $\epsilon^{ijk}F_{ab}^{i}E^{a}_{j}E^{b}_{k}(x)$. Hence the structure of the constraint we proposed in the previous section remains robust in the sense that it needs to be amended by $\widehat{\frac{1}{q^{\frac{1}{3}}}}$. Once again our choice of operator ordering is dictated by the dictum proposed in $\cite{qsd1}$, that is we do not want the Hamiltonian constraint to create new vertices away from the graph. The simplest choice (but may not be the only one) of operator ordering which accomplishes this, is by placing $\frac{1}{q^{\frac{1}{3}}}$ to the right of $\widehat{\epsilon^{ijk}F_{ab}^{i}E^{a}_{j}E^{b}_{k}}$.\\
The finite triangulation Hamiltonian constraint action on a spin-network state $\vert {\bf s}\rangle$ with a single non-degenerate vertex $v$  is given by,
\begin{equation}
\begin{array}{lll}
\hat{H}_{T(\delta)}[N]\vert{\bf s}\rangle\ =\\
\vspace*{0.1in}
\frac{1}{\delta}N(v)\widehat{\epsilon^{ijk}F_{ab}^{i}E^{a}_{j}E^{b}_{k}}\vert_{\delta}(v)\widehat{\frac{1}{q^{\frac{1}{3}}}}\vert_{\delta}(v)\vert{\bf s}\rangle\\
\vspace*{0.1in}
=\frac{1}{\delta}\sum_{e,e^{\prime}\vert e\neq e^{\prime}; b(e)=b(e^{\prime})}\hat{\tau}^{i}\vert_{b(e)}\otimes[\ \hat{\tau^{i}}\vert_{b(e^{\prime})}, \widehat{\left(h_{s_{e}}\circ h_{\phi^{\delta}_{\hat{e}(0)}(s_{e^{\prime}})}\circ h_{s_{e^{\prime}}}^{-1}\right)}\ ]\\
\hspace*{1.0in}\bigotimes\otimes_{e^{\prime\prime}\notin\{e,e^{\prime}\}}\widehat{\left(h_{s_{\tilde{e}}}\circ h_{\phi^{\delta}_{\hat{e}(0)}(s_{\tilde{e}})}\circ h_{s_{\tilde{e}}}^{-1}\right)}\hat{\frac{1}{q^{\frac{1}{3}}}}\vert_{\delta} (v)\vert{\bf s}\rangle
\end{array}
\end{equation}
As shown above, the most general inverse volume operator action consistent with the properties of matrix elements of volume operator is given by,
\begin{equation}
\begin{array}{lll}
\widehat{\frac{1}{q}^{\frac{1}{3}}}(v)\vert_{\delta}\vert{\bf s}\rangle\ =\

Each of the above states are linear combination of spin-network states based on a graph which has \emph{at most} a degenerate vertex at $v$ and a unique non-degenerate vertex at $v_{e,\delta}$. The grasping operators $\tau_{i}$ can be expressed as (appropriately normalized) 3-j symbols, and can be used to re-write the above states in terms of spin-network basis.\\
We use certain structural properties of the resulting intertwiners to illustrate some key aspects of the states.\\

\section{Analyzing the structure of $T_{({\bf s})_{2}^{\delta}(v,e,e^{\prime})}(A)$.}
The second equation in eq.(\ref{deformedstates-1}) has a more intricate structure as the two grasping operators do not act in a way as to yield a $J=0$ intertwiner.
Using the fact that ${\cal I}_{v}^{I}\cdot\hat{\tau}_{i}\vert_{e}$ is a (sum of) $J=1$ intertwiners, and for a spin-1 intertwiner, we have
\begin{equation}\label{j=1intertwiner}
\begin{array}{lll}
\left(\pi_{j_{e_{1}}}(g)\otimes\dots\otimes\pi_{j_{e_{N}}}(g)\right)_{m_{e_{1}},\dots,m_{e_{N}},n_{e_{1}},\dots,n_{e_{N}}}\cdot {\cal I}^{m_{e_{1}},\dots,m_{e_{N}}m}\ =\ {\cal I}^{n_{e_{1}},\dots,n_{e_{N}}n}\pi_{J=1}(g)_{nm}
\end{array}
\end{equation}
where the indices $n,m\in \{-1,0,1\}$ and $m$ index on the intertwiner ${\cal I}$ means it is a spin-1 intertwiner.\\
Above equation in conjunction with
\begin{equation}
\begin{array}{lll}
\tau_{i}\otimes\tau^{i}\ =\ \sum_{m=-1}^{1}(-1)^{1-m}L_{m}\otimes L_{-m}
\end{array}
\end{equation}
implies that
\begin{equation}
\begin{array}{lll}
T_{({\bf s})_{2}^{\delta}(v,e,e^{\prime})}(A)\ =\\
\vspace*{0.1in}
\sum_{k}(-1)^{k}\left(\pi_{j=1}(h_{s_{e_{1}}})_{kl}\otimes\pi_{j_{e^{\prime}}}(h_{s_{e^{\prime}}}^{-1}\cdot L_{-n}\cdot h_{s_{e^{\prime}}})_{\tilde{m_{e^{\prime}}}\tilde{n_{e^{\prime}}}}\right)\otimes\\
\hspace*{0.7in}{\cal I}_{v}^{I\ m_{e_{1}}\dots m_{e}^{\prime}\dots m_{e_{n_{v}}}}(L_{l})_{m_{e}^{\prime}m_{e}}\left(\pi_{j_{e_{1}}}(h_{\alpha(s_{e}, s_{e_{1}})}\cdot h_{e_{1}-s_{e_{1}}})_{m_{e_{1}}n_{e_{1}}}\otimes\dots\pi_{j_{e}}(h_{e-s_{e}})_{m_{e}n_{e}}\otimes\dots\otimes\right.\\
\hspace*{1.0in}\left.\pi_{j_{e^{\prime}}}(h_{\alpha(s_{e},s_{e^{\prime}})})_{m_{e^{\prime}}\tilde{m_{e^{\prime}}}}\otimes\pi_{j_{e^{\prime}}}(h_{e^{\prime}-s_{e^{\prime}}})_{\tilde{n_{e^{\prime}}}n_{e^{\prime}}}\otimes\dots\otimes\pi_{j_{e_{n_{v}}}}(h_{\alpha(s_{e},s_{e_{n_{v}}})}\cdot h_{e_{n_{v}}-s_{e_{n_{v}}}})_{m_{e_{n_{v}}},n_{e_{n_{v}}}}\right)
\end{array}
\end{equation}
In the above equation, we have re-arranged the terms such that the first line of the RHS corresponds to part of the state which has a (bivalent) vertex at $v$. The second and the third line correspond to the structure in the neighborhood of $v + \delta\hat{e}(0)$.\\
The above equation can be expanded further by expressing $h_{s_{e^{\prime}}}^{-1}\otimes h_{s_{e^{\prime}}}$ in terms of linear combination of $h_{s_{e}}$ is distinct irreps.\\
\begin{equation}
\begin{array}{lll}
\pi_{j_{e^{\prime}}}(h_{s_{e^{\prime}}}^{-1})_{\tilde{m_{e^{\prime}}},k}\otimes \pi_{j_{e^{\prime}}}(h_{s_{e^{\prime}}})_{l\tilde{n_{e^{\prime}}}}\ =\ 
\sum_{J=0}^{2j_{e^{\prime}}}\pi_{J}(h_{s_{e^{\prime}}})_{M_{e^{\prime}}N_{e^{\prime}}}
\end{array}
\end{equation}

({\bf CHECK FOR GAUGE INVARIANCE}).

\end{document}